\ifdefined\pdfoutput\pdfoutput=1\fi
\documentclass[11pt]{article}

\ifdefined\pdfmapfile
  \pdfmapfile{+lm.map}
  \pdfmapfile{+symbols.map}
\fi

\usepackage{flafter}
\usepackage[T1]{fontenc}
\usepackage{lmodern}
\usepackage[a4paper,margin=25mm]{geometry}
\usepackage{amsmath,amssymb,amsthm,mathtools}
\usepackage{booktabs,tabularx,array}
\usepackage{xcolor}
\usepackage{tikz}
\usepackage{float}
\usepackage{placeins}
\usetikzlibrary{arrows.meta,backgrounds,calc,decorations.pathreplacing,shapes.geometric}
\usepackage[colorlinks=true,linkcolor=blue,citecolor=blue,urlcolor=blue]{hyperref}
\numberwithin{equation}{section}

\newtheorem{theorem}{Theorem}[section]
\newtheorem{lemma}[theorem]{Lemma}
\newtheorem{proposition}[theorem]{Proposition}
\newtheorem{corollary}[theorem]{Corollary}
\theoremstyle{remark}

\newcommand{\Tr}{\operatorname{Tr}}
\newcommand{\OSR}{\operatorname{OSR}}
\newcommand{\SR}{\operatorname{SR}}
\newcommand{\End}{\operatorname{End}}
\newcommand{\id}{\mathbf 1}
\newcommand{\ket}[1]{\lvert #1\rangle}
\newcommand{\bra}[1]{\langle #1\rvert}

\newcommand{\cJ}{\mathcal J}
\newcommand{\Swap}{\mathsf P}

\definecolor{opentblue}{RGB}{38,98,156}
\definecolor{opentlightblue}{RGB}{225,237,248}
\definecolor{opentorange}{RGB}{205,104,35}
\definecolor{opentlightorange}{RGB}{250,231,215}
\definecolor{opentgray}{RGB}{105,111,118}
\tikzset{
  opent figure/.style={line cap=round,line join=round,font=\small},
  opent/wire/.style={draw=black!78,line width=0.75pt},
  opent/wire faded/.style={draw=black!22,line width=0.65pt},
  opent/gate/.style={
    rectangle,draw=opentblue,fill=opentlightblue,rounded corners=1.2pt,
    line width=0.8pt,minimum width=13mm,minimum height=5mm,
    inner sep=1pt,font=\small
  },
  opent/gate inverse/.style={
    opent/gate,draw=opentgray,fill=black!5
  },
  opent/gate faded/.style={
    rectangle,draw=black!20,fill=black!2,rounded corners=1.2pt,
    line width=0.55pt,minimum width=9.5mm,minimum height=4.2mm,
    inner sep=0.5pt,font=\scriptsize,text=black!35
  },
  opent/gate causal/.style={
    rectangle,draw=opentblue,fill=opentlightblue,rounded corners=1.2pt,
    line width=0.8pt,minimum width=9.5mm,minimum height=4.2mm,
    inner sep=0.5pt,font=\scriptsize
  },
  opent/operator/.style={
    rectangle,draw=opentorange,fill=opentlightorange,rounded corners=1pt,
    line width=0.9pt,minimum width=5.5mm,minimum height=5.5mm,
    inner sep=0.5pt,font=\small
  },
  opent/folded wire/.style={draw=opentblue!92!black,line width=1.05pt},
  opent/crossing/.style={
    circle,draw=opentblue!92!black,fill=opentblue!92!black,
    line width=0.6pt,minimum size=2.35mm,inner sep=0pt
  },
  opent/identity/.style={
    circle,draw=opentblue!92!black,fill=white,
    line width=0.9pt,minimum size=3.4mm,inner sep=0pt
  },
  opent/folded source/.style={
    rectangle,draw=opentorange,fill=opentlightorange,rounded corners=0.8pt,
    line width=0.9pt,minimum size=4.5mm,inner sep=0pt,font=\scriptsize
  },
  opent/flow arrow/.style={-{Latex[length=2mm]},draw=black!65,line width=0.7pt},
  opent/time arrow/.style={-{Latex[length=2mm]},draw=opentgray,line width=0.7pt},
  opent/region frame/.style={
    draw=opentorange,densely dashed,rounded corners=2pt,line width=0.8pt
  },
  opent/panel label/.style={font=\small\bfseries},
  opent/annotation/.style={font=\footnotesize,align=center,text=black!75}
}

\title{On the growth of operator entanglement in brickwork circuits with Yang--Baxter gates}
\author{Balázs Pozsgay\\
\small MTA-ELTE ``Momentum'' Integrable Quantum Dynamics Research Group\\
\small ELTE E\"otv\"os Lor\'and University, Budapest, Hungary\\
\small \href{mailto:pozsgay.balazs@ttk.elte.hu}{\texttt{pozsgay.balazs@ttk.elte.hu}}}

\begin{document}
\maketitle

\begin{abstract}
We study the operator entanglement of local operators in infinite
one-dimensional brickwork circuits whose two-site gate satisfies the braid relation;
throughout this work, we call such a gate a Yang--Baxter gate.  We
establish upper bounds for several structured, overlapping classes of
Yang--Baxter gates.  We show that the operator Schmidt rank remains uniformly
bounded in time for all qubit Yang--Baxter gates and, in arbitrary local
dimension, for permutation gates obtained from non-degenerate Yang--Baxter
maps.  We also show that it grows at most polynomially for involutive
dual-unitary Yang--Baxter gates and for arbitrary phase dressings of
permutation gates obtained from non-degenerate Yang--Baxter maps.  These results
imply, respectively, constant and logarithmic upper bounds on the operator
entanglement.  Conversely, we construct a seven-state involutive
Yang--Baxter gate without dual unitarity and a one-site operator whose exact
operator Schmidt rank grows exponentially, although the corresponding
operator entropies remain undetermined.  Entanglement growth
in the general Yang--Baxter case remains open.
{All proofs and selected examples were constructed using ChatGPT.} 
\end{abstract}

\section{Introduction}
\label{sec:introduction}

Understanding the complexity of quantum many-body dynamics is a central problem in
contemporary theoretical physics. One manifestation of this complexity is the ability
of time evolution to generate entanglement from initially unentangled states. A
closely related question arises in the Heisenberg picture, where a local observable
spreads through the system and develops entanglement in operator space. The 
complexity of the evolved operator can be quantified by its operator Schmidt rank
(OSR) and by the associated operator-space entropies \cite{Zanardi2001}. When this complexity remains
low, the dynamics can be simulated efficiently with matrix-product-operator methods.
Generic chaotic dynamics typically produces an operator-space von Neumann entropy that
grows linearly in time, whereas tractable systems often exhibit logarithmic growth or
saturation.
When investigating these questions, often it is more fruitful to focus on systems with discrete time evolution.
One-dimensional quantum brickwork circuits belong to this class of systems.

We consider the following general question: In a brickwork circuit in infinite volume, what type of special
conditions should the local two-site gate satisfy, so that operator entanglement remains logarithmically growing or
perhaps even bounded? What types of mechanisms can exist that limit this growth?

There are various known examples in the literature for the constrained growth of operator entanglement, both for quantum
circuits but also for continuous-time systems. 

Clifford circuits provide an early example: they map every Pauli string
to a single Pauli string, an evolved Pauli operator has $\OSR=1$, while an arbitrary
one-site qubit operator has $\OSR\leq 4$ \cite{Gottesman1999}.
The same mechanism applies to qudit Clifford circuits of arbitrary local dimension
$d$: generalized Pauli strings remain single strings, and an arbitrary one-site
operator therefore has $\OSR\leq d^2$ \cite{HostensDehaeneDeMoor2005}.
Another early example is provided by free-fermionic systems, where the
operator-space entropies grow logarithmically in time
\cite{ProsenPizorn2007,PizornProsen2009}.

Polynomial-memory mechanisms were identified in the Rule~54 model. A
time-dependent matrix-product construction and operator-entanglement analysis yielded,
in particular, an exact $O(t^2)$ product decomposition for local operators
\cite{KlobasMedenjakProsenVanicat2019,AlbaDubailMedenjak2019}. Although Rule~54 is a
cellular circuit with a slightly different geometry, it provides a useful comparison to the brickwork circuits.

Dual-unitary circuits provide an analytically tractable setting for
local-operator entanglement.  For completely chaotic circuits,
Bertini, Kos, and Prosen obtained exact linear-growth results in
asymptotic light-cone-edge limits
 \cite{BertiniKosProsenI}.  For qubit
circuits with ultralocal solitons propagating in both directions, they
proved that the operator entanglement of every one-site local operator
remains bounded \cite{BertiniKosProsenII}.

A deterministic ``hard-core gas'' cellular automaton was 
introduced in \cite{MedenjakKlobasProsen2017}, and its transport properties were calculated exactly.
{
Operator spreading in the corresponding quantum hard-core-gas family and its
generalizations was studied by Medenjak in Ref.~\cite{Medenjak2022}.
The exact time-dependent matrix-product representation constructed there has
auxiliary dimension $O(t)$, giving an $O(t)$ upper bound on the OSR and
logarithmic upper bounds on operator-space entanglement.
}
These models are integrable; their algebraic structures are related to the so-called XXC models and
their generalizations, 
introduced earlier by Maassarani \cite{su3-xx,Maassarani1998XXC,Maassarani1999Multiplicity}. The two-site gates
satisfy the braid relations, and the circuits are in fact superintegrable \cite{GomborPozsgay2022}.

A different finite-memory mechanism was found by
Wang in circuits constructed from finite-dimensional $C^*$-Hopf algebras. These
circuits are dual-unitary, and their local observables admit exact matrix-product
representations with time-independent bond dimension; the exact OSR is therefore
bounded \cite{Wang2025}. The two-site gates in \cite{Wang2025} satisfy the braid relation in the Abelian cases. 

Most recently, Tan and Prosen showed that operator entanglement grows at most
logarithmically in a clean nonintegrable, semi-ergodic brickwork circuit with dual
unitarity \cite{TanProsen2026}. The evolution of a traceless one-site operator is
confined to a restricted subspace and maps to a sequential qutrit--qubit scattering
problem.

Taken together, these works
reveal several mechanisms for slow or bounded
growth, including free-fermion structure, Clifford solvability, 
 time-dependent matrix-product representations, dual unitarity with extra constraints, and Hopf-algebraic
structures. Their diversity suggests that no single mechanism accounts for all known
cases.

The appearance of the braid relation in both XXC-type circuits and Abelian
Hopf-algebra circuits raises the question of how strongly that relation alone
constrains operator entanglement.
If the two-site gate satisfies the braid relation, then the model has ballistically propagating conserved local operators
known as gliders \cite{GomborPozsgay2022}.  When the simplest gliders are not proportional to identity operators, their
products generate the exponentially large family of conserved quantities characteristic of a superintegrable circuit.
The gliders have time-independent Schmidt rank for co-moving bipartitions. However, their existence does not
imply that the dynamics in these models is trivial, simply because the full operator algebra is typically not spanned by
the gliders only, and the remaining operators can have very non-trivial dynamics.

In this work we ask the following main question: How does operator entanglement grow in circuits where the two-site gate
satisfies the braid relation? For such unitary gates we will use the name ``Yang--Baxter gate''. We also ask: What
additional algebraic structures are needed so that sublinear growth can be rigorously proven? Similarly, which
constraints guarantee that the operator entanglement remains bounded?

One of our central objects is the OSR itself.  {
In several families of circuits we prove polynomial or time-independent
upper bounds on the OSR, which imply logarithmic or time-independent upper
bounds on the von Neumann and R\'enyi operator entropies, respectively.
} We also construct a Yang--Baxter circuit in which the OSR grows
exponentially.  That result does not determine
the 
corresponding growth law
for the R\'enyi or von Neumann operator entropies, however, it signals that polynomially growing OSR is not a common
feature of all circuits with Yang--Baxter gates.

The paper consists of 12 sections and {5} appendices.
The next section introduces the problem, fixes the necessary definitions, and
states the main results. In Section~\ref{sec:rectangular-reduction} we perform the basic reductions needed for the
computation of the OSR. Later sections treat the following main families and
examples:
\begin{itemize}
\item Elementary bounded-rank constructions
  (Section~\ref{sec:elementary-bounded-rank}).
\item Yang--Baxter gates for qubits
  (Section~\ref{sec:qubit-classification}).
\item Dual-unitary and involutive Yang--Baxter gates
  (Section~\ref{sec:inv-polynomial-proof}).
\item The LPW normal-form family, introduced later in
  Subsection~\ref{subsec:inv-LPW-normal-form} and analyzed in
  Section~\ref{sec:normal-form-dynamics}.
\item Non-degenerate Yang--Baxter maps and their phase dressings
  (Section~\ref{sec:yb-maps}).
\item A counterexample showing that the braid relation alone does not
  guarantee polynomial growth of the operator Schmidt rank
  (Section~\ref{sec:braid-only}).
\end{itemize}

In order to make the manuscript more accessible, we include Figure~\ref{fig:section-dependencies} at the end of this
document. It explains the logical dependencies and possible reading paths through the
manuscript. 

\subsection{The use of AI in this work}

We made heavy use of generative AI in both the computations and manuscript preparation. {
Almost all results were obtained by ChatGPT 5.6 Sol, whereas the eight-state
example and its proof, presented below in
Appendix~\ref{app:nonmonomial-involutive}, were developed later by ChatGPT 6 Astra.
}

The human author selected the specific research topic, and asked various questions from the LLM. All proofs and concrete
examples were constructed by the LLM. The author checked all the computations, and afterwards he guided the LLM to
prepare this manuscript. The author made a significant effort to make the manuscript accessible to human readers. This
involved multiple rewritings of the output of the LLM. All responsibility lies with the author.

{The technical contributions of the AI include} the proof for the dual-unitary involutive
Yang--Baxter gates, the various proofs for the permutation gates (bare and dressed), and the construction of the
exceptional gate with $D=7$ that 
produces exponentially growing OSR. The first and the last results were completely autonomously derived by the AI,
whereas in the case of the permutation gates the initial ideas of the author
also played a role. However, the actual proofs were again constructed by the AI. The most important contribution
of the human author was to select the main problem, and also to select the various particular cases that the AI could eventually
successfully attack. We note that earlier models of ChatGPT were not able to solve any of these concrete questions, and
all progress started when ChatGPT 5.6 Sol was released.

In certain cases we used a modified version of a special prompt published by OpenAI, which was used to solve the ``Cycle
double cover'' conjecture \cite{cdc}.

\section{The setting and main results}

\label{sec:problem-results}

In this section we fix the circuit conventions and state the main results.
In most of this work we focus on the exact operator Schmidt rank across a fixed spatial cut.
We first define this quantity, together with the different measures for the operator space entanglement.  We then introduce the
Yang--Baxter and dual-unitarity conditions imposed on the two-site gate and also discuss the integrability properties of
the circuits. Afterwards we
summarize the main results.

\subsection{Circuit and exact operator Schmidt rank}
\label{subsec:circuit-osr}

The local Hilbert space is $V\cong\mathbb C^D$, with $D>1$.  We consider a
homogeneous brickwork circuit built from a two-site unitary gate
$R\in U(V\otimes V)$.  One Floquet period consists of the two staggered
layers
\begin{align}
 U_{\rm even}&=\prod_{j\in\mathbb Z}R_{2j,2j+1},&
 U_{\rm odd}&=\prod_{j\in\mathbb Z}R_{2j-1,2j},&
 U_F&=U_{\rm odd}U_{\rm even}.
 \label{eq:setup-floquet}
\end{align}
The Heisenberg evolution of an operator is
\begin{equation}
 O(t)=U_F^{-t}OU_F^t .
 \label{eq:setup-heisenberg}
\end{equation}
Here and throughout the paper, $t$ counts complete Floquet periods.  Thus one
unit of $t$ consists of the two staggered layers $U_{\rm even}$ and
$U_{\rm odd}$.  We refer to the action of either layer separately as a
half-step.  Consequently, each branch of the Heisenberg network for $O(t)$
contains $2t$ brickwork layers.

Throughout the paper, $O$ is a one-site operator initially placed at site
zero.  Restricting to this case is sufficient because we are interested in upper
bounds on the OSR.  We argue in Section~\ref{sec:rectangular-reduction} that the placement of the operator and the
choice of the cut do not affect the core arguments. Furthermore, for these
OSR bounds it is enough to focus on one-site operators.
Indeed, an operator of width $w$, supported
on $w$ consecutive sites, can be expanded in a finite product basis; each
product term evolves into a product of $w$ evolved one-site operators.
Subadditivity of OSR under sums and submultiplicativity under products then
transfer any one-site upper bound to fixed $w$, without changing whether the
OSR growth is bounded, polynomial, or exponential.  We therefore state and prove
only one-site bounds below.

We work in infinite volume, taking the thermodynamic limit before the
long-time limit.  At every fixed $t$, all gates outside the causal window of
$O$ cancel exactly, and the problem reduces to a finite tensor network.

We place the spatial cut $c$ between sites $-1$ and $0$.  Across this cut,
the operator Schmidt decomposition is
\begin{equation}
 \begin{aligned}
 O&=\sum_{\mu=1}^{\chi}s_\mu
 O_{L,\mu}\otimes O_{R,\mu},
 \qquad s_\mu>0,\\
 \operatorname{Tr}(O_{L,\mu}^\dagger O_{L,\nu})
 &=\operatorname{Tr}(O_{R,\mu}^\dagger O_{R,\nu})
 =\delta_{\mu\nu}.
 \end{aligned}
 \label{eq:setup-osr}
\end{equation}
Thus the left and right operators are orthonormal in the Hilbert--Schmidt
inner product, and $s_\mu$ are the nonzero operator-Schmidt singular values.
Their number $\chi$ is the exact operator Schmidt rank, denoted by
$\OSR_c(O)$.  After vectorizing the operator, \eqref{eq:setup-osr} is the
ordinary Schmidt decomposition of a state in the doubled Hilbert space.

We set $\OSR_c(0)=0$.  Operator entropies, by contrast, require
Hilbert--Schmidt normalization and are therefore defined only for nonzero
operators.  Every operator-entropy statement below is understood with this
restriction.  For such an operator $O$, the normalized vectorized state has
Schmidt probabilities
$p_\mu=s_\mu^2/\sum_{\nu=1}^{\chi}s_\nu^2
=s_\mu^2/\|O\|_{\mathrm{HS}}^2$.  The R\'enyi operator entropies and the von
Neumann operator entropy are
\begin{align}
 S_\alpha^{\rm op}(O)
 &=\frac{1}{1-\alpha}\log\!\left(\sum_\mu p_\mu^\alpha\right),
 &&\alpha>0,\quad \alpha\ne1,
 \label{eq:setup-renyi}\\
 S_1^{\rm op}(O)&=-\sum_\mu p_\mu\log p_\mu,
 \label{eq:setup-von-neumann}\\
 S_0^{\rm op}(O)&=\log\OSR_c(O).
 \label{eq:setup-S0}
\end{align}
The Schmidt probabilities are supported on exactly
$\chi=\OSR_c(O)$ values.  For a probability distribution on $\chi$
outcomes, every R\'enyi entropy of nonnegative order is bounded by the entropy
of the uniform distribution.  Therefore
\begin{equation}
 S_0^{\rm op}(O)=\log\OSR_c(O),
 \qquad
 S_\alpha^{\rm op}(O)\leq\log\OSR_c(O),
 \quad \alpha>0.
 \label{eq:setup-entropy-rank}
\end{equation}
For $\alpha>0$, equality holds precisely when all nonzero
operator-Schmidt probabilities are equal.

We also note an elementary equivalence of brickwork circuits.  For a
one-site unitary $U$, replacing the gate $R$ by
\begin{equation}
 R'=(U\otimes U)R(U^\dagger\otimes U^\dagger).
 \label{eq:setup-onsite-gauge}
\end{equation}
conjugates the circuit, within every finite causal window, by the homogeneous
product of $U$ over all sites.  This product factorizes across the spatial
cut.  Hence, after replacing the one-site operator $O$ by $U^\dagger O U$,
the transformed evolution has the same exact OSR and the same
operator-Schmidt spectrum.  Therefore all operator-entanglement measures are
unchanged.  Furthermore, multiplying $R$ by an overall phase does not change
the Heisenberg evolution.

\subsection{Yang--Baxter gates}
\label{subsec:yb-gates}

The braid relation, or constant Yang--Baxter equation, for the two-site gate
$R$ is
\begin{equation}
 R_{12}R_{23}R_{12}=R_{23}R_{12}R_{23}
 \qquad\text{on }V^{\otimes3}.
 \label{eq:setup-braid}
\end{equation}
Homogeneous onsite unitary conjugation \eqref{eq:setup-onsite-gauge} preserves
this relation, because the intermediate one-site factors cancel on both
sides.  Multiplication by an overall phase also preserves it.
Throughout the paper, a \emph{Yang--Baxter gate} is a unitary two-site gate
$R\in U(V\otimes V)$ satisfying \eqref{eq:setup-braid}.  Thus ``gate''
already includes ordinary unitarity, while ``Yang--Baxter'' refers to the
braid relation.
Classifications of solutions to the braid relation include all two-state
solutions \cite{Hietarinta1992}, their unitary subclass \cite{Dye2003}, the
nonsingular upper-triangular and additive-charge-conserving three-state cases
\cite{Hietarinta1993UpperTriangular,HietarintaMartinRowell2024}, and the
charge-conserving class in arbitrary dimension \cite{MartinRowell2026}; only
Ref.~\cite{Dye2003} restricts the classification to unitary solutions.
More recently, Galindo and Rowell proposed a conjectural framework for
generating all unitary Yang--Baxter operators \cite{GalindoRowell2026}.  We
discuss its precise scope and its relation to the present problem later in
Subsection~\ref{subsec:involutive-monomial-relation}.

We now turn to set-theoretical Yang--Baxter maps, which provide a systematic
construction of unitary permutation gates.  Let $X$ be a set of $D$
local colours and identify $\{\ket{x}:x\in X\}$ with an orthonormal basis of
$V$.  A \emph{set-theoretical Yang--Baxter map} is a bijection
$r:X^2\to X^2$ satisfying the braid relation in set-theoretical form,
\begin{equation}
 r_{12}r_{23}r_{12}=r_{23}r_{12}r_{23}
 \qquad\text{on }X^3,
 \label{eq:setup-set-theoretical-braid}
\end{equation}
where $r_{12}=r\times\id_X$ and $r_{23}=\id_X\times r$.  It defines the
unitary permutation gate
\begin{equation}
 P_r\ket{x,y}=\ket{r(x,y)}.
 \label{eq:setup-permutation-gate}
\end{equation}
The map relation \eqref{eq:setup-set-theoretical-braid} is equivalent to the
gate relation \eqref{eq:setup-braid} for $P_r$.  Writing
$r(x,y)=(\lambda_x(y),\mu_y(x))$, the map is called \emph{non-degenerate} when
$\lambda_x:X\to X$ and $\mu_y:X\to X$ are permutations for every
$x,y\in X$.
Concrete examples and small-size enumerations can be found in the early work
\cite{EtingofSchedlerSoloviev1999}, the systematic enumeration
\cite{AkgunMerebVendramin2022}, and Ref.~\cite{GomborPozsgay2022}.

A \emph{phase-dressed Yang--Baxter-map gate} is constructed from a set-theoretical
Yang--Baxter map $r$ and unit-modulus numbers $\omega(x,y)$
via the formula
\begin{equation}
 R\ket{x,y}=\omega(x,y)\ket{r(x,y)},
 \qquad |\omega(x,y)|=1,
 \label{eq:setup-phase-dressed-map}
\end{equation}
We also call gates of the form \eqref{eq:setup-phase-dressed-map}
\emph{monomial gates}.  Throughout this work, the term ``monomial'' includes
the possible phase dressing $\omega$; a bare permutation gate is the special
case $\omega(x,y)=1$ for every $x,y\in X$.  We refer to $r$ as the underlying classical
map or ``support''.  The phase-dressed gate is dual-unitary when $r$ is
non-degenerate; dual unitarity is defined in
Subsection~\ref{subsec:dual-unitarity}.

The braid relation poses additional constraints on $\omega(x,y)$.
When these phase constraints hold, the full gate $R$ is itself a
Yang--Baxter gate.  However, throughout this work we use the term
``phase-dressed Yang--Baxter-map gate'' whenever the support $r$ satisfies
Eq.~\eqref{eq:setup-set-theoretical-braid}; the phases $\omega(x,y)$ need not
satisfy the additional constraints required for $R$ itself to obey the braid
relation.  The polynomial OSR bound proved later in
Subsection~\ref{subsec:phase-general} does not require these additional
constraints.

We shall also use involutivity.  A gate is involutive when $R^2=\id$, while
a set-theoretical Yang--Baxter map is involutive when $r^2=\id$.
Involutivity implies $R=R^\dagger$, and this can be understood as manifest time-reflection symmetry of the gate.

\subsection{Integrability}
\label{subsec:integrability}

The braid relation has striking dynamical consequences: it produces ballistically propagating gliders.  If the simplest
gliders are not proportional to identity operators, their products generate an exponentially growing number of conserved
quantities that form a non-Abelian symmetry algebra, and the circuit is superintegrable \cite{GomborPozsgay2022}.
In contrast to this exponentially large non-Abelian family, the traditional
notion of integrability is based on mutually commuting extensive charges with
local densities.  These charges are generated by a commuting family of transfer
matrices, and their number grows linearly with the system size.  For the pure
braid circuits considered here, this construction arises naturally when the
Yang--Baxter gate is involutive, because involutivity allows the constant gate to
be Baxterized.  In this subsection we discuss both structures.  We first describe
an exponentially growing family of gliders and then construct the commuting transfer matrices
and their local charges in the involutive case.

The \emph{gliders} are ballistically propagating local operators.
It was shown in \cite{GomborPozsgay2022} that the circuits in question have an exponentially growing number of algebraically
independent gliders. The paper \cite{GomborPozsgay2022} considered only the involutive case, but some of those results
can be 
extended also to the non-involutive case. As an example we consider the simplest gliders and their products.

With the shorthand $R_j=R_{j,j+1}$, define the two
three-site operators
\begin{align}
 \mathcal G^{\rm e}_{2k}
 &=R_{2k}^{\dagger}R_{2k+1}R_{2k},&
 \mathcal G^{\rm o}_{2k+1}
 &=R_{2k+1}R_{2k+2}R_{2k+1}^{\dagger}.
 \label{eq:setup-general-gliders}
\end{align}
Unitarity and the braid relation alone imply
\begin{align}
 U_F\mathcal G^{\rm e}_{2k}U_F^{-1}
 &=\mathcal G^{\rm e}_{2k+2},&
 U_F\mathcal G^{\rm o}_{2k+1}U_F^{-1}
 &=\mathcal G^{\rm o}_{2k-1}.
 \label{eq:setup-general-glider-translation}
\end{align}
Thus the two sublattices carry oppositely moving gliders.  These formulas
directly generalize the involutive construction presented in \cite{GomborPozsgay2022}.
A brief self-contained
verification is given in Appendix~\ref{app:general-gliders}.

Products of neighbouring gliders on either sublattice give longer operators
which translate with the same velocity and without changing their shape.
For every such density $g_j$, summing over its allowed translates gives a
conserved charge: time evolution merely relabels the terms in the sum.  To infer
superintegrability, we require the elementary gliders in
Eq.~\eqref{eq:setup-general-gliders} not to reduce to multiples of the identity.
Under this condition, their products generate an exponentially large family of
conserved charges.

The gliders discussed so far are only the simplest members of the full
family.  In the involutive case, Ref.~\cite{GomborPozsgay2022} constructs a
much larger class using representation-theoretic arguments.  For
noninvolutive gates, it is not known which members of that construction
remain gliders; to the best of our knowledge, this question has not yet been
investigated.  We do not need this finer classification here.  Our purpose is
only to establish the existence of an exponentially large family of conserved
gliders under the nontriviality condition stated above.

Gliders move rigidly under time evolution and therefore do not exhibit operator
spreading.  The glider sector nevertheless does not settle
the dynamics of a generic local observable.  The space of operators of range
$\ell$ has dimension $D^{2\ell}$, so there typically remain exponentially many local
operators which are not gliders.  Whether these operators spread and how
their operator entanglement grows are therefore nontrivial questions even in
a superintegrable Yang--Baxter circuit.

Let us now also discuss the traditional object in integrability: the commuting set of transfer matrices. These can be
introduced if the gate is involutive. In such a case the Yang--Baxter gate can be Baxterized: it can be lifted to a
spectral-parameter-dependent $R$-matrix. We introduce
\begin{equation}
 {
 \check{\mathcal R}(u)
 =\frac{\id+i u R}{1+i u}.}
 \label{eq:int-unitary-baxterization}
\end{equation}
For real $u$ it gives
\begin{equation}
 \check{\mathcal R}(u)^\dagger
 =\check{\mathcal R}(-u)
 =\check{\mathcal R}(u)^{-1},
 \qquad
 \check{\mathcal R}(0)=\id,
 \qquad
 \lim_{u\to\pm\infty}\check{\mathcal R}(u)=R.
 \label{eq:int-baxterization-properties}
\end{equation}
Thus the spectral-parameter-dependent matrix is unitary on the real
axis, is regular at the origin, and approaches the constant gate without an
additional scalar or phase. Furthermore, the braid relation implies that the $R$-matrix satisfies
\begin{equation}
 \check{\mathcal R}_{12}(u)
 \check{\mathcal R}_{23}(u+v)
 \check{\mathcal R}_{12}(v)
 =
 \check{\mathcal R}_{23}(v)
 \check{\mathcal R}_{12}(u+v)
 \check{\mathcal R}_{23}(u).
 \label{eq:int-checked-YBE}
\end{equation}
We stress that involutivity is required to prove both the properties in
Eq.~\eqref{eq:int-baxterization-properties} and the spectral-parameter
Yang--Baxter relation \eqref{eq:int-checked-YBE}.

Let
$\Swap$ denote the ordinary swap on $V\otimes V$ and define
\begin{equation}
 \mathcal R(u)=\Swap\,\check{\mathcal R}(u).
 \label{eq:int-vertex-R}
\end{equation}
Then $\mathcal R(u)$ satisfies the ordinary Yang--Baxter equation
\begin{align}
 &\mathcal R_{12}(u-v)\mathcal R_{13}(u-w)
 \mathcal R_{23}(v-w)
 \notag\\
 &\qquad=
 \mathcal R_{23}(v-w)\mathcal R_{13}(u-w)
 \mathcal R_{12}(u-v),
 \label{eq:int-vertex-YBE}
\end{align}
and the regularity condition is
\begin{equation}
 \mathcal R(0)=\Swap.
 \label{eq:int-vertex-regularity}
\end{equation}
The operators $\mathcal{R}(u)$ can be used to define the standard transfer matrices, and one obtains integrable quantum
circuits whose two-site gates are given by $\check{\mathcal R}(u)$ with some real fixed $u$
\cite{VanicatZadnikProsen2018}.  This is a standard procedure and it will be reviewed in Appendix~\ref{app:tms}.
However, our circuits are special because they use the actual constant gate, which is obtained in the
limit $u\to\infty$. In order to obtain non-trivial transfer matrices in this limit, a special scaling
procedure is required. We present it in Appendix~\ref{app:tms} and here we just cite the results. 

{
For this construction, we temporarily work on a periodic chain of even
length $L$.  Local charges are considered at fixed derivative order and
sufficiently large $L$.
}
For $\sigma\in\{+,-\}$, we define the two transfer matrices
\begin{equation}
 \tau_\sigma(v)
 =\Tr_a\!
 \left[
 \mathcal L_{aL}^{(\sigma)}(v)
 \mathcal L_{a,L-1}^{(\sigma)}(v)
 \cdots
 \mathcal L_{a1}^{(\sigma)}(v)
 \right],
 \label{eq:int-endpoint-transfer-explicit}
\end{equation}
where
\begin{align}
 \mathcal L_{aj}^{(+)}(v)
 &=
 \begin{cases}
  \mathcal R_{aj}(v),&j\ \text{even},\\
  (\mathcal R_\infty)_{aj},&j\ \text{odd},
 \end{cases}
 &
 \mathcal L_{aj}^{(-)}(v)
 &=
 \begin{cases}
  (\mathcal R_\infty)_{aj},&j\ \text{even},\\
  \mathcal R_{aj}(v),&j\ \text{odd},
 \end{cases}
 \label{eq:int-endpoint-Lax-operators}
\end{align}
where we also used
\begin{equation}
 \mathcal R_\infty
 :=\lim_{u\to\pm\infty}\mathcal R(u)
 =\Swap R.
 \label{eq:int-vertex-infinite-limit}
\end{equation}
In Appendix~\ref{app:tms} we prove that
\begin{equation}
 [\tau_\sigma(v),\tau_{\sigma'}(w)]=0,
 \qquad
 \sigma,\sigma'\in\{+,-\}.
 \label{eq:int-endpoint-transfer-commutativity}
\end{equation}
and that the Floquet operator can be obtained as
\begin{equation}
 U_F=\tau_-(0)^{-1}\tau_+(0).
 \label{eq:int-endpoint-Floquet-ratio}
\end{equation}
Combining this result with
Eq.~\eqref{eq:int-endpoint-transfer-commutativity}, we obtain
\begin{equation}
 [U_F,\tau_\pm(v)]=0
 \qquad
 \text{for all }v.
 \label{eq:int-endpoint-U-commutes}
\end{equation}
This establishes the Yang--Baxter integrability of the pure braid circuit.

The transfer matrices generate local conserved charges.  We
define
\begin{equation}
 Q_{n,\infty}^{(\sigma)}
 =
 -i\left.
 \frac{d^n}{dv^n}
 \log\!\left[
 \tau_\sigma(v)\tau_\sigma(0)^{-1}
 \right]
 \right|_{v=0},
 \qquad
 \sigma\in\{+,-\},\quad n\geq1.
 \label{eq:int-endpoint-charges}
\end{equation}
For the first charges we obtain explicitly
\begin{equation}
{\begin{aligned}
 Q_{1,\infty}^{(+)}
 &=\sum_{k=1}^{L/2}
 \left(R_{2k}R_{2k+1}R_{2k}-\id\right),\\
 Q_{1,\infty}^{(-)}
 &=\sum_{k=1}^{L/2}
 \left(R_{2k-1}R_{2k}R_{2k-1}-\id\right).
 \end{aligned}}
 \label{eq:int-endpoint-first-charges}
\end{equation}
Apart from the additive terms, these extensive charges are simply sums over the simplest gliders.

It can be shown that the higher derivatives of the transfer matrices are also spatial sums of gliders. A general
proof can be given using the representation-theoretic ideas of \cite{GomborPozsgay2022}, which connect involutive Yang--Baxter gates with
symmetric groups. However, this statement lies outside the main developments of this paper, so we omit the
proof.

\subsection{Dual unitarity}
\label{subsec:dual-unitarity}

Let us use the matrix elements
\begin{equation}
 R^{cd}_{ab}=\bra{c,d}R\ket{a,b}.
 \label{eq:setup-indices}
\end{equation}
The reshuffled matrix is
\begin{equation}
 (R^\Gamma)_{(a,c),(b,d)}:=R^{cd}_{ab}.
 \label{eq:setup-reshuffle}
\end{equation}
We say that a gate is \emph{dual-unitary} if both $R$ and $R^\Gamma$ are
unitary.  In indices the second unitarity condition is
\begin{equation}
 \sum_{b,d}R^{cd}_{ab}\,\overline{R^{c'd}_{a'b}}
 =\delta_{aa'}\delta_{cc'}.
 \label{eq:setup-duality-indices}
\end{equation}
This second unitarity condition expresses unitary propagation after a
quarter-turn of the gate in space--time.

For both the permutation gate
$P_r$ in \eqref{eq:setup-permutation-gate} and the phase-dressed gate $R$ in
\eqref{eq:setup-phase-dressed-map}, the classical map $r$ is non-degenerate if
and only if the corresponding gate is dual-unitary
\cite{GomborPozsgay2022}.

The dynamical properties of circuits built from dual-unitary gates were
reviewed in Ref.~\cite{BertiniClaeysProsen2026}.

{
\subsection{Remarks on the classification of Yang--Baxter gates}
\label{subsec:involutive-monomial-relation}

A full classification of Yang--Baxter gates is currently available only
for local dimension $D=2$, which we discuss later in
Section~\ref{sec:qubit-classification}.  We therefore study several
overlapping families separately, focusing on dual unitarity and
involutivity.  We also consider the separated phase-dressed exchanges and
controlled-swap gates introduced later in
Subsections~\ref{subsec:separated-phase-dressed} and
\ref{sec:controlled-swap}, as well as the phase-dressed Yang--Baxter-map
gates of Eq.~\eqref{eq:setup-phase-dressed-map}.

Involutive dual-unitary Yang--Baxter gates and phase-dressed
Yang--Baxter-map gates with non-degenerate support form overlapping
families.  We prove polynomial upper bounds on the OSR for these families
later in Sections~\ref{sec:inv-polynomial-proof} and
\ref{sec:yb-maps}, respectively.  Neither family contains the other, even
after homogeneous onsite unitary changes of basis:
\begin{itemize}
\item \emph{Involutive dual-unitary gates need not have a monomial local form.}
The example constructed later in
Appendix~\ref{app:nonmonomial-involutive} admits no phase-dressed permutation
form after any local unitary change of basis.
\item \emph{Non-degenerate Yang--Baxter maps need not be involutive.}
Such noninvolutive maps give bare permutation Yang--Baxter gates which are
dual-unitary but not involutive, without any phase dressing.
\end{itemize}

A recent preprint by Galindo and Rowell
\cite{GalindoRowell2026} approaches the wider problem of classifying all
unitary Yang--Baxter operators.  They conjecture that, up to an overall phase
and a homogeneous local change of basis, every such operator is generated
from three sources: monomial solutions, group-type solutions, and solutions
from twisted group-algebra towers.  Here ``generated'' includes external
products, direct sums, and admissible internal orthogonal sums.
Their work does not address dual unitarity.  The distinction described
above within the involutive dual-unitary class is a result of the present
work.
}

\subsection{Main results}
\label{subsec:main-results}

We now summarize the main conclusions.

\begin{enumerate}
\item \emph{For separated phase-dressed exchanges, time-evolved one-site
operators have uniformly bounded operator Schmidt rank.}  Consider the
separated phase-dressed exchange gates $R_{f,g,q}$, introduced later in
Subsection~\ref{subsec:separated-phase-dressed}.  For arbitrary local
dimension $D$, every one-site operator $O$ initially placed at site zero
satisfies, for every $t\geq0$ and every spatial cut $c$,
\begin{equation}
 \OSR_c(O(t))\leq D(D-1)+1.
 \label{eq:results-separated-exchange}
\end{equation}
Consequently, for nonzero $O$, every nonnegative-R\'enyi operator entropy
is at most \(\log\!\bigl(D(D-1)+1\bigr)\), uniformly in time.
The bound uses only the separated form and therefore remains valid even when
the gate does not satisfy the braid relation.
This result is proved in
Section~\ref{sec:elementary-bounded-rank}.

\item \emph{For Yang--Baxter gates built from controlled one-site
unitaries, time-evolved one-site operators have uniformly bounded operator
Schmidt rank.}  The two orientations of this construction, $R_{\rm rc}$
and $R_{\rm lc}$, are defined later in
Subsection~\ref{sec:controlled-swap}.  For every one-site operator $O$
initially placed at site zero and every $t\geq1$,
\begin{align}
 R=R_{\rm rc}
 &\quad\Longrightarrow\quad
 \OSR_c(O(t))\leq D^2,
 \nonumber\\
 R=R_{\rm lc}
 &\quad\Longrightarrow\quad
 \OSR_c(O(t))\leq D^2-D+1,
 \label{eq:results-controlled-swap}
\end{align}
Consequently, every nonnegative-R\'enyi operator entropy of a nonzero $O$
is bounded uniformly in time.
This result is proved in
Section~\ref{sec:elementary-bounded-rank}.

\item \emph{Time-evolved operators in qubit Yang--Baxter circuits have
uniformly bounded Schmidt rank.}  Let $D=2$ and let $R$ be any Yang--Baxter
gate.  Then every one-site operator $O$
initially placed at site zero satisfies, for every $t\geq0$ and every spatial
cut $c$,
\begin{equation}
 \OSR_c(O(t))\leq4.
 \label{eq:results-qubit}
\end{equation}
Consequently, for nonzero $O$, every nonnegative-R\'enyi operator entropy is
at most $\log4$, uniformly in time.
This result is proved in Section~\ref{sec:qubit-classification}.

\item \emph{Involutive Yang--Baxter gates with dual unitarity yield
polynomial bounds on the operator Schmidt rank.}  For general local
dimension $D$, let $R$ be an involutive Yang--Baxter gate with dual
unitarity.  Then every one-site operator at site zero satisfies
\begin{equation}
 \OSR_c(O(t))\leq
 \binom{t+D^2-1}{D^2-1}\leq(t+1)^{D^2-1}.
 \label{eq:results-involutive-one}
\end{equation}
Consequently every nonnegative-R\'enyi operator entropy is at most
$(D^2-1)\log(t+1)$.
This result is proved in Section~\ref{sec:inv-polynomial-proof}.

\item \emph{LPW normal-form gates admit a sharp classification of the
possible growth powers.}  For a special class of gates, the maximal one-site operator Schmidt
rank is bounded, grows linearly, or grows quadratically in time, depending on
the block structure.  These alternatives are exhaustive and all three are
attained.  The class is defined later in
Subsection~\ref{subsec:inv-LPW-normal-form}, and the result is proved in
Section~\ref{sec:normal-form-dynamics}.

\item \emph{Non-degenerate Yang--Baxter maps yield uniform bounds on the
operator Schmidt rank.}
For every finite non-degenerate set-theoretical Yang--Baxter map, with no
involutivity assumption, every one-site operator at site zero obeys
\begin{equation}
 \OSR_c(O(t))\leq \kappa_\lambda^4\kappa_\rho^2\leq(D!)^6.
 \label{eq:results-map-one}
\end{equation}
Here $\kappa_\lambda$ and $\kappa_\rho$ are two finite, gate-dependent
integers.  Their
precise definition, and in particular the bound
$\kappa_\lambda,\kappa_\rho\leq D!$, is given in
Section~\ref{sec:yb-maps}, where this result is also proved.

\item \emph{Arbitrary phase dressings of non-degenerate Yang--Baxter
maps yield polynomial bounds on the operator Schmidt rank.}  For every
phase-dressed Yang--Baxter-map gate with finite non-degenerate support,
whether or not the full dressed gate satisfies the braid relation, every
one-site operator at site zero satisfies
\begin{equation}
 \OSR_c(O(t))\leq \kappa_\lambda^4\kappa_\rho^2
 \binom{t+D\kappa_\lambda\kappa_\rho-1}
 {D\kappa_\lambda\kappa_\rho-1}^{\!2}
 \leq \kappa_\lambda^4\kappa_\rho^2
 (t+1)^{2D\kappa_\lambda\kappa_\rho-2}.
 \label{eq:results-phase-one}
\end{equation}
This result is proved in Section~\ref{sec:yb-maps}.

\item Let $R$ be a phase-dressed Yang--Baxter-map gate with
finite, non-degenerate, involutive support, and suppose that the full dressed
gate satisfies the braid relation.  The dressed gate itself need not be
involutive.  Then every one-site operator at site zero satisfies
\begin{equation}
 \OSR_c(O(t))\leq
 \binom{t+D^2-1}{D^2-1}\leq(t+1)^{D^2-1}
 \label{eq:results-phase-involutive}
\end{equation}
This result is proved in
Subsection~\ref{subsec:phase-involutive-support}.  On this class, both the
present estimate and the preceding general phase-dressing estimate apply, so
the smaller of the two upper bounds may be used.

\item \emph{Without dual unitarity, the operator Schmidt rank can grow
exponentially.}
There exist a fixed $D=7$ involutive Yang--Baxter gate without dual unitarity
and a fixed one-site operator such that
\begin{equation}
 \OSR_c(O(t))\geq2^{t-1},\qquad t\geq1.
 \label{eq:results-braid-only}
\end{equation}
This result is proved in Section~\ref{sec:braid-only}.
We only treat the OSR; the proof does not provide estimates for the operator
entropies.

\end{enumerate}

\section{Rearranging the circuit}
\label{sec:rectangular-reduction}

In this section we reduce the infinite Heisenberg circuit to a finite
rectangular braid, which will be used in all later arguments.  We first work with the
ordinary quantum circuit and isolate the causal network of the local
operator.  Afterwards we fold ket and bra indices into a single layer,
thereby isolating a finite rectangular braid that is easier to manipulate.
The entire reduction uses only
ordinary unitarity.
An equivalent formulation of this reduction was given in
Ref.~\cite{BertiniKosProsenI} (see also Ref.~\cite{JonayHuseNahum2018}).  In
that formulation, the surviving network is expressed as a rectangular corner
transfer matrix.

\subsection{Brickwork evolution and the causal network}
\label{subsec:ordinary-causal-network}

Throughout this section we use circuit diagrams to explain the algebraic
steps in the reduction.  Let us first fix the graphical conventions.  The gate $R$ acts on two
neighbouring copies of $V$ and is represented by a rectangle.  The two
staggered layers in \eqref{eq:setup-floquet} form one Floquet period, as
shown in Fig.~\ref{fig:ordinary-gate-brickwork}.  Time runs upwards in all
ordinary circuit figures in this subsection.

\begin{figure}[t]
\centering
\begin{tikzpicture}[opent figure,x=0.88cm,y=0.72cm]
  \begin{scope}[shift={(0,0)}]
    \node[opent/panel label] at (0.8,5.35) {(a) Two-site gate};
    \draw[opent/wire] (0.2,0.25) -- (0.2,4.75);
    \draw[opent/wire] (1.4,0.25) -- (1.4,4.75);
    \node[opent/gate,minimum width=18mm,minimum height=8mm] at (0.8,2.5) {$R$};
    \node[opent/annotation] at (0.2,-0.15) {$V$};
    \node[opent/annotation] at (1.4,-0.15) {$V$};
    \draw[opent/time arrow] (2.15,0.4) -- (2.15,4.55);
    \node[opent/annotation,rotate=90] at (2.48,2.5) {time};
  \end{scope}

  \begin{scope}[shift={(5.0,0)}]
    \node[opent/panel label] at (2.5,5.35) {(b) Brickwork circuit};
    \foreach \x in {0,1,2,3,4,5} {
      \draw[opent/wire] (\x,0.25) -- (\x,4.75);
      \node[opent/annotation] at (\x,-0.15) {$\x$};
    }
    \foreach \x in {0.5,2.5,4.5}
      \node[opent/gate] at (\x,0.95) {$R$};
    \foreach \x in {1.5,3.5}
      \node[opent/gate] at (\x,1.85) {$R$};
    \foreach \x in {0.5,2.5,4.5}
      \node[opent/gate] at (\x,2.95) {$R$};
    \foreach \x in {1.5,3.5}
      \node[opent/gate] at (\x,3.85) {$R$};
    \draw[decorate,decoration={brace,amplitude=4pt},draw=opentgray]
      (5.55,2.20) -- (5.55,0.58);
    \node[opent/annotation,anchor=west] at (5.85,1.39) {$U_F$};
    \node[opent/annotation,anchor=east] at (-0.25,0.95) {$U_{\rm even}$};
    \node[opent/annotation,anchor=east] at (-0.25,1.85) {$U_{\rm odd}$};
  \end{scope}
\end{tikzpicture}
\caption{Notations for quantum circuits.  A two-site gate is a rectangle with two
incoming and two outgoing legs.  The even and odd layers make one Floquet
period $U_F=U_{\rm odd}U_{\rm even}$.  The second period is included only to
make the staggering visible.}
\label{fig:ordinary-gate-brickwork}
\end{figure}
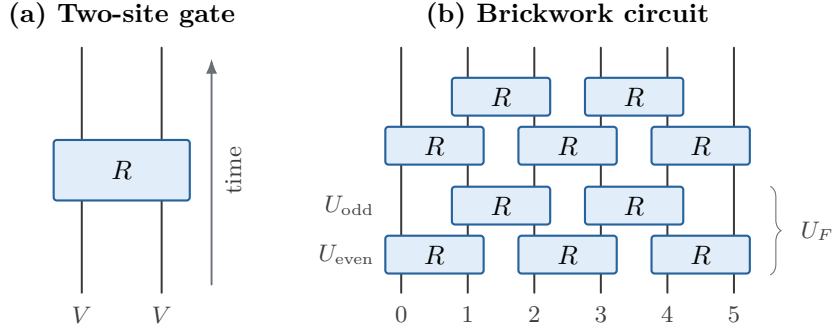

For a local operator $O$, the lower branch of the Heisenberg tensor network
is $U_F^t$, while the upper branch is
$U_F^{-t}=(U_F^t)^\dagger$.
Figure~\ref{fig:heisenberg-causal-network}(a) displays this double-layer
structure.

Ordinary unitarity now removes every gate pair outside the causal region of
the insertion.  This cancellation is exact, and at fixed $t$ it leaves a
finite tensor network.  In Fig.~\ref{fig:heisenberg-causal-network}(b), each
cancelled pair has been replaced by two straight vertical strands.  The
strands represent identity operators acting on the corresponding qudits.  The
braid relation has not been used at this stage.

\begin{figure}[t]
\centering
\begin{tikzpicture}[opent figure,x=0.58cm,y=0.62cm]
  \begin{scope}[shift={(0,0)}]
    \node[opent/panel label] at (4.5,7.55) {(a) Heisenberg evolution};
    \foreach \x in {0,1,2,3,4,5,6,7,8,9}
      \draw[opent/wire] (\x,0.10) -- (\x,7.00);

    \foreach \x in {0.5,2.5,4.5,6.5,8.5}
      \node[opent/gate causal] at (\x,0.65) {$R$};
    \foreach \x in {1.5,3.5,5.5,7.5}
      \node[opent/gate causal] at (\x,1.55) {$R$};
    \foreach \x in {0.5,2.5,4.5,6.5,8.5}
      \node[opent/gate causal] at (\x,2.45) {$R$};

    \node[opent/operator] at (5,3.55) {$O$};

    \foreach \x in {0.5,2.5,4.5,6.5,8.5}
      \node[opent/gate causal,draw=opentgray,fill=black!5] at (\x,4.65) {$R^\dagger$};
    \foreach \x in {1.5,3.5,5.5,7.5}
      \node[opent/gate causal,draw=opentgray,fill=black!5] at (\x,5.55) {$R^\dagger$};
    \foreach \x in {0.5,2.5,4.5,6.5,8.5}
      \node[opent/gate causal,draw=opentgray,fill=black!5] at (\x,6.45) {$R^\dagger$};

    \node[opent/annotation,anchor=east] at (-0.55,1.55) {$U_F^t$};
    \node[opent/annotation,anchor=east] at (-0.55,5.55) {$U_F^{-t}$};
    \node[opent/annotation] at (4.5,-0.40) {$O(t)=U_F^{-t}OU_F^t$};
  \end{scope}

  \begin{scope}[shift={(12.0,0)}]
    \node[opent/panel label] at (4.5,7.55) {(b) Finite causal network};
    \fill[opentlightorange,opacity=0.55]
      (5,3.55) -- (2.15,0.10) -- (7.85,0.10) -- cycle;
    \fill[opentlightorange,opacity=0.55]
      (5,3.55) -- (2.15,7.00) -- (7.85,7.00) -- cycle;
    \foreach \x in {0,1,2,3,4,5,6,7,8,9}
      \draw[opent/wire] (\x,0.10) -- (\x,7.00);

    \foreach \x in {2.5,4.5,6.5}
      \node[opent/gate causal] at (\x,0.65) {$R$};
    \foreach \x in {3.5,5.5}
      \node[opent/gate causal] at (\x,1.55) {$R$};
    \node[opent/gate causal] at (4.5,2.45) {$R$};
    \node[opent/operator] at (5,3.55) {$O$};
    \node[opent/gate causal,draw=opentgray,fill=black!5] at (4.5,4.65) {$R^\dagger$};
    \foreach \x in {3.5,5.5}
      \node[opent/gate causal,draw=opentgray,fill=black!5] at (\x,5.55) {$R^\dagger$};
    \foreach \x in {2.5,4.5,6.5}
      \node[opent/gate causal,draw=opentgray,fill=black!5] at (\x,6.45) {$R^\dagger$};
    \node[opent/annotation] at (4.5,-0.40)
      {cancelled pairs are replaced by straight strands};
  \end{scope}
\end{tikzpicture}
\caption{Evolution of a local operator and its causal structure.  Panel (a)
keeps the forward and backward circuits visible.  Ordinary unitarity removes
all gates which cannot affect the insertion.  Panel (b) displays the resulting
finite network.  The braid relation has not yet been used.}
\label{fig:heisenberg-causal-network}
\end{figure}
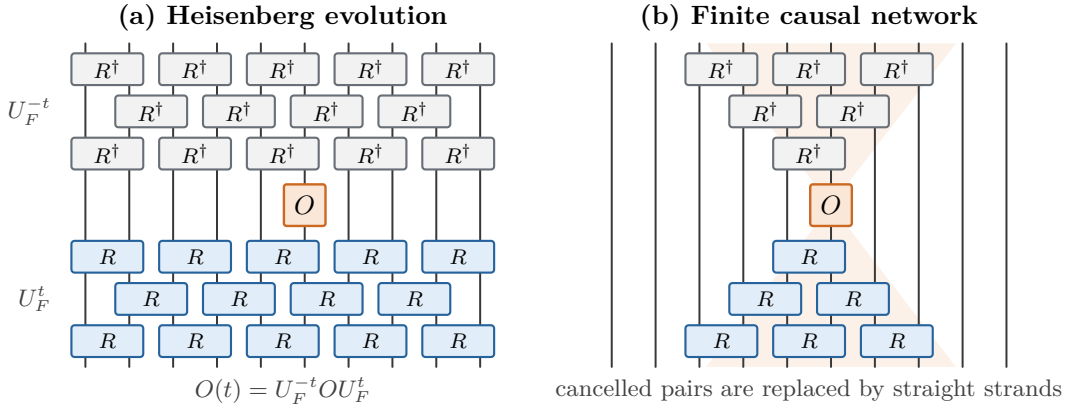

\subsection{Folding the operator network}
\label{subsec:folding-network}

We now fold the double-layer operator network into a single layer.  Put
$\mathcal W=\End(V)$.  Folding groups the ket and bra indices on each
physical site into one $D^2$-dimensional local space.  With the orientation
used below, the one-site vectorization is fixed by
\begin{equation}
 \mathfrak f\bigl(\ket{a'}\!\bra{a}\bigr)=\ket{a,a'}.
 \label{eq:centralizer-folding}
\end{equation}
We write $x_O:=\mathfrak f(O)$ for the folded vector associated with a
one-site operator $O$.
Graphically, an operator is first drawn in the ordinary way as a box on one
vertical physical wire, with one index below and one above.  Folding bends
the lower leg upwards and produces a cup with two outgoing indices.  In all
subsequent folded networks, we abbreviate this cup by an orange operator box
with one folded strand leaving vertically upwards.  These three stages are
shown in Fig.~\ref{fig:folding-dictionary}(a).

A distinguished folded vector is the vectorized identity,
\begin{equation}
 \iota:=\mathfrak f(\id)=\sum_{a=1}^D\ket{a,a}.
 \label{eq:folded-identity}
\end{equation}
It is denoted by the empty circle in
Fig.~\ref{fig:folding-dictionary}(b).

\begin{figure}[t]
\centering
\begin{tikzpicture}[opent figure,x=0.84cm,y=0.80cm]
  \begin{scope}[shift={(0,0)}]
    \node[opent/panel label] at (2.75,3.10) {(a) One-site folding};

    \draw[opent/wire] (0.35,0.55) -- (0.35,2.40);
    \node[opent/operator] at (0.35,1.48) {$O$};
    \node[opent/annotation] at (0.35,2.72) {$a$};
    \node[opent/annotation] at (0.35,0.22) {$a'$};

    \draw[opent/flow arrow] (0.95,1.48) -- (1.55,1.48);

    \draw[opent/folded wire]
      (1.95,2.15) .. controls (1.95,1.20) and (2.25,0.82) ..
      (2.70,0.82) .. controls (3.15,0.82) and (3.45,1.20) ..
      (3.45,2.15);
    \node[opent/folded source] at (2.70,0.87) {$O$};
    \node[opent/annotation] at (1.95,2.47) {$a$};
    \node[opent/annotation] at (3.45,2.47) {$a'$};

    \draw[opent/flow arrow] (3.88,1.48) -- (4.48,1.48);

    \draw[opent/folded wire] (4.90,0.87) -- (4.90,2.15);
    \node[opent/folded source] at (4.90,0.87) {$O$};
    \node[opent/annotation] at (4.90,2.47) {$(a,a')$};
  \end{scope}

  \begin{scope}[shift={(8.0,0)}]
    \node[opent/panel label] at (0.8,3.10) {(b) Identity vector};
    \draw[opent/folded wire] (0.8,0.67) -- (0.8,2.40);
    \node[opent/identity] at (0.8,0.67) {};
    \node[opent/annotation] at (0.8,0.18) {empty circle};
    \node[opent/annotation] at (0.8,-0.25) {$\iota=\mathfrak f(\id)$};
  \end{scope}
\end{tikzpicture}
\caption{Folding a one-site operator.  In panel (a), the lower physical leg
is bent upwards, so the folded operator is a cup with outgoing indices
$a$ and $a'$.  Later diagrams use the equivalent one-strand shorthand
shown on the right.  Panel (b) reserves an empty circle for the vectorized
identity $\iota$.}
\label{fig:folding-dictionary}
\end{figure}
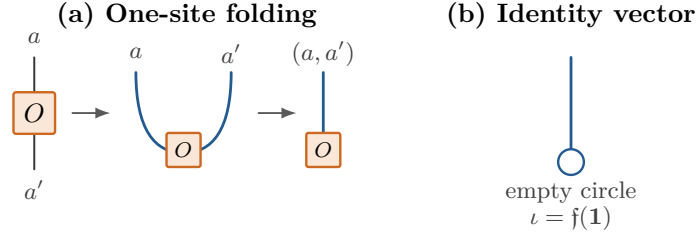

A blue strand in the folded diagrams therefore carries the
$D^2$-dimensional space $\mathcal W$.  Most importantly, the operator
Schmidt rank across a spatial cut is now the ordinary Schmidt rank of the
folded vector across the corresponding cut.  The entanglement problem has
thus been converted into a one-layer state problem.

We next fold the local Heisenberg action $X\mapsto R^\dagger XR$.  Let
$\mathfrak f_2=\mathfrak f\otimes\mathfrak f$ denote the sitewise folding of
a two-site operator.  The folded Heisenberg gate $\widehat R_{\rm H}$ is
defined by
\[
 \widehat R_{\rm H}\,\mathfrak f_2(X)
 =\mathfrak f_2(R^\dagger XR).
\]
Thus the two physical gates entering the folded map are precisely
$R^\dagger$ and $R$.

We now derive the components directly from this definition.  Apply the
folded map to the two-site matrix unit
$\ket{a',b'}\!\bra{c,d}$.  Its folded vector is
\[
 \mathfrak f_2\bigl(\ket{a',b'}\!\bra{c,d}\bigr)
 =\ket{c,a'}\otimes\ket{d,b'}.
\]
Using the matrix-element convention \eqref{eq:setup-indices}, ordinary
matrix multiplication gives
\[
 R^\dagger\ket{a',b'}\!\bra{c,d}R
 =
 \sum_{a,b,c',d'}
 R_{ab}^{cd}\,\overline{R_{c'd'}^{a'b'}}
 \ket{c',d'}\!\bra{a,b}.
\]
Folding the matrix units on the right-hand side therefore yields
\[
 \widehat R_{\rm H}
 \bigl(\ket{c,a'}\otimes\ket{d,b'}\bigr)
 =
 \sum_{a,b,c',d'}
 R_{ab}^{cd}\,\overline{R_{c'd'}^{a'b'}}
 \ket{a,c'}\otimes\ket{b,d'}.
\]
Thus $(c,a')$ and $(d,b')$ are the incoming folded indices, while
$(a,c')$ and $(b,d')$ are the outgoing folded indices.  In components,
\begin{equation}
 (\widehat R_{\rm H})^{(a,c'),(b,d')}_{(c,a'),(d,b')}
 =R_{ab}^{cd}\,\overline{R_{c'd'}^{a'b'}}.
 \label{eq:centralizer-folded-gate}
\end{equation}
Here $R_{ab}^{cd}=\bra{c,d}R\ket{a,b}$ comes from the $R$ branch, whereas
$\overline{R_{c'd'}^{a'b'}}=\bra{c',d'}R^\dagger\ket{a',b'}$ comes from the
$R^\dagger$ branch.  This is the origin of the complex conjugate in
Eq.~\eqref{eq:centralizer-folded-gate}.

Figure~\ref{fig:folded-gate-dictionary} depicts this calculation and the
resulting order of the folded indices.
The subscript ${\rm H}$ emphasizes that this is the folded Heisenberg gate,
not a second physical gate.

From this point onward we draw
$\widehat R_{\rm H}$ in braid notation: two blue strands cross, and a small
filled circle marks the crossing.  This symbol is deliberately different
from the rectangular gate used for the original quantum circuit.

\begin{figure}[t]
\centering
\begin{tikzpicture}[opent figure,x=0.88cm,y=0.82cm]
  \draw[opent/wire] (0.0,0.20) -- (0.0,1.90);
  \draw[opent/wire] (1.0,0.20) -- (1.0,1.90);
  \node[opent/gate causal,draw=opentgray,fill=black!5] at (0.5,1.02) {$R$};
  \node[opent/annotation] at (0.0,-0.14) {$a$};
  \node[opent/annotation] at (1.0,-0.14) {$b$};
  \node[opent/annotation] at (0.0,2.18) {$c$};
  \node[opent/annotation] at (1.0,2.18) {$d$};

  \draw[opent/wire] (0.0,2.85) -- (0.0,4.55);
  \draw[opent/wire] (1.0,2.85) -- (1.0,4.55);
  \node[opent/gate causal] at (0.5,3.72) {$R^\dagger$};
  \node[opent/annotation] at (0.0,2.57) {$a'$};
  \node[opent/annotation] at (1.0,2.57) {$b'$};
  \node[opent/annotation] at (0.0,4.88) {$c'$};
  \node[opent/annotation] at (1.0,4.88) {$d'$};

  \draw[opent/flow arrow] (1.75,2.38) -- (2.85,2.38);

  \begin{scope}[shift={(3.65,0)}]
    \draw[opent/folded wire] (0.0,0.55) -- (1.5,4.35);
    \draw[opent/folded wire] (1.5,0.55) -- (0.0,4.35);
    \node[opent/crossing] at (0.75,2.45) {};
    \node[opent/annotation] at (0.0,0.12) {$(c,a')$};
    \node[opent/annotation] at (1.5,0.12) {$(d,b')$};
    \node[opent/annotation] at (0.0,4.79) {$(a,c')$};
    \node[opent/annotation] at (1.5,4.79) {$(b,d')$};
  \end{scope}
\end{tikzpicture}
\caption{Folding the local Heisenberg action.  The left-hand side shows the
physical gates $R$ and $R^\dagger$.  The ends facing the operator insertion
become the incoming ends at the bottom of the folded crossing, while the four
outer ends become its outgoing ends at the top.  The filled circle denotes
the resulting folded gate $\widehat R_{\rm H}$.}
\label{fig:folded-gate-dictionary}
\end{figure}
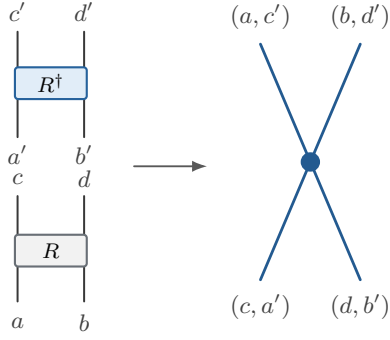

Taking the adjoint of \eqref{eq:setup-braid} shows that $R^\dagger$ satisfies
the same braid relation.  The folded crossing therefore inherits the local
braid relation.
Figure~\ref{fig:folded-braid-relation} shows the graphical identity that we
shall use from this point on.
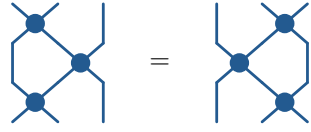
\begin{figure}[t]
\centering
\begin{tikzpicture}[opent figure,x=0.60cm,y=0.52cm,baseline=-0.5ex]
  \draw[opent/folded wire] (0,0) -- (1,1);
  \draw[opent/folded wire] (1,0) -- (0,1);
  \draw[opent/folded wire] (2,0) -- (2,1);
  \node[opent/crossing] at (0.5,0.5) {};
  \draw[opent/folded wire] (0,1) -- (0,2);
  \draw[opent/folded wire] (1,1) -- (2,2);
  \draw[opent/folded wire] (2,1) -- (1,2);
  \node[opent/crossing] at (1.5,1.5) {};
  \draw[opent/folded wire] (0,2) -- (1,3);
  \draw[opent/folded wire] (1,2) -- (0,3);
  \draw[opent/folded wire] (2,2) -- (2,3);
  \node[opent/crossing] at (0.5,2.5) {};

  \node at (3.25,1.5) {$=$};

  \begin{scope}[shift={(4.5,0)}]
    \draw[opent/folded wire] (0,0) -- (0,1);
    \draw[opent/folded wire] (1,0) -- (2,1);
    \draw[opent/folded wire] (2,0) -- (1,1);
    \node[opent/crossing] at (1.5,0.5) {};
    \draw[opent/folded wire] (0,1) -- (1,2);
    \draw[opent/folded wire] (1,1) -- (0,2);
    \draw[opent/folded wire] (2,1) -- (2,2);
    \node[opent/crossing] at (0.5,1.5) {};
    \draw[opent/folded wire] (0,2) -- (0,3);
    \draw[opent/folded wire] (1,2) -- (2,3);
    \draw[opent/folded wire] (2,2) -- (1,3);
    \node[opent/crossing] at (1.5,2.5) {};
  \end{scope}
\end{tikzpicture}
\caption{The braid relation for folded crossings.  The two diagrams contain
the same three strands and the same three pairwise crossings, in a different
order.}
\label{fig:folded-braid-relation}
\end{figure}

Ordinary unitarity takes a particularly simple form after folding.  Since
unitary conjugation leaves the two-site identity invariant, we have
\begin{equation}
 \widehat R_{\rm H}(\iota\otimes\iota)=\iota\otimes\iota.
 \label{eq:folded-unitarity}
\end{equation}
Graphically, two identity vectors pass unchanged through a folded crossing,
as shown in Fig.~\ref{fig:folded-unitarity}.

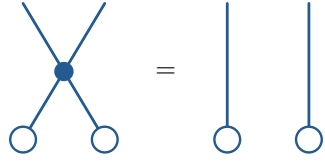
\begin{figure}[t]
\centering
\begin{tikzpicture}[opent figure,x=0.90cm,y=0.90cm]
  \draw[opent/folded wire] (0,0.55) -- (1.2,2.55);
  \draw[opent/folded wire] (1.2,0.55) -- (0,2.55);
  \node[opent/crossing] at (0.6,1.55) {};
  \node[opent/identity] at (0,0.55) {};
  \node[opent/identity] at (1.2,0.55) {};

  \node at (2.1,1.55) {$=$};

  \draw[opent/folded wire] (3.0,0.55) -- (3.0,2.55);
  \draw[opent/folded wire] (4.2,0.55) -- (4.2,2.55);
  \node[opent/identity] at (3.0,0.55) {};
  \node[opent/identity] at (4.2,0.55) {};
\end{tikzpicture}
\caption{The local unitarity rule in the folded picture.  A crossing acts
trivially on two vectorized identities, so the two crossed strands may be
replaced by straight strands.}
\label{fig:folded-unitarity}
\end{figure}

Let us apply this local rule to the finite Heisenberg network.  Folding turns
Fig.~\ref{fig:heisenberg-causal-network}(b) into a single network of blue
strands and folded crossings.  Initially every site except the insertion
carries the identity vector $\iota$; the folded insertion
$x_O$ is shown as an orange box labelled by $O$, as in
Fig.~\ref{fig:folding-dictionary}(a).  Applying the local rule of
Fig.~\ref{fig:folded-unitarity}, equivalently
\eqref{eq:folded-unitarity}, layer by layer removes every crossing that has
not been reached by the insertion.  The crossings capable of transporting
$O$ outwards remain and form the expanding causal network in
Fig.~\ref{fig:folded-causal-simplification}.

\begin{figure}[t]
\centering
\begin{tikzpicture}[opent figure,x=0.74cm,y=0.90cm]
  \foreach \x in {0,1,2,3,6,7,8,9}
    \draw[opent/folded wire] (\x,0.35) -- (\x,1.30);
  \draw[opent/folded wire] (4,0.35) -- (5,1.30);
  \draw[opent/folded wire] (5,0.35) -- (4,1.30);
  \node[opent/crossing] at (4.5,0.825) {};

  \foreach \x in {0,1,2,7,8,9}
    \draw[opent/folded wire] (\x,1.30) -- (\x,2.25);
  \draw[opent/folded wire] (3,1.30) -- (4,2.25);
  \draw[opent/folded wire] (4,1.30) -- (3,2.25);
  \node[opent/crossing] at (3.5,1.775) {};
  \draw[opent/folded wire] (5,1.30) -- (6,2.25);
  \draw[opent/folded wire] (6,1.30) -- (5,2.25);
  \node[opent/crossing] at (5.5,1.775) {};

  \foreach \x in {0,1,8,9}
    \draw[opent/folded wire] (\x,2.25) -- (\x,3.20);
  \draw[opent/folded wire] (2,2.25) -- (3,3.20);
  \draw[opent/folded wire] (3,2.25) -- (2,3.20);
  \node[opent/crossing] at (2.5,2.725) {};
  \draw[opent/folded wire] (4,2.25) -- (5,3.20);
  \draw[opent/folded wire] (5,2.25) -- (4,3.20);
  \node[opent/crossing] at (4.5,2.725) {};
  \draw[opent/folded wire] (6,2.25) -- (7,3.20);
  \draw[opent/folded wire] (7,2.25) -- (6,3.20);
  \node[opent/crossing] at (6.5,2.725) {};

  \foreach \x in {0,9}
    \draw[opent/folded wire] (\x,3.20) -- (\x,4.15);
  \draw[opent/folded wire] (1,3.20) -- (2,4.15);
  \draw[opent/folded wire] (2,3.20) -- (1,4.15);
  \node[opent/crossing] at (1.5,3.675) {};
  \draw[opent/folded wire] (3,3.20) -- (4,4.15);
  \draw[opent/folded wire] (4,3.20) -- (3,4.15);
  \node[opent/crossing] at (3.5,3.675) {};
  \draw[opent/folded wire] (5,3.20) -- (6,4.15);
  \draw[opent/folded wire] (6,3.20) -- (5,4.15);
  \node[opent/crossing] at (5.5,3.675) {};
  \draw[opent/folded wire] (7,3.20) -- (8,4.15);
  \draw[opent/folded wire] (8,3.20) -- (7,4.15);
  \node[opent/crossing] at (7.5,3.675) {};

  \foreach \x in {0,1,2,3,4,6,7,8,9}
    \node[opent/identity] at (\x,0.35) {};
  \node[opent/folded source] at (5,0.35) {$O$};
\end{tikzpicture}
\caption{Causal simplification of the folded network.  The ten input
strands make the cancelled exterior visible: empty-circle identity inputs
continue vertically wherever the local rule of
Fig.~\ref{fig:folded-unitarity} removes a crossing.  Only the expanding
network of crossings reached by the folded local operator remains.}
\label{fig:folded-causal-simplification}
\end{figure}
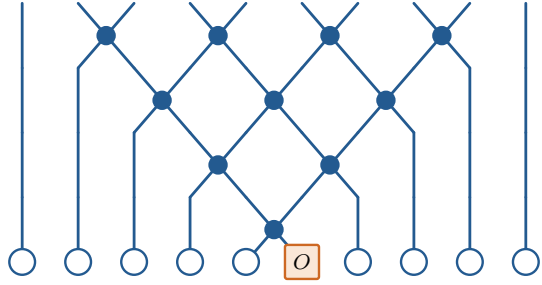

We now label the strands of the simplified network, after the unitarity
cancellation has been performed.  At time $t$, the active input strands start
at the sites $-2t,\ldots,2t-1$.  For $0\leq j,k<2t$, we denote the strand
starting at site $-2t+j$ by $A_j$ and the strand starting at site $k$ by
$B_k$.  The label follows the strand through all subsequent crossings, and
the folded local operator is inserted on $B_0$.  For $t=2$, the active inputs
are therefore $A_0,A_1,A_2,A_3,B_0,B_1,B_2,B_3$, and the surviving crossing
pairs are $(B_k,A_j)$ with $0\leq k\leq j<4$.

We now isolate a rectangular braid, which is a central object in all
calculations below.  We first show the extraction for $t=2$ and the isolated
rectangle for $t=3$.  Lemma~\ref{lem:yb-rectangle} then states the general
result for arbitrary $t$.

For $t=2$, the rectangular piece is obtained as follows.
Split both lists after two entries and set
$A=(A_2,A_3)$ and $B=(B_0,B_1)$.  Each strand in $B$ crosses each strand in
$A$ exactly once, so these four crossings form a complete $2\times2$
block.  Figure~\ref{fig:rectangular-crossing-extraction} draws this rectangle
separately from the two remaining triangular pieces.

\begin{figure}[t]
\centering
\begin{tikzpicture}[opent figure,x=0.88cm,y=0.78cm]
  \begin{scope}

    \foreach \x in {0,1,2,5,6,7}
      \draw[opent/folded wire] (\x,0.35) -- (\x,1.25);
    \draw[opent/folded wire] (3,0.35) -- (4,1.25);
    \draw[opent/folded wire] (4,0.35) -- (3,1.25);
    \node[opent/crossing] at (3.5,0.80) {};

    \foreach \x in {0,1,6,7}
      \draw[opent/folded wire] (\x,1.25) -- (\x,2.15);
    \draw[opent/folded wire] (2,1.25) -- (3,2.15);
    \draw[opent/folded wire] (3,1.25) -- (2,2.15);
    \node[opent/crossing] at (2.5,1.70) {};
    \draw[opent/folded wire] (4,1.25) -- (5,2.15);
    \draw[opent/folded wire] (5,1.25) -- (4,2.15);
    \node[opent/crossing] at (4.5,1.70) {};

    \foreach \x in {0,1,2,5,6,7}
      \draw[opent/folded wire] (\x,2.15) -- (\x,3.05);
    \draw[opent/folded wire] (3,2.15) -- (4,3.05);
    \draw[opent/folded wire] (4,2.15) -- (3,3.05);
    \node[opent/crossing] at (3.5,2.60) {};

    \foreach \x in {0,3,4,7}
      \draw[opent/folded wire] (\x,3.05) -- (\x,3.95);
    \draw[opent/folded wire] (1,3.05) -- (2,3.95);
    \draw[opent/folded wire] (2,3.05) -- (1,3.95);
    \node[opent/crossing] at (1.5,3.50) {};
    \draw[opent/folded wire] (5,3.05) -- (6,3.95);
    \draw[opent/folded wire] (6,3.05) -- (5,3.95);
    \node[opent/crossing] at (5.5,3.50) {};

    \draw[opent/folded wire] (0,3.95) -- (1,4.85);
    \draw[opent/folded wire] (1,3.95) -- (0,4.85);
    \node[opent/crossing] at (0.5,4.40) {};
    \draw[opent/folded wire] (2,3.95) -- (3,4.85);
    \draw[opent/folded wire] (3,3.95) -- (2,4.85);
    \node[opent/crossing] at (2.5,4.40) {};
    \draw[opent/folded wire] (4,3.95) -- (5,4.85);
    \draw[opent/folded wire] (5,3.95) -- (4,4.85);
    \node[opent/crossing] at (4.5,4.40) {};
    \draw[opent/folded wire] (6,3.95) -- (7,4.85);
    \draw[opent/folded wire] (7,3.95) -- (6,4.85);
    \node[opent/crossing] at (6.5,4.40) {};

    \begin{scope}[on background layer]
      \draw[opent/region frame] (1.75,0.52) rectangle (5.25,2.92);
    \end{scope}
    \node[opent/annotation,fill=white,inner sep=1pt,text=opentorange]
      at (3.5,2.92) {rectangular core};
    \node[opent/annotation] at (1.5,5.15) {$\ell_2$};
    \node[opent/annotation] at (5.5,5.15) {$r_2$};

    \foreach \x in {0,1,2,3,5,6,7}
      \node[opent/identity] at (\x,0.35) {};
    \node[opent/folded source] at (4,0.35) {$O$};
    \foreach \x/\lab in
      {0/A_0,1/A_1,2/A_2,3/A_3,4/B_0,5/B_1,6/B_2,7/B_3}
      \node[opent/annotation] at (\x,-0.12) {$\lab$};
  \end{scope}
\end{tikzpicture}
\caption{The factorized crossing network for $t=2$.  The four crossings
between $A=(A_2,A_3)$ and $B=(B_0,B_1)$ form the complete $2\times2$
folded rectangular core.  The other six crossings form the two
triangular networks $\ell_2$ and $r_2$.}
\label{fig:rectangular-crossing-extraction}
\end{figure}
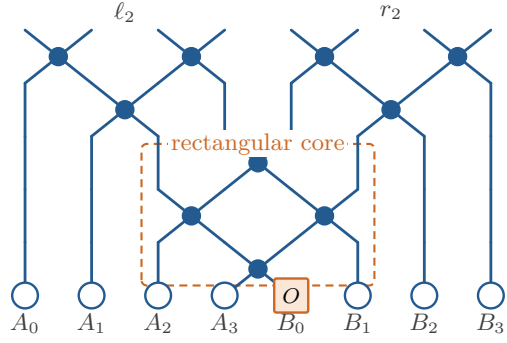

Figure~\ref{fig:rectangular-crossing-extraction} retains the original causal
labels to show which strands enter the rectangle.  Once the rectangle has
been isolated, we relabel its two length-$t$ blocks internally as
$A=(A_0,\ldots,A_{t-1})$ and $B=(B_0,\ldots,B_{t-1})$.

We next draw the rectangular piece in isolation for $t=3$.  The two source
blocks in Fig.~\ref{fig:isolated-three-by-three-block} are
\begin{equation}
 A=(A_0,A_1,A_2),\qquad B=(B_0,B_1,B_2).
 \label{eq:three-by-three-source-blocks}
\end{equation}
Each $B$ strand crosses each $A$ strand exactly once, giving nine crossings
in total.

\begin{figure}[t]
\centering
\begin{tikzpicture}[opent figure,x=0.95cm,y=0.82cm]
  \foreach \x in {0,1,4,5}
    \draw[opent/folded wire] (\x,0.40) -- (\x,1.30);
  \draw[opent/folded wire] (2,0.40) -- (3,1.30);
  \draw[opent/folded wire] (3,0.40) -- (2,1.30);
  \node[opent/crossing] at (2.5,0.85) {};

  \foreach \x in {0,5}
    \draw[opent/folded wire] (\x,1.30) -- (\x,2.20);
  \draw[opent/folded wire] (1,1.30) -- (2,2.20);
  \draw[opent/folded wire] (2,1.30) -- (1,2.20);
  \node[opent/crossing] at (1.5,1.75) {};
  \draw[opent/folded wire] (3,1.30) -- (4,2.20);
  \draw[opent/folded wire] (4,1.30) -- (3,2.20);
  \node[opent/crossing] at (3.5,1.75) {};

  \draw[opent/folded wire] (0,2.20) -- (1,3.10);
  \draw[opent/folded wire] (1,2.20) -- (0,3.10);
  \node[opent/crossing] at (0.5,2.65) {};
  \draw[opent/folded wire] (2,2.20) -- (3,3.10);
  \draw[opent/folded wire] (3,2.20) -- (2,3.10);
  \node[opent/crossing] at (2.5,2.65) {};
  \draw[opent/folded wire] (4,2.20) -- (5,3.10);
  \draw[opent/folded wire] (5,2.20) -- (4,3.10);
  \node[opent/crossing] at (4.5,2.65) {};

  \foreach \x in {0,5}
    \draw[opent/folded wire] (\x,3.10) -- (\x,4.00);
  \draw[opent/folded wire] (1,3.10) -- (2,4.00);
  \draw[opent/folded wire] (2,3.10) -- (1,4.00);
  \node[opent/crossing] at (1.5,3.55) {};
  \draw[opent/folded wire] (3,3.10) -- (4,4.00);
  \draw[opent/folded wire] (4,3.10) -- (3,4.00);
  \node[opent/crossing] at (3.5,3.55) {};

  \foreach \x in {0,1,4,5}
    \draw[opent/folded wire] (\x,4.00) -- (\x,4.90);
  \draw[opent/folded wire] (2,4.00) -- (3,4.90);
  \draw[opent/folded wire] (3,4.00) -- (2,4.90);
  \node[opent/crossing] at (2.5,4.45) {};

  \foreach \x in {0,...,5}
    \draw[opent/folded wire] (\x,4.90) -- (\x,5.65);

  \begin{scope}[on background layer]
    \draw[opent/region frame] (-0.30,0.58) rectangle (5.30,5.05);
  \end{scope}
  \node[opent/annotation,fill=white,inner sep=1pt,text=opentorange]
    at (2.5,5.05) {rectangular core};

  \foreach \x in {0,1,2,4,5}
    \node[opent/identity] at (\x,0.40) {};
  \node[opent/folded source] at (3,0.40) {$O$};
  \foreach \x/\lab in
    {0/A_0,1/A_1,2/A_2,3/B_0,4/B_1,5/B_2}
    \node[opent/annotation] at (\x,-0.10) {$\lab$};

  \foreach \x/\lab in
    {0/B_0,1/B_1,2/B_2,3/A_0,4/A_1,5/A_2}
    \node[opent/annotation] at (\x,6.02) {$\lab$};
\end{tikzpicture}
\caption{The isolated $3\times3$ folded rectangular core.  The
input order is $A_0,A_1,A_2,B_0,B_1,B_2$, with the folded local operator on
$B_0$.  Every $B$ strand crosses every $A$ strand once, and the common output
height displays the resulting order $B_0,B_1,B_2,A_0,A_1,A_2$.}
\label{fig:isolated-three-by-three-block}
\end{figure}
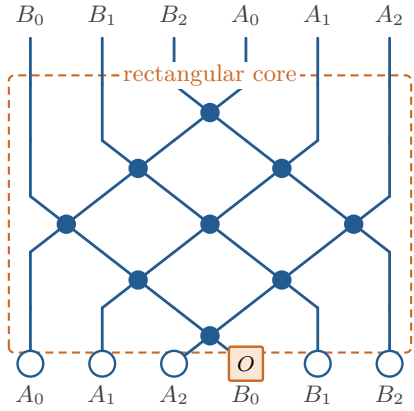

The two examples motivate the general definition.  Let $W_t$ denote the
ordinary unitary circuit associated with the complete $t\times t$ rectangular
braid.  Its input consists of two ordered blocks $A\mid B$, each containing
$t$ strands.  Every strand in $B$ crosses every strand in $A$ exactly once,
with $R$ acting at all $t^2$ crossings and with the crossing order inherited
from the causal network.  The output order is
$B_{\rm out}\mid A_{\rm out}$.

We now consider the Heisenberg evolution of a local operator under the
rectangular circuit $W_t$ alone.  We place the operator on $B_0$, take the
identity on every other input strand, and define the source operator and its
rectangular evolution by
\begin{align}
 S_t(O)&:=\id_A^{\otimes t}\otimes O_{B_0}\otimes
 \id_{B_1\cdots B_{t-1}},\nonumber\\
 \Phi_t(O)&:=W_t^\dagger S_t(O)W_t.
 \label{eq:yb-rectangle-operator}
\end{align}
These definitions are made directly in the ordinary operator circuit.  We
also need the same rectangular output in the folded picture.  Let
$\mathfrak f_{2t}=\mathfrak f^{\otimes 2t}$ be the sitewise folding of a
$2t$-site operator.  Since $\mathfrak f$ is a linear isomorphism, the formula
\begin{equation}
 \widehat\Phi_t(x_O)
 :=\mathfrak f_{2t}\bigl(\Phi_t(O)\bigr),
 \qquad x_O=\mathfrak f(O),
 \label{eq:folded-rectangle-output}
\end{equation}
defines a linear map
$\widehat\Phi_t:\mathcal W\to\mathcal W^{\otimes 2t}$ for every folded
one-site vector.  Equivalently, we can construct
$\widehat\Phi_t(x)$ directly in the folded picture.  We start from
\begin{equation*}
 \iota^{\otimes t}\otimes x\otimes\iota^{\otimes(t-1)}
\end{equation*}
in the input order $A\mid B$ and apply the complete $t\times t$ rectangular
network of folded crossings $\widehat R_{\rm H}$.  The output order is
$B_{\rm out}\mid A_{\rm out}$.  The equivalence of the two constructions
follows by applying the defining relation for $\widehat R_{\rm H}$ at every
crossing.  Thus $\Phi_t(O)$ and $\widehat\Phi_t(x_O)$ are the ordinary and
folded descriptions of one object.

The next lemma uses this construction for arbitrary $t$.  It shows that the
evolution under $W_t$ alone is sufficient to compute the operator-Schmidt rank
of the full Heisenberg operator and, after Hilbert--Schmidt normalization, its
complete nonzero operator-Schmidt spectrum.

\begin{lemma}[Rectangular-core identity]
\label{lem:yb-rectangle}
Let $R$ be any unitary two-site gate, and let $W_t$, $\Phi_t$, and
$\widehat\Phi_t$ be defined above.  For every one-site operator $O$ initially
at site zero and every $t\geq1$,
\begin{equation}
 \OSR_c(O(t))
 =\OSR_{B_{\rm out}\mid A_{\rm out}}(\Phi_t(O))
 =\SR_{B_{\rm out}\mid A_{\rm out}}(\widehat\Phi_t(x_O)).
 \label{eq:yb-rectangle-rank}
\end{equation}
More strongly, after Hilbert--Schmidt normalization, the complete nonzero
operator-Schmidt probability spectrum of the full fixed-cut Heisenberg
operator agrees with the Schmidt probability spectrum of
$\widehat\Phi_t(x_O)$, equivalently with the operator-Schmidt probability
spectrum of $\Phi_t(O)$.
\end{lemma}

\begin{proof}
Ordinary unitarity reduces the infinite circuit to the finite causal window
$[-2t,2t-1]$.  With the strand labels introduced above, the $R^\dagger$
branch $q_t$ contains exactly the crossings
\begin{equation}
 \{(B_k,A_j):0\leq k\leq j<2t\}.
 \label{eq:yb-causal-crossings}
\end{equation}
Splitting both source lists after $t$ entries separates these crossings into
a left triangle, the complete $t\times t$ rectangle, and a right triangle.
The rectangle is formed by $A_t,\ldots,A_{2t-1}$ and
$B_0,\ldots,B_{t-1}$, which are the two blocks on which $W_t$ acts.  Gates
acting on disjoint strand pairs commute, so the three pieces can be grouped
without changing the crossing order along any strand.  Hence
\begin{equation}
 q_t=(\ell_t\otimes r_t)W_t^\dagger,
 \label{eq:yb-factorization}
\end{equation}
where $W_t^\dagger$ acts on the two rectangular blocks, while $\ell_t$ and
$r_t$ act on opposite sides of the fixed cut.  With the identity factors on
the two outer blocks understood, it follows that
\begin{equation}
 q_tS_t(O)q_t^\dagger
 =(\ell_t\otimes r_t)\Phi_t(O)(\ell_t^\dagger\otimes r_t^\dagger).
 \label{eq:yb-local-factorization}
\end{equation}
The conjugation on the right-hand side is local with respect to the cut and
therefore preserves all operator-Schmidt singular values.  After removing
this conjugation, the two outer blocks contribute only identity factors.
These factors do not change the Schmidt rank and multiply all unnormalized
singular values by the same Hilbert--Schmidt norm.  This proves the first
equality in Eq.~\eqref{eq:yb-rectangle-rank} and the equality of the
normalized nonzero operator-Schmidt spectra.  Finally, sitewise folding maps
every operator-Schmidt decomposition of $\Phi_t(O)$ to an ordinary Schmidt
decomposition of $\widehat\Phi_t(x_O)$ with the same coefficients.  This
proves the second equality in Eq.~\eqref{eq:yb-rectangle-rank} and completes
the proof.

Once the rectangle has been isolated, we relabel its $A$ strands as
$A_0,\ldots,A_{t-1}$.  This relabeling plays no role in the argument; it only
fixes the labels used in the isolated-rectangle diagrams and in all later
proofs involving $W_t$.
\end{proof}

The geometric simplifications above use the concrete convention that the
local operator is placed at site zero and the cut lies immediately to its
left.  Placing the operator on the odd sublattice gives the spatially
reflected computation; for a gate without reflection symmetry, this also
reverses its leg ordering, \(R\mapsto R_{21}:=\Swap R\Swap\), without
changing the geometric argument.  If the cut is placed at a different fixed
distance from the operator, the same unitarity cancellations and grouping of
disjoint gates isolate an analogous, generally unequal-sided rectangular
core, with only the adjoining triangular pieces changed.

\section{Braid invariance of the output space}
\label{sec:centralizer-output-space}

In this section we establish a general restriction on the output of the
rectangular braid, which will be used later in
Section~\ref{sec:inv-polynomial-proof}.  The key point is that the rectangular
braid does not explore the full output space.  We may braid the outgoing $A$
strands further without changing $\widehat\Phi_t(x)$.  The same holds for the
outgoing strands $B_1,\ldots,B_{t-1}$.  Only braidings involving $B_0$ are
excluded, because this is the distinguished strand which carries the local
insertion.

These invariances confine the two output blocks to smaller subspaces.  Once
we know their dimensions, the Schmidt rank of $\widehat\Phi_t(x)$ cannot
exceed the dimension of either side.  The rectangular-core identity of
Lemma~\ref{lem:yb-rectangle} then converts this statement into an upper bound
on the exact OSR of the physical Heisenberg operator $O(t)$.  Thus it is
sufficient to control the intermediate rectangular output.

For a more transparent view of the mechanism, we provide a graphical
proof in Fig.~\ref{fig:centralizer-pull-through}.  We attach an additional crossing to
two outgoing strands and use the braid relation to pull it backwards through
the rectangle.  At the input, the crossing acts on two vectorized identities
and disappears by the folded unitarity rule
\eqref{eq:folded-unitarity}.

We now formalize this graphical proof by expressing the restriction in terms
of the centralizer of the braid-image algebra.  We use the folded space $\mathcal W$, crossing
$\widehat R_{\rm H}$, and identity vector $\iota$ introduced in
Subsection~\ref{subsec:folding-network}, together with the folded rectangular
output $\widehat\Phi_t$ defined in
Eq.~\eqref{eq:folded-rectangle-output}.  Define the algebra and its
centralizer by
\begin{align}
 \mathcal A_t(R)&=\operatorname{alg}\{R_1,\ldots,R_{t-1}\},
 \label{eq:centralizer-braid-algebra}\\
 \mathcal C_t(R)&=\mathcal A_t(R)'
 =\{X\in\End(V^{\otimes t}):[X,R_i]=0,\ 1\leq i<t\},
 \qquad c_t(R)=\dim\mathcal C_t(R).
 \label{eq:centralizer-definition}
\end{align}
Equation~\eqref{eq:centralizer-definition} defines a subspace of the ordinary
operator space.  To distinguish it from its folded counterpart, let
$\mathfrak f_t=\mathfrak f^{\otimes t}$ and define
\begin{equation}
 \widehat{\mathcal C}_t(R)
 :=\mathfrak f_t\bigl(\mathcal C_t(R)\bigr)
 \subset\mathcal W^{\otimes t},
 \qquad
 \dim\widehat{\mathcal C}_t(R)=c_t(R).
 \label{eq:centralizer-folded-space}
\end{equation}
The vectors in $\widehat{\mathcal C}_t(R)$ are precisely the folded tensors
invariant under every local folded crossing.  Indeed, the defining relation
for $\widehat R_{\rm H}$ maps this invariance to
$R_i^\dagger XR_i=X$, equivalently $[X,R_i]=0$.  For $t=0$, set
$\mathcal C_0(R)=\widehat{\mathcal C}_0(R)=\mathbb C$ and $c_0(R)=1$.

\begin{theorem}[Braid-invariant output space]
\label{thm:central-invariance}
For every Yang--Baxter gate, every $t\geq1$, and every $x\in\mathcal W$, the
folded rectangular output obeys
\begin{equation}
 \widehat\Phi_t(x)\in
 \bigl(\mathcal W\otimes\widehat{\mathcal C}_{t-1}(R)\bigr)
 \otimes\widehat{\mathcal C}_t(R)
 \label{eq:centralizer-output-space}
\end{equation}
across $B_{\rm out}\mid A_{\rm out}$.  Thus it is invariant under every
local folded crossing among the $A$ outputs and among the $B$ outputs other
than the strand carrying $x$.
\end{theorem}

\begin{proof}
The proof pulls an additional folded crossing backwards through the
rectangle.  For compactness, denote a crossing acting on the indicated pair
of adjacent output legs by
\begin{equation}
 \widehat R_{{\rm H},A_i}
 :=(\widehat R_{\rm H})_{A_iA_{i+1}},
 \qquad
 \widehat R_{{\rm H},B_j}
 :=(\widehat R_{\rm H})_{B_jB_{j+1}}.
 \label{eq:centralizer-output-crossing-notation}
\end{equation}
Consider first two adjacent $A$ outputs.  Repeated use of the
braid relation moves their crossing through the complete rectangular network
until it acts on the corresponding $A$ inputs.  Both inputs carry $\iota$,
so Eq.~\eqref{eq:folded-unitarity} removes the crossing.  Therefore
\begin{equation}
 \widehat R_{{\rm H},A_i}\widehat\Phi_t(x)
 =\widehat\Phi_t(x),
 \qquad 0\leq i<t-1.
 \label{eq:centralizer-naturality-A}
\end{equation}
The same argument applies to two adjacent $B$ outputs provided that neither
strand descends from the distinguished input $B_0$.  Their inputs also both
carry $\iota$, and hence
\begin{equation}
 \widehat R_{{\rm H},B_j}\widehat\Phi_t(x)
 =\widehat\Phi_t(x),
 \qquad 1\leq j<t-1.
 \label{eq:centralizer-naturality-B}
\end{equation}
There is no corresponding identity involving $B_0$, because that input
carries the arbitrary vector $x$.  Figure~\ref{fig:centralizer-pull-through}
shows the first pull-through operation for $t=3$.

\begin{figure}[t]
\centering
\begin{tikzpicture}[opent figure,x=0.58cm,y=0.75cm]
  \def\drawsmallwthree#1{%
    \begin{scope}[shift={(0,#1)}]
      \foreach \x in {0,1,4,5}
        \draw[opent/folded wire] (\x,0) -- (\x,0.8);
      \draw[opent/folded wire] (2,0) -- (3,0.8);
      \draw[opent/folded wire] (3,0) -- (2,0.8);
      \node[opent/crossing] at (2.5,0.4) {};

      \foreach \x in {0,5}
        \draw[opent/folded wire] (\x,0.8) -- (\x,1.6);
      \draw[opent/folded wire] (1,0.8) -- (2,1.6);
      \draw[opent/folded wire] (2,0.8) -- (1,1.6);
      \node[opent/crossing] at (1.5,1.2) {};
      \draw[opent/folded wire] (3,0.8) -- (4,1.6);
      \draw[opent/folded wire] (4,0.8) -- (3,1.6);
      \node[opent/crossing] at (3.5,1.2) {};

      \draw[opent/folded wire] (0,1.6) -- (1,2.4);
      \draw[opent/folded wire] (1,1.6) -- (0,2.4);
      \node[opent/crossing] at (0.5,2.0) {};
      \draw[opent/folded wire] (2,1.6) -- (3,2.4);
      \draw[opent/folded wire] (3,1.6) -- (2,2.4);
      \node[opent/crossing] at (2.5,2.0) {};
      \draw[opent/folded wire] (4,1.6) -- (5,2.4);
      \draw[opent/folded wire] (5,1.6) -- (4,2.4);
      \node[opent/crossing] at (4.5,2.0) {};

      \foreach \x in {0,5}
        \draw[opent/folded wire] (\x,2.4) -- (\x,3.2);
      \draw[opent/folded wire] (1,2.4) -- (2,3.2);
      \draw[opent/folded wire] (2,2.4) -- (1,3.2);
      \node[opent/crossing] at (1.5,2.8) {};
      \draw[opent/folded wire] (3,2.4) -- (4,3.2);
      \draw[opent/folded wire] (4,2.4) -- (3,3.2);
      \node[opent/crossing] at (3.5,2.8) {};

      \foreach \x in {0,1,4,5}
        \draw[opent/folded wire] (\x,3.2) -- (\x,4.0);
      \draw[opent/folded wire] (2,3.2) -- (3,4.0);
      \draw[opent/folded wire] (3,3.2) -- (2,4.0);
      \node[opent/crossing] at (2.5,3.6) {};
    \end{scope}%
  }

  \begin{scope}[shift={(0,0)}]
    \node[opent/panel label] at (2.5,6.05)
      {(a) Cross the $A$ outputs};
    \drawsmallwthree{0.55}
    \foreach \x in {0,1,2,5}
      \draw[opent/folded wire] (\x,4.55) -- (\x,5.35);
    \draw[opent/folded wire] (3,4.55) -- (4,5.35);
    \draw[opent/folded wire] (4,4.55) -- (3,5.35);
    \node[opent/crossing] at (3.5,4.95) {};
    \begin{scope}[on background layer]
      \draw[opent/region frame] (-0.25,0.72) rectangle (5.25,4.42);
    \end{scope}
    \foreach \x in {0,1,2,4,5}
      \node[opent/identity] at (\x,0.55) {};
    \node[opent/folded source] at (3,0.55) {$O$};
    \foreach \x/\lab in
      {0/A_0,1/A_1,2/A_2,3/B_0,4/B_1,5/B_2}
      \node[opent/annotation] at (\x,-0.15) {$\lab$};
  \end{scope}

  \draw[opent/flow arrow] (5.75,2.90) -- (7.85,2.90)
    node[midway,above,font=\scriptsize,align=center] {braid\\relation};

  \begin{scope}[shift={(8.6,0)}]
    \node[opent/panel label] at (2.5,6.05)
      {(b) Pull the crossing below};
    \foreach \x in {2,3,4,5}
      \draw[opent/folded wire] (\x,0.35) -- (\x,1.15);
    \draw[opent/folded wire] (0,0.35) -- (1,1.15);
    \draw[opent/folded wire] (1,0.35) -- (0,1.15);
    \node[opent/crossing] at (0.5,0.75) {};
    \drawsmallwthree{1.15}
    \foreach \x in {0,...,5}
      \draw[opent/folded wire] (\x,5.15) -- (\x,5.35);
    \begin{scope}[on background layer]
      \draw[opent/region frame] (-0.25,1.32) rectangle (5.25,5.02);
    \end{scope}
    \foreach \x in {0,1,2,4,5}
      \node[opent/identity] at (\x,0.35) {};
    \node[opent/folded source] at (3,0.35) {$O$};
    \foreach \x/\lab in
      {0/A_0,1/A_1,2/A_2,3/B_0,4/B_1,5/B_2}
      \node[opent/annotation] at (\x,-0.35) {$\lab$};
  \end{scope}

  \draw[opent/flow arrow] (14.35,2.90) -- (16.45,2.90)
    node[midway,above,font=\scriptsize] {unitarity};

  \begin{scope}[shift={(17.2,0)}]
    \node[opent/panel label] at (2.5,6.05)
      {(c) Remove the crossing};
    \foreach \x in {0,...,5}
      \draw[opent/folded wire] (\x,0.35) -- (\x,1.15);
    \drawsmallwthree{1.15}
    \foreach \x in {0,...,5}
      \draw[opent/folded wire] (\x,5.15) -- (\x,5.35);
    \begin{scope}[on background layer]
      \draw[opent/region frame] (-0.25,1.32) rectangle (5.25,5.02);
    \end{scope}
    \foreach \x in {0,1,2,4,5}
      \node[opent/identity] at (\x,0.35) {};
    \node[opent/folded source] at (3,0.35) {$O$};
    \foreach \x/\lab in
      {0/A_0,1/A_1,2/A_2,3/B_0,4/B_1,5/B_2}
      \node[opent/annotation] at (\x,-0.35) {$\lab$};
  \end{scope}
\end{tikzpicture}
\caption{Graphical invariance argument for the $A$ strands leaving the
rectangular block.  Panel (a)
acts on the output with an additional folded crossing between $A_0$ and $A_1$.  Repeated
use of the braid relation moves it through the rectangular crossing block to
the corresponding inputs in panel (b).  Those inputs are both vectorized
identities, so folded unitarity removes
the crossing and gives panel (c), which is the original intermediate tensor.
The same argument applies to every adjacent pair of $A$ strands.  It also applies
to adjacent $B$ strands whose inputs are identities, but not to $B_0$, which
carries the local operator.}
\label{fig:centralizer-pull-through}
\end{figure}
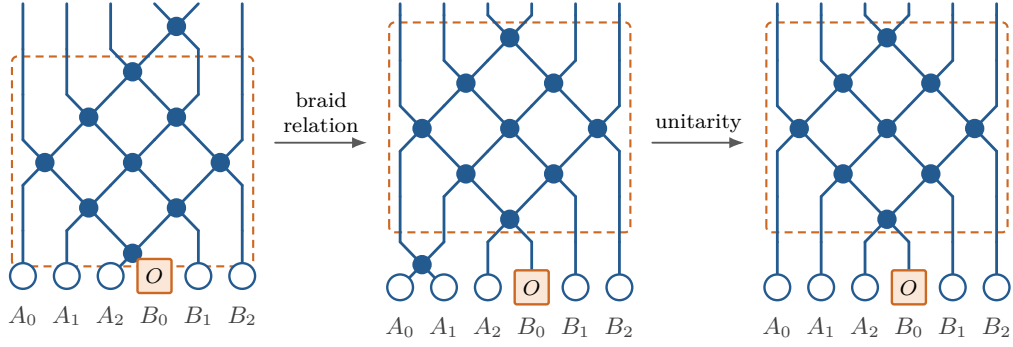

Equation~\eqref{eq:centralizer-naturality-A} places the complete $A$ output
block in $\widehat{\mathcal C}_t(R)$.  Equation~
\eqref{eq:centralizer-naturality-B} places the identity-derived block
$B_1,\ldots,B_{t-1}$ in $\widehat{\mathcal C}_{t-1}(R)$, while the $B_0$ leg
remains in the unrestricted space $\mathcal W$.  Since the two sets of folded
crossings act on separate output blocks, imposing both invariances gives
their tensor-product intersection, which is precisely
Eq.~\eqref{eq:centralizer-output-space}.
\end{proof}

\begin{corollary}[Operator-Schmidt-rank bound]
\label{cor:centralizer-osr-bound}
For every Yang--Baxter gate and every one-site operator initially at site
zero,
\begin{equation}
 \OSR_c(O(t))\leq
 \min\{c_t(R),D^2c_{t-1}(R)\},\qquad t\geq1.
 \label{eq:centralizer-bound}
\end{equation}
\end{corollary}

\begin{proof}
Lemma~\ref{lem:yb-rectangle} identifies the OSR with the Schmidt rank of
$\widehat\Phi_t(x_O)$.  Theorem~\ref{thm:central-invariance} confines
this vector to the two folded invariant sectors in
\eqref{eq:centralizer-output-space}.  Their dimensions are
$D^2c_{t-1}(R)$ and $c_t(R)$, respectively, so its Schmidt rank cannot exceed
the smaller of the two.
\end{proof}

The folded invariant sector need not coincide with the tensors actually
generated by evolving a one-site operator.  Equation~
\eqref{eq:centralizer-folded-space} maps the operator centralizer onto the
full symmetry-allowed folded sector, while the physical dynamics may explore
only a small part of it.

\section{Elementary bounded-rank constructions}
\label{sec:elementary-bounded-rank}

In this section we consider two simple constructions for which a
time-independent OSR bound can be established directly, without introducing
the techniques of later sections.
In both cases we use the rectangular-core identity of
Lemma~\ref{lem:yb-rectangle}, which identifies the physical fixed-cut OSR
with the operator Schmidt rank of the rectangular braid.

\subsection{Phase-dressed exchanges}
\label{subsec:separated-phase-dressed}

Let \(X\) be the set of \(D\) basis colours, and let
\(f,g\in\operatorname{Sym}(X)\).  Denote their permutation unitaries by
\begin{equation}
 F\ket{a}=\ket{f(a)},
 \qquad
 G\ket{a}=\ket{g(a)}.
 \label{eq:separated-one-site-rotations}
\end{equation}
For unit-modulus numbers \(q_{ab}\), define the diagonal unitary
\begin{equation}
 Q_q=\sum_{a,b\in X}q_{ab}E_{aa}\otimes E_{bb},
 \qquad E_{ab}=\ket{a}\!\bra{b},
 \label{eq:separated-diagonal-dressing}
\end{equation}
and the phase-dressed exchange
\begin{align}
 R_{f,g,q}
 &=(F\otimes G)\Swap Q_q,
 \nonumber\\
 R_{f,g,q}\ket{a,b}
 &=q_{ab}\ket{f(b),g(a)}.
 \label{eq:separated-gate}
\end{align}
We call this a \emph{separated phase-dressed exchange}.  The permutations \(f\) and \(g\)
need not be equal.

The SWAP gate is obtained by taking \(f=g=\id_X\) and \(q_{ab}=1\), while
arbitrary unit-modulus \(q_{ab}\) with \(f=g=\id_X\) give the phase-dressed
SWAP gates
\begin{equation}
 R_q:=\Swap Q_q.
 \label{eq:separated-phase-dressed-swap}
\end{equation}

For completeness, let us record when \eqref{eq:separated-gate} is a
Yang--Baxter gate.  Its classical support is
\begin{equation}
 r_{f,g}(a,b)=\bigl(f(b),g(a)\bigr).
 \label{eq:separated-support}
\end{equation}
The two sides of the set-theoretical braid relation send
\((a,b,c)\), respectively, to
\begin{align}
 r_{12}r_{23}r_{12}(a,b,c)
 &=\bigl(f^2(c),gf(b),g^2(a)\bigr),
 \nonumber\\
 r_{23}r_{12}r_{23}(a,b,c)
 &=\bigl(f^2(c),fg(b),g^2(a)\bigr).
 \label{eq:separated-support-braid-check}
\end{align}
Hence the support satisfies the braid relation if and only if
\begin{equation}
 fg=gf.
 \label{eq:separated-commuting-rotations}
\end{equation}
Assuming this condition, comparison of the scalar factors in the two braid
words gives the remaining phase constraint
\begin{equation}
 q_{ab}\,q_{g(a),c}\,q_{f(b),f(c)}
 =
 q_{bc}\,q_{a,f(c)}\,q_{g(a),g(b)}
 \qquad(a,b,c\in X).
 \label{eq:separated-phase-cocycle}
\end{equation}
Equations~\eqref{eq:separated-commuting-rotations} and
\eqref{eq:separated-phase-cocycle} are therefore equivalent to the braid
relation for \(R_{f,g,q}\).  For \(f=g=\id_X\), the three phases on the two
sides are the same factors in a different order, so
\eqref{eq:separated-phase-cocycle} is automatic, consistently with the
arbitrary phases allowed in the phase-dressed SWAP gates above.

We put forward the following statement: in every separated phase-dressed
exchange circuit, the OSR of an evolved one-site operator remains bounded
uniformly in time.  This statement is independent of the Yang--Baxter
relation.  Neither Eq.~\eqref{eq:separated-commuting-rotations} nor
Eq.~\eqref{eq:separated-phase-cocycle} enters the proof; these conditions are
needed only to embed the separated family into the class of Yang--Baxter
gates.  The statement will be proved in
Theorem~\ref{thm:separated-uniform-rank} below.

Before turning to the proof,
let us illustrate the physical mechanism.  In these circuits each matrix-unit
component of the one-site operator propagates ballistically.  A diagonal
one-site operator is simply transported, with its diagonal entries permuted.
For an off-diagonal component, crossings with other strands only leave behind
a string of diagonal one-site operators.  The component therefore remains a
single product string, so its Schmidt rank does not increase.

For the circuit with the bare SWAP, this statement is immediate: every
one-site operator propagates ballistically without changing and leaves no
trail.  Let us then illustrate the more general mechanism for the
phase-dressed SWAP gate \(R_q\) of
Eq.~\eqref{eq:separated-phase-dressed-swap}, where the diagonal trail is
nontrivial.  We now evaluate the rectangular-braid evolution \(\Phi_t\)
defined in Eq.~\eqref{eq:yb-rectangle-operator} for this gate.

For \(x\ne y\), define the one-site diagonal operator
\begin{equation}
 \Delta_{xy}
 :=
 \sum_{b\in X}q_{yb}\overline{q_{xb}}E_{bb}.
 \label{eq:separated-phase-trail}
\end{equation}
Direct conjugation gives
\begin{equation}
 R_q^\dagger(\id\otimes E_{xy})R_q
 =E_{xy}\otimes\Delta_{xy}.
 \label{eq:separated-first-crossing}
\end{equation}
The operator \(\Delta_{xy}\) which appears here is the first member of the
string of diagonal one-site operators produced along the ballistic
trajectory.

To obtain the complete diagonal trail, let the same off-diagonal matrix unit
cross \(n\) identity strands consecutively in the same orientation.  Repeated use of
\eqref{eq:separated-first-crossing} gives, up to the ordering of the output
sites,
\begin{equation}
 \underbrace{\id\otimes\id\otimes\cdots\otimes\id}
 _{n\text{ identity strands}}\otimes E_{xy}
 \longmapsto
 E_{xy}\otimes
 \underbrace{\Delta_{xy}\otimes\Delta_{xy}\otimes\cdots\otimes\Delta_{xy}}
 _{n\text{ diagonal factors}}.
 \label{eq:separated-consecutive-crossings}
\end{equation}
This formula is sufficient for the complete rectangular braid: the
off-diagonal factor crosses the identity strands one after another, while
every remaining crossing is between diagonal factors and merely exchanges
them.

This calculation establishes the product-string mechanism for the
phase-deformed swap.  We now state the bound for the full separated family.

\begin{theorem}[Uniform rank for separated phase-dressed exchanges]
\label{thm:separated-uniform-rank}
Let \(R_{f,g,q}\) be a gate of the separated form
\eqref{eq:separated-gate}; it need not satisfy the Yang--Baxter relation.
For every one-site operator \(O\), every
\(t\geq0\), and every spatial cut \(c\),
\begin{equation}
 \OSR_c(O(t))\leq D(D-1)+1.
 \label{eq:separated-uniform-bound}
\end{equation}
Consequently, for every nonzero \(O\) and every \(\alpha\geq0\),
\begin{equation}
 S_\alpha^{\rm op}(O(t))
 \leq\log\!\bigl(D(D-1)+1\bigr).
 \label{eq:separated-uniform-entropy}
\end{equation}
Both bounds are uniform in time.
\end{theorem}

\begin{proof}
The proof has two steps.  We first extend the product-string mechanism from
the phase-deformed swap to the full gate \(R_{f,g,q}\), and then expand the
one-site source into at most \(D(D-1)+1\) such strings.

Let \(D_0\) be any diagonal one-site operator.  Direct conjugation by
\(R_{f,g,q}\) has the structural form
\begin{align*}
 R_{f,g,q}^{\dagger}(D_0\otimes E_{xy})R_{f,g,q}
 &=E_{g^{-1}(x),g^{-1}(y)}\otimes D'_{xy},\\
 R_{f,g,q}^{\dagger}(E_{xy}\otimes D_0)R_{f,g,q}
 &=D''_{xy}\otimes E_{f^{-1}(x),f^{-1}(y)},
\end{align*}
where \(D'_{xy}\) and \(D''_{xy}\) are again diagonal one-site
operators.  The phase ratios only change their diagonal entries.  If both
incoming operators are diagonal, the phases cancel and
\[
 R_{f,g,q}^{\dagger}(D_1\otimes D_2)R_{f,g,q}
 =(G^\dagger D_2G)\otimes(F^\dagger D_1F).
\]
It follows that a product string with exactly one off-diagonal matrix unit
and diagonal factors elsewhere remains a single product string after every
brickwork layer.  The off-diagonal factor moves from strand to strand and
its labels are rotated, while the other factors remain diagonal.  No sum of
product strings is created.  The evolution of any off-diagonal source
\(E_{xy}\), with \(x\ne y\), therefore has OSR one across every cut.

The diagonal part of the source also contributes only one product operator.
Indeed, a diagonal one-site operator is transported along a single
trajectory, with its diagonal entries permuted by \(F\) and \(G\), and no
phase trail is produced.  Finally, an arbitrary one-site operator is the sum
of its diagonal part and at most \(D(D-1)\) off-diagonal matrix units.  The
subadditivity of the OSR gives Eq.~\eqref{eq:separated-uniform-bound}, and
Eq.~\eqref{eq:setup-entropy-rank} gives
Eq.~\eqref{eq:separated-uniform-entropy}.

The bounded-rank mechanism uses only the separated form of the gate.  The
conditions \eqref{eq:separated-commuting-rotations} and
\eqref{eq:separated-phase-cocycle} are needed to place this family among the
Yang--Baxter gates, but they do not enter the rank estimate itself.
\end{proof}

A concrete member of this family was first analyzed from the viewpoint of
operator entanglement in Ref.~\cite{BertiniKosProsenII}, as a dual-unitary
Trotterization of the XXZ chain.  In the computational basis
$(\ket{00},\ket{01},\ket{10},\ket{11})$, its two-site gate is, up to an
overall phase,
\begin{equation}
 R_K=
 \begin{pmatrix}
  1&0&0&0\\
  0&0&e^{2iK}&0\\
  0&e^{2iK}&0&0\\
  0&0&0&1
 \end{pmatrix},
 \qquad K\in\mathbb R.
 \label{eq:qubit-dual-unitary-XXZ}
\end{equation}
This is a phase-dressed swap and hence a qubit Yang--Baxter gate.
Reference~\cite{BertiniKosProsenII} showed that the operator
entanglement of every evolved one-site operator remains bounded.  This result
is consistent with Theorem~\ref{thm:separated-uniform-rank}, which gives
\(\OSR_c(O(t))\leq3\) for every evolved one-site operator in this circuit.
The same gate was later considered in Ref.~\cite{GiudiceEtAl2022}, where the
focus was instead on the temporal entanglement of the influence matrix.

\subsection{Controlled-swap Yang--Baxter gates}
\label{sec:controlled-swap}

We now consider a family in which one incoming state is transmitted
ballistically through a crossing and, at the same time, selects a unitary
rotation of the other state.  These gates are simple enough that the
rectangular circuit of Section~\ref{sec:rectangular-reduction} can be
evaluated directly.  The resulting bounds are uniform in time and do not
require a classification of the allowed one-site unitaries.

There are two possible ways to construct such a controlled-swap gate,
depending on whether the incoming state transmitted without modification
travels to the left or to the right.  The two choices are related by spatial
reflection.  If we fixed one gate orientation, a complete analysis of
one-site operator entanglement would have to treat the two possible
placements of the source at \(t=0\), on the even and odd sublattices.  We use
the equivalent convention of keeping the source on the distinguished strand \(B_0\) of
Eq.~\eqref{eq:yb-rectangle-operator} and carrying out the proof for both gate
orientations.

Let \(X\) be a set of \(D\) basis labels and let
\(\{u_a:a\in X\}\subset U(V)\) be a collection of one-site unitaries.  We
define the right-controlled and left-controlled gates by
\begin{align}
 R_{\rm rc}\ket{a,b}
 &=\ket b\otimes u_b\ket a,
 &
 R_{\rm lc}\ket{a,b}
 &=u_a\ket b\otimes\ket a.
 \label{eq:controlled-swap-two-orientations}
\end{align}
Thus, for \(R_{\rm rc}\), the right incoming label \(b\) is transmitted to
the left output and controls the rotation of \(a\).  For \(R_{\rm lc}\), the
left incoming label \(a\) is transmitted to the right output and controls
the rotation of \(b\).  The two gates obey
\begin{equation}
 R_{\rm lc}=\Swap R_{\rm rc}\Swap.
 \label{eq:controlled-swap-reflection}
\end{equation}
Both gates are unitary and dual-unitary for arbitrary choices of the
unitaries \(u_a\).

\subsubsection{The braid constraint}
\label{subsec:controlled-swap-braid}

The braid relation has a particularly transparent form in this family.
For \(R_{\rm lc}\), applying the two braid words to
\(\ket{a,b,c}\) gives
\begin{align}
 R_{{\rm lc},12}R_{{\rm lc},23}R_{{\rm lc},12}\ket{a,b,c}
 &=
 \sum_x (u_a)_{xb}\,
 u_xu_a\ket c\otimes\ket{x,a},
 \nonumber\\
 R_{{\rm lc},23}R_{{\rm lc},12}R_{{\rm lc},23}\ket{a,b,c}
 &=
 \sum_x (u_a)_{xb}\,
 u_au_b\ket c\otimes\ket{x,a}.
 \label{eq:controlled-swap-two-braid-words}
\end{align}
Comparing the components proportional to \((u_a)_{xb}\) gives
\begin{equation}
 (u_a)_{xb}\bigl(u_xu_a-u_au_b\bigr)=0,
 \qquad a,b,x\in X.
 \label{eq:controlled-swap-braid-condition}
\end{equation}
The reflected gate \(R_{\rm rc}\) satisfies the same condition, with the
label of the controlling input renamed.  Hence either gate is a
Yang--Baxter gate if and only if
\begin{equation}
 (u_a)_{xb}\ne0
 \quad\Longrightarrow\quad
 u_x=u_au_bu_a^{-1}.
 \label{eq:controlled-swap-conjugation-rule}
\end{equation}

This is the group-type compatibility condition of
Ref.~\cite{GalindoRowell2014}.  Let us define
\begin{equation}
 \mathcal S=\{u_a:a\in X\}
 \label{eq:controlled-swap-support-set}
\end{equation}
be the set of distinct unitaries occurring among the labels.

The following lemma was proved by Galindo and Rowell in the language of
group-type braided vector spaces and Yetter--Drinfeld modules; see
Proposition~4.2 and Subsection~4.2 of Ref.~\cite{GalindoRowell2014}.  For
completeness, Appendix~\ref{app:controlled-swap-conjugation} repeats the
direct calculation in our notation and braid convention.

\begin{lemma}[Finite conjugation memory]
\label{lem:controlled-swap-finite-conjugation-memory}
Suppose that the two braid words in
Eq.~\eqref{eq:controlled-swap-two-braid-words} agree for every
\(a,b,c\in X\), and let \(G=\langle u_a:a\in X\rangle\).  Then
\begin{equation}
 g\mathcal Sg^{-1}=\mathcal S,
 \qquad g\in G.
 \label{eq:controlled-swap-conjugation-memory}
\end{equation}
\end{lemma}

We stress that Eq.~\eqref{eq:controlled-swap-conjugation-memory}
follows from the braid relation, but the converse is false.  This
conjugation property will play an important role in what follows.

\subsubsection{Right-controlled gates}
\label{subsec:controlled-swap-right}

We first treat \(R_{\rm rc}\), for which the right incoming label is the
ballistic control.  Let
\begin{equation*}
 A=(a_0,\ldots,a_{t-1}),
 \qquad
 B=(b_0,\ldots,b_{t-1})
\end{equation*}
be the two input words of the complete rectangular braid.  Every \(B\)
strand passes unchanged through the complete \(A\) block.  Each \(A\)
strand, on the other hand, encounters the same ordered sequence of
controlled rotations.  Consequently,
\begin{align}
 W_t^{\rm rc}\ket{A,B}
 &=\ket B\otimes g(B)^{\otimes t}\ket A,
 &
 g(B)&=u_{b_{t-1}}\cdots u_{b_1}u_{b_0}.
 \label{eq:controlled-swap-right-block-action}
\end{align}
After restoring the block order with \(\Swap_{\rm blk}\), this is the
controlled unitary
\begin{equation}
 \widehat W_t^{\rm rc}
 :=\Swap_{\rm blk}W_t^{\rm rc}
 =\sum_{B\in X^t}g(B)^{\otimes t}\otimes\ket B\!\bra B.
 \label{eq:controlled-swap-right-controlled-block}
\end{equation}

The block exchange sends the source \(S_t(O)\) of
Eq.~\eqref{eq:yb-rectangle-operator} to \(O\) on the first site of the
\(A\) block.  Conjugating it by
Eq.~\eqref{eq:controlled-swap-right-controlled-block} therefore gives,
with identity factors on the remaining \(A\) sites understood,
\begin{equation}
 \Phi_t^{\rm rc}(O)
 =\sum_{B\in X^t}
 \bigl[g(B)^{\dagger}Og(B)\bigr]_{A_0}
 \otimes\ket B\!\bra B.
 \label{eq:controlled-swap-right-evolved-source}
\end{equation}
Equation~\eqref{eq:controlled-swap-right-evolved-source} already gives the
desired time-independent upper bound.  Define
\[
 \mathcal V_t(O)
 :=\operatorname{span}
 \bigl\{g(B)^{\dagger}Og(B):B\in X^t\bigr\}
 \subseteq\operatorname{End}(V),
\]
and put \(r_t=\dim\mathcal V_t(O)\).  If
\(K_1,\ldots,K_{r_t}\) is a basis of \(\mathcal V_t(O)\), we may write
\[
 g(B)^{\dagger}Og(B)=\sum_{\mu=1}^{r_t}c_\mu(B)K_\mu.
\]
Regrouping the terms in
Eq.~\eqref{eq:controlled-swap-right-evolved-source} then gives
\[
 \Phi_t^{\rm rc}(O)
 =\sum_{\mu=1}^{r_t}[K_\mu]_{A_0}\otimes
 \left(\sum_{B\in X^t}c_\mu(B)\ket B\!\bra B\right).
\]
Thus the operator Schmidt rank is at most
\(r_t\leq\dim\operatorname{End}(V)=D^2\).  The conjugates
\(g(B)^{\dagger}Og(B)\) span \(\mathcal V_t(O)\), so the coefficient matrix
\(c_\mu(B)\) has row rank \(r_t\).  Since the projectors
\(\ket B\!\bra B\) are linearly independent, the right tensor factors in
the last display are also linearly independent.  This gives
\begin{equation}
 \OSR\bigl(\Phi_t^{\rm rc}(O)\bigr)
 =r_t
 \leq D^2.
 \label{eq:controlled-swap-right-rank}
\end{equation}
If \(O\) is traceless, all its conjugates remain in the traceless subspace
and \(D^2\) can be replaced by \(D^2-1\).

It is worth stressing what Eq.~\eqref{eq:controlled-swap-right-rank} does
not count.  The number of distinct words \(g(B)\) may grow rapidly with
\(t\), and may even grow exponentially.  Nevertheless, every word appears
only through its adjoint action on one fixed one-site operator.  All these
conjugates belong to the \(D^2\)-dimensional space
\(\operatorname{End}(V)\), independently of the number of group words.  In
particular, the bound in Eq.~\eqref{eq:controlled-swap-right-rank} does not
use the braid relation.  However, the Yang--Baxter structure will be used for
the left-controlled gates in
Subsubsection~\ref{subsec:controlled-swap-left}.

\subsubsection{Left-controlled gates}
\label{subsec:controlled-swap-left}

The reflected orientation is slightly less immediate.  Now every \(A\)
label is transmitted unchanged, while it applies a rotation to each
crossing \(B\) strand.  The rectangular action is
\begin{align}
 W_t^{\rm lc}\ket{A,B}
 &=f(A)^{\otimes t}\ket B\otimes\ket A,
 &
 f(A)&=u_{a_0}u_{a_1}\cdots u_{a_{t-1}}.
 \label{eq:controlled-swap-left-block-action}
\end{align}
Equivalently, after restoring the input block order,
\begin{equation}
 \widehat W_t^{\rm lc}
 =\sum_{A\in X^t}\ket A\!\bra A\otimes f(A)^{\otimes t}.
 \label{eq:controlled-swap-left-controlled-block}
\end{equation}

Write \(A=(a,\alpha)\), where
\(\alpha=(a_1,\ldots,a_{t-1})\), and define
\begin{equation}
 h(\alpha)=u_{a_1}\cdots u_{a_{t-1}}.
 \label{eq:controlled-swap-left-background-word}
\end{equation}
When the local source connects \(a\) on the ket history to \(c\) on the
bra history, all other control labels agree.  Thus
\begin{equation}
 f(a,\alpha)^{\dagger}f(c,\alpha)
 =h(\alpha)^{\dagger}u_a^{\dagger}u_ch(\alpha).
 \label{eq:controlled-swap-left-relative-word}
\end{equation}
Expanding \(O=\sum_{a,c}O_{ac}\ket a\!\bra c\) gives the exact rectangular
expression
\begin{align}
 \Phi_t^{\rm lc}(O)
 ={}&\sum_{a,c\in X}\sum_{\alpha\in X^{t-1}}
 O_{ac}\ket{a,\alpha}\!\bra{c,\alpha}
 \nonumber\\[-2mm]
 &\hspace{22mm}\otimes
 \left[h(\alpha)^{\dagger}u_a^{\dagger}u_ch(\alpha)\right]^{\otimes t}.
 \label{eq:controlled-swap-left-evolved-source}
\end{align}

Without further assumptions, Eq.~\eqref{eq:controlled-swap-left-evolved-source}
already gives a polynomial bound.  Every operator in the second tensor
factor has the form \(q^{\otimes t}\), with
\(q\in\operatorname{End}(V)\).  Their span is contained in the symmetric
tensor power of the \(D^2\)-dimensional one-site operator space, and hence
\begin{equation}
 \OSR\bigl(\Phi_t^{\rm lc}(O)\bigr)
 \leq
 \binom{t+D^2-1}{D^2-1}.
 \label{eq:controlled-swap-left-no-braid-bound}
\end{equation}
We record this observation only to distinguish the two orientations.  For a
Yang--Baxter gate, the same expression yields a much stronger conclusion.

Indeed, Eq.~\eqref{eq:controlled-swap-conjugation-memory} implies
that for every background word \(h(\alpha)\) there are
\(u_x,u_y\in\mathcal S\) such that
\begin{equation}
 h(\alpha)^{\dagger}u_ah(\alpha)=u_x,
 \qquad
 h(\alpha)^{\dagger}u_ch(\alpha)=u_y.
 \label{eq:controlled-swap-left-two-conjugates}
\end{equation}
Therefore every relative unitary in
Eq.~\eqref{eq:controlled-swap-left-evolved-source} belongs to the fixed
finite set
\begin{equation}
 \mathcal Q
 :=\{u_x^{\dagger}u_y:u_x,u_y\in\mathcal S\}.
 \label{eq:controlled-swap-relative-set}
\end{equation}
Grouping together all histories with the same value \(q\in\mathcal Q\)
turns Eq.~\eqref{eq:controlled-swap-left-evolved-source} into a product
decomposition
\begin{equation}
 \Phi_t^{\rm lc}(O)
 =\sum_{q\in\mathcal Q}K_{q,t}(O)\otimes q^{\otimes t},
 \label{eq:controlled-swap-left-finite-decomposition}
\end{equation}
where \(K_{q,t}(O)\) is the sum of the corresponding operators
\(O_{ac}\ket{a,\alpha}\!\bra{c,\alpha}\) on the \(A\) block.  Its explicit
form is immaterial for the rank estimate.  If
\(s=|\mathcal S|\leq D\), all \(s\) diagonal pairs \(u_x^{\dagger}u_x\)
give the identity, while the ordered pairs with \(x\ne y\) give at most
\(s(s-1)\) additional values.  Consequently,
\begin{equation}
 |\mathcal Q|
 \leq s(s-1)+1
 \leq D^2-D+1.
 \label{eq:controlled-swap-relative-count}
\end{equation}

We summarize the result.

\begin{theorem}[Uniform rank for controlled-swap Yang--Baxter gates]
\label{thm:controlled-swap-uniform-rank}
Let \(R\) be a Yang--Baxter gate of either form in
Eq.~\eqref{eq:controlled-swap-two-orientations}, and let \(O\) be an
arbitrary one-site operator initially placed at site zero.  Then, for every
\(t\geq1\),
\begin{align}
 R=R_{\rm rc}
 &\quad\Longrightarrow\quad
 \OSR_c(O(t))\leq D^2,
 \label{eq:controlled-swap-theorem-right}\\
 R=R_{\rm lc}
 &\quad\Longrightarrow\quad
 \OSR_c(O(t))\leq|\mathcal Q|
 \leq D^2-D+1.
 \label{eq:controlled-swap-theorem-left}
\end{align}
For \(R_{\rm rc}\) and traceless \(O\), the first bound improves to
\(D^2-1\).  In particular, every nonnegative-R\'enyi operator entropy is
bounded uniformly in time.
\end{theorem}

\begin{proof}
Equations~\eqref{eq:controlled-swap-right-rank} and
\eqref{eq:controlled-swap-left-finite-decomposition} give the respective
rank bounds for the rectangular cores.  Lemma~\ref{lem:yb-rectangle}
identifies these ranks with the physical fixed-cut ranks
\(\OSR_c(O(t))\).  Equation~\eqref{eq:setup-entropy-rank} then gives the
entropy statement.
\end{proof}

\section{Yang--Baxter gates on qubits}
\label{sec:qubit-classification}

In this section we prove that the OSR of every time-evolved one-site operator
in a qubit Yang--Baxter circuit remains bounded uniformly in time.  The proof
has two steps.  We first restate Dye's classification of qubit Yang--Baxter
gates~\cite{Dye2003} in Theorem~\ref{thm:qubit-Dye-classification}, using the
braid convention of this paper.  We then reduce the classified gates to three
elementary types of dynamics: separated phase-dressed exchanges, one
Clifford gate, and the scalar identity.

\begin{theorem}[Dye's classification of qubit Yang--Baxter gates
\cite{Dye2003}]
\label{thm:qubit-Dye-classification}
Let \(R\in U(\mathbb C^2\otimes\mathbb C^2)\) satisfy the braid relation.
In the ordered basis
\(\{\ket{00},\ket{01},\ket{10},\ket{11}\}\), every such gate belongs to at
least one of the following five, not necessarily disjoint, families.  Each
family has the form
\begin{equation}
 R=k(Q\otimes Q)\widehat S
 (Q^{-1}\otimes Q^{-1})\Swap,
 \qquad |k|=1,
 \qquad
 Q=\begin{pmatrix}a&b\\c&d\end{pmatrix}\in GL(2,\mathbb C).
 \label{eq:app-Dye-form}
\end{equation}
The algebraic Yang--Baxter representative \(\widehat S\) and the
corresponding restrictions are as follows.  The parametrization assumes
\(d\ne0\) in families 1--4 and \(a\ne0\) in family 3; the family-2
constraint then guarantees that its remaining denominators are nonzero.
\begin{enumerate}
\item In family 1,
\begin{equation}
 \widehat S_1=
 \begin{pmatrix}
  1&0&0&0\\
  0&p&0&0\\
  0&0&q&0\\
  0&0&0&r
 \end{pmatrix},
 \qquad |p|=|q|=|r|=1,
 \qquad c=-\frac{a\overline b}{\overline d}.
 \label{eq:qubit-Dye-family-one}
\end{equation}

\item In family 2,
\begin{equation}
 \widehat S_2=
 \begin{pmatrix}
  0&0&0&p\\
  0&0&1&0\\
  0&1&0&0\\
  q&0&0&0
 \end{pmatrix},
 \qquad |pq|=1,
 \qquad c\ne-\frac{a\overline b}{\overline d},
 \label{eq:qubit-Dye-family-two}
\end{equation}
where
\begin{align}
 p&=
 \frac{(|b|^2+|d|^2)(\overline a b+\overline c d)}
      {(|a|^2+|c|^2)(a\overline b+c\overline d)},
 \nonumber\\
 q&=
 \frac{(|a|^2+|c|^2)(a\overline b+c\overline d)}
      {(|b|^2+|d|^2)(\overline a b+\overline c d)}.
 \label{eq:qubit-Dye-family-two-parameters}
\end{align}

\item In family 3, \(\widehat S_3\) has the same matrix form as
\(\widehat S_2\), while
\begin{equation}
 |pq|=1,
 \qquad
 |p|^2=\frac{|d|^4}{|a|^4},
 \qquad
 |q|^2=\frac{|a|^4}{|d|^4},
 \qquad
 c=-\frac{a\overline b}{\overline d}.
 \label{eq:qubit-Dye-family-three}
\end{equation}

\item In family 4,
\begin{equation}
 \widehat S_4=\frac1{\sqrt2}
 \begin{pmatrix}
  1&0&0&1\\
  0&1&1&0\\
  0&1&-1&0\\
  -1&0&0&1
 \end{pmatrix},
 \qquad
 c=-\frac{a\overline b}{\overline d},
 \qquad |a|=|d|.
 \label{eq:qubit-Dye-family-four}
\end{equation}

\item In family 5,
\begin{equation}
 \widehat S_5=\Swap=
 \begin{pmatrix}
  1&0&0&0\\
  0&0&1&0\\
  0&1&0&0\\
  0&0&0&1
 \end{pmatrix},
 \label{eq:qubit-Dye-family-five}
\end{equation}
and there is no further restriction on \(Q\).
\end{enumerate}
\end{theorem}

We present Dye's solutions in the braid convention
\eqref{eq:setup-braid} used in this paper, which accounts for the final swap
in \eqref{eq:app-Dye-form}.

Dye obtains this list by starting from Hietarinta's classification of
algebraic Yang--Baxter solutions~\cite{Hietarinta1992} and subsequently
imposing the conditions under which their \(GL(2,\mathbb C)\) transforms are
unitary.  This procedure proves that the
list is exhaustive, but it does not make the resulting unitary families
disjoint.  Once unitarity restricts the parameters, representatives
originating from different algebraic cases may describe the same unitary
gates.  In particular, family 2 is contained, up to homogeneous onsite
unitary conjugation, in a special locus of family 1; this inclusion is shown
explicitly in Proposition~\ref{prop:qubit-Dye-dynamical-reduction}.  We
nevertheless retain Dye's numbering in order to keep the connection with the
original classification transparent.

To analyze the form \eqref{eq:app-Dye-form}, we could always use the polar
decomposition \(Q=UH\), with \(U\) unitary
and \(H\) positive, and absorb \(H\otimes H\) into the representative
\(\widehat S\).  This isolates the physically harmless onsite unitary
\(U\), but replaces the five fixed representatives above by matrices that
depend on \(H\).  The \(GL(2,\mathbb C)\) form \eqref{eq:app-Dye-form} is the
most compact way to state Dye's classification.  The
family-specific reduction relevant for the dynamics is given below in
Proposition~\ref{prop:qubit-Dye-dynamical-reduction}.

The next proposition extracts exactly the information needed for operator
entanglement.

\begin{proposition}[Dynamical reduction of Dye's five families]
\label{prop:qubit-Dye-dynamical-reduction}
Up to an overall phase and the homogeneous onsite unitary conjugation
\eqref{eq:setup-onsite-gauge}, the five families in
Theorem~\ref{thm:qubit-Dye-classification} reduce as follows:
\begin{enumerate}
\item families 1 and 2 are phase-dressed swaps, while family 3 is a
phase-dressed swap complemented by simultaneous spin flips;
\item family 4 is represented by the Clifford gate \(B_4\) in
\eqref{eq:app-Dye-four};
\item family 5 is the identity gate.
\end{enumerate}
In the first item the two permutation gates are \(\Swap\) and
\((X\otimes X)\Swap\), respectively.  Their classical maps are
\begin{equation}
 r_{\rm sw}(a,b)=(b,a),\qquad
 r_{\rm comp}(a,b)=(1-b,1-a),
 \qquad a,b\in\{0,1\}.
 \label{eq:app-qubit-classical-maps}
\end{equation}
\end{proposition}

\begin{proof}
In family 1, the constraint on \(Q\) makes the two columns of \(Q\)
orthogonal.  Its polar decomposition may therefore be written
\(Q=U\Delta\), with \(U\) unitary and \(\Delta\) positive diagonal.  The
factor \(\Delta\otimes\Delta\) commutes with the diagonal representative
\(\widehat S_1\) and with \(\Swap\).  Removing the allowed onsite unitary
\(U\) from \eqref{eq:app-Dye-form} leaves a phase-dressed swap.

In family 2, \eqref{eq:qubit-Dye-family-two-parameters} gives \(q=p^{-1}\).
Choose \(s\in\mathbb C^\times\) with \(s^2=p\) and define
\begin{equation}
 C=\begin{pmatrix}0&s\\s^{-1}&0\end{pmatrix}.
 \label{eq:qubit-Dye-family-two-C}
\end{equation}
A direct tensor product then gives \(\widehat S_2=C\otimes C\).  Therefore,
with
\begin{equation}
 B=QCQ^{-1},
 \label{eq:qubit-Dye-family-two-B}
\end{equation}
the physical gate takes the factorized form
\begin{equation}
 R=k(B\otimes B)\Swap.
 \label{eq:qubit-Dye-family-two-factorized}
\end{equation}
Since \(R\), \(k\), and \(\Swap\) are unitary,
\((B^\dagger B)\otimes(B^\dagger B)=\id\).  The matrix \(B^\dagger B\) is
positive, so this identity forces \(B^\dagger B=\id\).  Hence \(B\) is
unitary.  Writing
\(B=V\operatorname{diag}(e^{i\alpha},e^{i\beta})V^\dagger\) and removing
the allowed onsite unitary \(V\) leaves a phase-dressed swap.  After also
removing the overall phase, its family-1 parameters lie on the locus
\(p=q=t\), \(r=t^2\), where \(t=e^{i(\beta-\alpha)}\).  This proves the
inclusion stated after Theorem~\ref{thm:qubit-Dye-classification}.

The same polar-decomposition argument as in family 1 turns family 3 into a
phase-dressed gate whose classical map is the complemented swap.  This
proves the first item.

For family 4, the restrictions on \(Q\) imply
\begin{equation}
 Q^\dagger Q=(|a|^2+|b|^2)\id.
 \label{eq:app-Dye-four-Q}
\end{equation}
Thus \(Q\) is a scalar multiple of a unitary matrix, and the scalar cancels
from \eqref{eq:app-Dye-form}.  After the allowed onsite change of basis,
\(\widehat S_4\Swap\) becomes, up to an overall phase,
\begin{equation}
 B_4=\frac1{\sqrt2}
 \begin{pmatrix}
  1&0&0&1\\
  0&1&1&0\\
  0&-1&1&0\\
  -1&0&0&1
 \end{pmatrix}
 =\frac{\id+iY\otimes X}{\sqrt2}
 =\exp\!\left(\frac{i\pi}{4}Y\otimes X\right).
 \label{eq:app-Dye-four}
\end{equation}
In the standard Jordan--Wigner convention,
\(Y\otimes X=i\gamma_1\gamma_3\), with
\(\gamma_1=X\otimes\id\) and \(\gamma_3=Z\otimes X\).  Thus \(B_4\) is
the exponential of a Majorana bilinear and, at the angle \(\pi/4\), a
Clifford gate.
Finally, \(\widehat S_5=\Swap\), and \(Q\otimes Q\) commutes with
\(\Swap\).  The two swaps in \eqref{eq:app-Dye-form} therefore cancel,
leaving the scalar identity.  This proves the remaining two items.
\end{proof}

We can now state the uniform bound.

\begin{theorem}[Uniform rank for qubit Yang--Baxter gates]
\label{thm:qubit-uniform}
Let \(D=2\) and let \(R\) be any Yang--Baxter gate.  For every one-site
operator \(O\) initially placed at site zero, every \(t\geq0\), and every
spatial cut \(c\),
\begin{equation}
 \OSR_c(O(t))\leq4.
 \label{eq:qubit-uniform-bound}
\end{equation}
Consequently, for every nonzero \(O\) and every \(\alpha\geq0\),
\begin{equation}
 S_\alpha^{\rm op}(O(t))\leq\log4.
 \label{eq:qubit-uniform-entropy}
\end{equation}
Both bounds are uniform in time.
\end{theorem}

\begin{proof}
Proposition~\ref{prop:qubit-Dye-dynamical-reduction} leaves three cases.
Throughout the case analysis, the overall phase and homogeneous onsite
conjugation used in that proposition do not change the operator-Schmidt
spectrum, as explained after Eq.~\eqref{eq:setup-onsite-gauge}.

In families 1--3, the two classical supports in
Eq.~\eqref{eq:app-qubit-classical-maps} are separated exchanges of the form
\eqref{eq:separated-support}.  For the ordinary swap we have
\(f=g=\id_X\), while for the complemented swap both \(f\) and \(g\) are the
bit flip.  Theorem~\ref{thm:separated-uniform-rank} therefore applies and
gives
\begin{equation}
 \OSR_c(O(t))\leq D(D-1)+1=3.
 \label{eq:qubit-separated-bound}
\end{equation}

In family 4, the gate \(B_4\) is Clifford.  Standard Pauli propagation maps
each one-site Pauli operator to a single Pauli string
\cite{Gottesman1999}.  Expanding an arbitrary one-site operator in the four
Pauli basis elements therefore expresses its evolution as a sum of at most
four product operators.  Its OSR is consequently at most four across every
cut.  In family 5, the gate is the identity, so the local operator is
stationary and has OSR one.

These cases exhaust Dye's classification.  Taking their largest bound gives
Eq.~\eqref{eq:qubit-uniform-bound}, and
Eq.~\eqref{eq:setup-entropy-rank} gives
Eq.~\eqref{eq:qubit-uniform-entropy}.
\end{proof}

The exhaustiveness of this conclusion is specific to qubits: the separated
exchange theorem holds in every local dimension \(D\), but no analogue of
Dye's classification is available for higher \(D\).

For a concrete qubit example, see the phase-dressed swap \(R_K\) in
Eq.~\eqref{eq:qubit-dual-unitary-XXZ} and its discussion in
Subsection~\ref{subsec:separated-phase-dressed}; see also
Refs.~\cite{BertiniKosProsenII,GiudiceEtAl2022}.

\section{Involutive Yang--Baxter gates}
\label{sec:involutive-unitary}

In this section we discuss involutive Yang--Baxter gates.  As recalled in
Subsection~\ref{subsec:yb-gates}, involutivity makes a unitary gate Hermitian
and lets the braid-group representation factor naturally through the
symmetric group.  These gates also admit the natural Baxterization reviewed
in Subsection~\ref{subsec:integrability}, so the corresponding circuits are
integrable in the commuting-transfer-matrix sense discussed there.

Lechner, Pennig, and Wood (LPW) classified these gates through the compatible
symmetric-group representations generated on all tensor powers
\cite{LechnerPennigWood2019}.  Their classification theorem is central to
our argument, and we state it below as
Theorem~\ref{thm:inv-LPW-classification}.
Their notion of equivalence is much coarser than an onsite change of basis,
and this distinction will matter for the entanglement problem.

We do not establish new results in this section.  Instead, we connect the
existing literature to familiar integrable models and to the
representation-theoretic structures needed for the operator-entanglement
bounds derived later in Sections~\ref{sec:inv-polynomial-proof}
and~\ref{sec:normal-form-dynamics}.

\subsection{The LPW invariant and equivalence relation}
\label{subsec:inv-LPW-invariant}

For $n\geq2$, put
\begin{equation}
 R_i=\id^{\otimes(i-1)}\otimes R\otimes
 \id^{\otimes(n-i-1)}.
 \label{eq:inv-Ri}
\end{equation}
The braid relation gives a representation of the braid group, and
$R_i^2=\id$ makes it factor through the symmetric group:
\begin{equation}
 \rho_R^{(n)}:S_n\longrightarrow U(V^{\otimes n}),
 \qquad \rho_R^{(n)}((i,i+1))=R_i.
 \label{eq:inv-Sn-representation}
\end{equation}
The normalized characters are
\begin{equation}
 \chi_R^{(n)}(\sigma)=D^{-n}
 \operatorname{Tr}_{V^{\otimes n}}\rho_R^{(n)}(\sigma).
 \label{eq:inv-character}
\end{equation}
LPW call two gates equivalent when they have the same local dimension and
the same characters at every tensor level.  Since $S_n$ is finite, this is
equivalent to unitary equivalence of $\rho_R^{(n)}$ and
$\rho_S^{(n)}$ for every fixed $n$.  Explicitly, equivalence means that for
each $n\geq2$ there is a unitary $U_n\in U(V^{\otimes n})$ such that
\begin{equation}
 U_n\rho_R^{(n)}(\sigma)=\rho_S^{(n)}(\sigma)U_n,
 \qquad \sigma\in S_n.
 \label{eq:inv-LPW-intertwiner}
\end{equation}
We stress that $U_n$ may depend on $n$ and need not factor into onsite
transformations.  LPW equivalence need not preserve locality with respect to
a spatial cut.

The remarkable point is that this all-level equivalence is determined by a
single $D\times D$ matrix.  Define the unnormalized right partial trace
\begin{equation}
 T_R=(\operatorname{id}\otimes\operatorname{Tr})R,
 \qquad (T_R)_a^{\ c}=\sum_bR_{ab}^{cb}.
 \label{eq:inv-partial-trace}
\end{equation}
It is Hermitian because $R$ is Hermitian.  The LPW partial-trace theorem
says that, for two involutive Yang--Baxter gates $R$ and $S$,
\begin{equation}
 R\sim S
 \quad\Longleftrightarrow\quad
 T_R\text{ and }T_S\text{ are unitarily similar}.
 \label{eq:inv-LPW-equivalence}
\end{equation}
Thus the spectrum of one $D\times D$ Hermitian matrix determines the
equivalence class of the entire compatible family
\eqref{eq:inv-Sn-representation}.  This statement is much stronger than a
classification at one tensor level, but weaker than a classification of
the two-site matrices under local basis changes.

LPW also go beyond this equivalence criterion and construct a standard
representative of every class.  We call these representatives signed-block
normal forms and introduce them, together with the corresponding
classification theorem, in
Subsection~\ref{subsec:inv-LPW-normal-form}.

\subsection{Signed-block normal forms and integer partitions}
\label{subsec:inv-LPW-normal-form}

We now introduce the signed-block normal forms, which are the standard
representatives of the LPW equivalence classes.  The possible spectra in
\eqref{eq:inv-LPW-equivalence} are highly restricted and encoded by these
blocks.  Choose positive integers
\begin{equation}
 d_1^+,\ldots,d_p^+,
 \qquad d_1^-,\ldots,d_q^-,
 \qquad
 \sum_{i=1}^p d_i^++\sum_{j=1}^q d_j^-=D.
 \label{eq:inv-block-sizes}
\end{equation}
Let
\begin{equation}
 V=\bigoplus_{i=1}^pV_i^+\oplus
   \bigoplus_{j=1}^qV_j^-,
 \qquad \dim V_i^+=d_i^+,
 \quad \dim V_j^-=d_j^-.
 \label{eq:inv-block-decomposition}
\end{equation}
It is convenient to combine the positive and negative summands into blocks
$V_\alpha$ with signs $\varepsilon_\alpha\in\{+1,-1\}$.  For each choice of
the integers $d_i^\pm$ in \eqref{eq:inv-block-sizes}, we denote the
corresponding LPW normal form by $N=N(d_1^+,\ldots,d_p^+;d_1^-,\ldots,d_q^-)$.
It is defined by the action
\begin{equation}
 N(u\otimes v)=
 \begin{cases}
  \varepsilon_\alpha u\otimes v,
    &u,v\in V_\alpha,\\
  v\otimes u,
    &u\in V_\alpha,\ v\in V_\beta,\ \alpha\ne\beta.
 \end{cases}
 \label{eq:inv-normal-form-action}
\end{equation}
In words, vectors belonging to different blocks exchange places, while two
vectors in the same block do not exchange and acquire the sign of that
block.

Directly from \eqref{eq:inv-normal-form-action},
\begin{equation}
 T_N|_{V_i^+}=d_i^+\id,
 \qquad
 T_N|_{V_j^-}=-d_j^-\id.
 \label{eq:inv-normal-partial-trace}
\end{equation}

The following classification theorem was proven by LPW
\cite{LechnerPennigWood2019}.
\begin{theorem}[LPW normal-form classification \cite{LechnerPennigWood2019}]
\label{thm:inv-LPW-classification}
Every involutive Yang--Baxter gate is LPW-equivalent to exactly one
signed-block normal form \eqref{eq:inv-normal-form-action}, up to reordering
blocks with the same sign.  Equivalently, its partial trace has eigenvalues
\begin{equation}
 +d_i^+\ \text{with multiplicity }d_i^+,
 \qquad
 -d_j^-\ \text{with multiplicity }d_j^-.
 \label{eq:inv-LPW-spectrum}
\end{equation}
\end{theorem}

For the normal form $N$ supplied by the theorem,
\eqref{eq:inv-LPW-intertwiner} holds with $S=N$.  Equivalently, since the
adjacent transpositions generate $S_n$, the same $U_n$ intertwines every
local generator:
\begin{equation}
 U_nR_k=N_kU_n,
 \qquad k=1,\ldots,n-1,
 \label{eq:inv-LPW-generator-intertwiner}
\end{equation}
where $N_k$ is defined from $N$ in the same way as $R_k$ is defined from
$R$ in \eqref{eq:inv-Ri}.

The following elementary examples illustrate the meaning of the blocks and
their signs.

\begin{enumerate}
\item A single positive block of size $D$ gives $N=\id_{V\otimes V}$ and
$T_N=D\id_V$.  It corresponds to the pair of partitions $((D),\varnothing)$.
A single negative block gives $-\id_{V\otimes V}$ and
$(\varnothing,(D))$.

\item Splitting $V$ into $D$ positive one-dimensional blocks gives the
ordinary swap $\Swap$.  Indeed, unequal basis labels are exchanged, while an
equal pair is fixed.  Here $T_{\Swap}=\id$ and the pair of partitions is
$((1^D),\varnothing)$.

\item In dimension two there are five LPW classes.  Their normal
representatives and partial-trace spectra are
\begin{equation}
\begin{array}{c|c|c|c}
 (d^+;d^-)&\text{representative}&\operatorname{spec}T_N
 &\text{dual unitarity}\ 
 \\ \hline
 (2;-)& \id &(2,2)&\text{no}\\
 (-;2)&-\id&(-2,-2)&\text{no}\\
 (1,1;-)&\Swap&(1,1)&\text{yes}\\
 (-;1,1)&N_{--}&(-1,-1)&\text{yes}\\
 (1;1)&N_{+-}&(1,-1)&\text{yes}
\end{array}
 \label{eq:inv-D2-classes}
\end{equation}
The last three representatives exchange unequal labels and attach the
indicated signs to equal-label pairs.  In the ordered basis
$(\ket{00},\ket{01},\ket{10},\ket{11})$ they are explicitly
\begin{align}
 \Swap&=
 \begin{pmatrix}
  1&0&0&0\\
  0&0&1&0\\
  0&1&0&0\\
  0&0&0&1
 \end{pmatrix},
 &
 N_{--}&=
 \begin{pmatrix}
  -1&0&0&0\\
  0&0&1&0\\
  0&1&0&0\\
  0&0&0&-1
 \end{pmatrix},
 \nonumber\\[1mm]
 N_{+-}&=
 \begin{pmatrix}
  1&0&0&0\\
  0&0&1&0\\
  0&1&0&0\\
  0&0&0&-1
 \end{pmatrix}.
 \label{eq:inv-D2-matrices}
\end{align}
In the occupation-number convention, with $\ket{0}$ empty and $\ket{1}$
occupied, the mixed-sign representative is precisely the standard
fermionic swap gate,
$N_{+-}=\mathrm{fSWAP}=\Swap\,\mathrm{CZ}$: the minus sign on
$\ket{11}$ is the fermionic exchange sign \cite{JiangEtAl2018}.
Equivalently, $N_{+-}$ is the graded permutation on $\mathbb C^{1|1}$.

\item As a three-dimensional example, let $V=V_+\oplus V_-$ with
$\dim V_+=2$ and $\dim V_-=1$.  The normal form is the identity on
$V_+\otimes V_+$, minus the identity on $V_-\otimes V_-$, and swap on the
two mixed sectors.  Its partial-trace spectrum is $(2,2,-1)$.
\end{enumerate}

\subsection{All-positive normal forms as XXC and multiplicity models}
\label{subsec:inv-XXC-multiplicity}

The all-positive LPW normal forms are closely related to known integrable
lattice models.  To make this relation explicit, set $q=0$, write $b=p$ for the number of
blocks, choose bases
$X_\alpha=\{e_{\alpha r}\}_{r=1}^{d_\alpha}$ of $V_\alpha$, and let
$P_\alpha$ be the corresponding projectors.  With
$E_{\alpha r,\beta s}=\ket{e_{\alpha r}}\bra{e_{\beta s}}$, define the
operators
\begin{align}
 \mathsf P_{\rm M}^{(2)}
 &=\sum_{\alpha=1}^{b}P_\alpha\otimes P_\alpha,
 \label{eq:inv-Maassarani-P2}\\
 \mathsf P_{\rm M}^{(3)}
 &=\sum_{\alpha\ne\beta}
   \sum_{r=1}^{d_\alpha}\sum_{s=1}^{d_\beta}
   E_{\beta s,\alpha r}\otimes E_{\alpha r,\beta s},
 \qquad
 \mathsf P_{\rm M}^{(1)}=\id-\mathsf P_{\rm M}^{(2)}.
 \label{eq:inv-Maassarani-P3}
\end{align}
The first operator is the identity on equal-block sectors and the second
exchanges states from distinct blocks.  Consequently the literal
all-positive LPW representative is exactly
\begin{equation}
 N_+=\mathsf P_{\rm M}^{(2)}+\mathsf P_{\rm M}^{(3)}.
 \label{eq:inv-Maassarani-N}
\end{equation}

Equation~\eqref{eq:inv-Maassarani-N} is exactly the unit-twist constant
generator in the XXC and multiplicity constructions of
Refs.~\cite{Maassarani1998XXC,Maassarani1999Multiplicity}.  For two blocks
it is the positive-sign rational XXC generator of type $d_1+d_2$.  For
arbitrary $b$ it is the multiplicity $A_{b-1}$ generator denoted
$(d_1,\ldots,d_b;b,D)$ in
Ref.~\cite{Maassarani1999Multiplicity}: each of the $b$ parent colours is
replaced by $d_\alpha$ copies, copies of one colour do not exchange, and
copies of distinct colours do exchange.

Ref.~\cite{GomborPozsgay2022} uses precisely the same constant block rule
for arbitrary partitions.  For the two-block partition $1+(D-1)$, with $0$
denoting the singleton block, the constant rule is
\begin{equation}
 N_+\ket{x,y}=\begin{cases}
  \ket{y,x},&x=0\text{ or }y=0,\\
  \ket{x,y},&x,y\ne0.
 \end{cases}
 \label{eq:literature-XXC}
\end{equation}
The type $1+2$ gate is the charged hard-core gas rule of
Ref.~\cite{MedenjakKlobasProsen2017}, and the full $1+(D-1)$ family is the
quantum hard-core gas treated in Ref.~\cite{Medenjak2022}.  In dimension
four, the types $1+3$, $1+1+2$, and $2+2$ are among the concrete maps
discussed in Ref.~\cite{GomborPozsgay2022}.

The XXC and multiplicity models discussed above have positive signs on all
blocks.

\subsection{Dual unitarity collapses all blocks to singletons}
\label{subsec:inv-DU-collapse}

We now impose dual unitarity.  Its entire effect on the LPW data follows from
one norm identity.  In the reshuffling convention
\eqref{eq:setup-reshuffle}, the vectorization identity
\begin{equation}
 \mathfrak f(T_R)=R^\Gamma\mathfrak f(\id)
 \label{eq:inv-vec-partial-trace}
\end{equation}
follows immediately from \eqref{eq:inv-partial-trace}.  Hence unitarity of
$R^\Gamma$ fixes the Hilbert--Schmidt norm:
\begin{equation}
 \|T_R\|_{\mathrm{HS}}^2
 =\|\mathfrak f(\id)\|^2=D.
 \label{eq:inv-DU-partial-trace-norm}
\end{equation}
This is the only point at which dual unitarity enters the LPW argument.

\begin{proposition}[Singleton collapse]
\label{prop:inv-singleton-collapse}
If an involutive Yang--Baxter gate is dual-unitary, every positive
and negative LPW block has dimension one.  Conversely, a signed-block LPW
normal form is dual-unitary if and only if all its blocks are
one-dimensional.
\end{proposition}

\begin{proof}
Using \eqref{eq:inv-LPW-spectrum} and
\eqref{eq:inv-DU-partial-trace-norm},
\begin{equation}
 D=\|T_R\|_{\mathrm{HS}}^2
 =\sum_i(d_i^+)^3+\sum_j(d_j^-)^3.
 \label{eq:inv-cube-sum}
\end{equation}
On the other hand, \eqref{eq:inv-block-sizes} gives
$D=\sum_i d_i^++\sum_jd_j^-$.  Subtraction yields
\begin{equation}
 \sum_i d_i^+\bigl((d_i^+)^2-1\bigr)
 +\sum_j d_j^-\bigl((d_j^-)^2-1\bigr)=0.
 \label{eq:inv-positive-cube-difference}
\end{equation}
Every summand is a nonnegative integer and vanishes only for a block of
size one.  This proves the first statement and the necessary direction for
the normal form itself.  Conversely, if all blocks of $N$ are
one-dimensional, then \eqref{eq:inv-normal-form-action} is the phase-dressed
SWAP gate $R_q$ of Eq.~\eqref{eq:separated-phase-dressed-swap}, with
$q_{aa}=\varepsilon_a$ and $q_{ab}=1$ for $a\ne b$.
Its reshuffling is again a monomial unitary, so it is dual-unitary.
\end{proof}

The proposition collapses the classification drastically.  First,
\begin{equation}
 \operatorname{spec}T_R\subset\{+1,-1\},
 \qquad T_R^2=\id.
 \label{eq:inv-partial-trace-involution}
\end{equation}
Second, the two partitions are $(1^p)$ and $(1^q)$ with
$p+q=D$.  Thus only $D+1$ of the LPW equivalence classes can contain a
gate with dual unitarity.

The converse in Proposition~\ref{prop:inv-singleton-collapse} applies to the
signed singleton normal form itself, not to every gate in its LPW class.  LPW
equivalence does not preserve dual unitarity, so a gate without dual
unitarity may have singleton LPW data and hence the same type of normal form.
The tensor-level intertwiners $U_n$ in
\eqref{eq:inv-LPW-intertwiner} may be nonlocal across the spatial cut and can
therefore change operator entanglement.  For a general representative
$R\sim N$, we use the normal form only to count the dimension of the
representation-theoretic commutant in the proof of
Theorem~\ref{thm:inv-polynomial-rank}; that dimension is invariant under
arbitrary unitary equivalence.

\section{Polynomial operator Schmidt rank for involutive gates with dual unitarity}
\label{sec:inv-polynomial-proof}

In this section we prove that the operator Schmidt rank of an evolved
one-site operator grows at most polynomially in circuits with an involutive
dual-unitary Yang--Baxter gate.

We combine the geometric and
representation-theoretic arguments developed in the previous sections.
Dual unitarity first gives the partial-trace identity
\eqref{eq:inv-DU-partial-trace-norm}, which determines the possible LPW
normal forms of these gates: every block is a signed singleton.  The rest of
the proof has three ingredients.  The rectangular
reduction of Section~\ref{sec:rectangular-reduction} puts one Schmidt leg in
the braid-invariant output space of
Theorem~\ref{thm:central-invariance}.  The LPW
classification identifies the dimension of that space with the commutant of
a signed singleton exchange.  Finally, in this section we count this
commutant by occupation numbers of $D^2$ ket--bra pair types.

For the signed singleton normal form $N$ itself,
Theorem~\ref{thm:normal-form-label-grouping} will show that the one-site
rank remains bounded independently of time.  The polynomial estimate in the
theorem below is needed for a general involutive dual-unitary gate, which is
only LPW-equivalent to that matrix.

\begin{theorem}[Polynomial rank for involutive gates with dual unitarity]
\label{thm:inv-polynomial-rank}
Let $R$ be an involutive Yang--Baxter gate with dual unitarity.  Then every
one-site operator initially at site zero obeys
\begin{equation}
 \OSR_c(O(t))\leq
 \binom{t+D^2-1}{D^2-1}\leq(t+1)^{D^2-1}.
 \label{eq:inv-main-rank-bound}
\end{equation}
Consequently, for every $\alpha\geq0$,
\begin{equation}
 S_\alpha^{\rm op}(O(t))\leq(D^2-1)\log(t+1).
 \label{eq:inv-main-entropy-bound}
\end{equation}
\end{theorem}

\begin{proof}
Let $\mathcal C_t(R)$ be the centralizer defined in
\eqref{eq:centralizer-braid-algebra}--\eqref{eq:centralizer-definition}.
To count its dimension, we pass to the auxiliary circuit in which the gate
$R$ is replaced by the signed singleton gate $N_{p,q}$ obtained from
\eqref{eq:inv-normal-form-action} by taking $p$ positive and $q$ negative
one-dimensional blocks.  This is possible because, for every $t$, the
generator-level relation \eqref{eq:inv-LPW-generator-intertwiner} shows that
conjugation by $U_t$ maps
$\mathcal A_t(R)$ onto $\mathcal A_t(N_{p,q})$ and hence maps
$\mathcal C_t(R)$ onto $\mathcal C_t(N_{p,q})$.  Thus this replacement leaves
the centralizer dimension unchanged.  The advantage is that $N_{p,q}$ acts
as a signed exchange in a product basis, which makes its centralizer easy to
count.

Choose the one-site basis adapted to the singleton blocks of $N_{p,q}$.
After folding, the vectorized matrix units
$\mathfrak f^{\otimes t}(\ket{x}\bra{y})$, with
$x,y\in\{1,\ldots,D\}^t$, form a computational product basis of
$\mathcal W^{\otimes t}$.  The idea is elementary: we take any computational-basis
matrix unit $M_{x,y}:=\ket{x}\bra{y}$, construct its orbit under the adjoint
action of $S_t$, and average its transforms over the group.  The singleton
normal form simultaneously permutes the ket and bra labels and may attach a
sign.  Its adjoint action is therefore
\begin{equation}
 \operatorname{Ad}\rho_{N_{p,q}}^{(t)}(\sigma)
 (M_{x,y})
 =\zeta(\sigma;x,y)M_{\sigma x,\sigma y},
 \qquad |\zeta(\sigma;x,y)|=1.
 \label{eq:inv-adjoint-monomial}
\end{equation}
The corresponding group average is
\begin{equation}
 M_{x,y}^{S_t}:=\frac{1}{t!}\sum_{\sigma\in S_t}
 \operatorname{Ad}\rho_{N_{p,q}}^{(t)}(\sigma)(M_{x,y}).
 \label{eq:inv-orbit-average}
\end{equation}
Reindexing the sum shows that $M_{x,y}^{S_t}$ is invariant under the adjoint
action and hence belongs to the commutant.  If it is nonzero, it spans the
invariant subspace supported on the orbit of $M_{x,y}$, because invariance
fixes all coefficients on that orbit up to an overall scale.  The group
average can vanish if the stabilizer signs render the sum zero, in which case
that orbit contributes no invariant direction.  Thus each orbit contributes
either one or zero invariant operators.  Since we seek only an upper bound on
the centralizer dimension, the possible cancellations cause no problem: we
may ignore the signs and count the orbits of the underlying simultaneous
permutation action.

The orbit under simultaneous permutation of positions in the pair $(x,y)$
forgets the order of the sites.  It remembers only how often each
ket--bra pair $(a,b)$ occurs, namely the $D^2$ joint
occupation numbers
\begin{equation}
 n_{ab}=\#\{k:(x_k,y_k)=(a,b)\},
 \qquad n_{ab}\geq0,
 \qquad \sum_{a,b}n_{ab}=t.
 \label{eq:inv-joint-occupations}
\end{equation}
There are $\binom{t+D^2-1}{D^2-1}$ such possibilities.  Consequently
\begin{equation}
 c_t(R)=c_t(N_{p,q})
 \leq\binom{t+D^2-1}{D^2-1}.
 \label{eq:inv-commutant-bound}
\end{equation}
The OSR consequence of the braid-invariant output-space theorem,
Corollary~\ref{cor:centralizer-osr-bound}, gives
$\OSR_c(O(t))\leq c_t(R)$.  Together with
\eqref{eq:inv-commutant-bound}, this proves
\eqref{eq:inv-main-rank-bound}.  Since
$S_\alpha^{\rm op}\leq\log\OSR$ for every $\alpha\geq0$, it also proves
\eqref{eq:inv-main-entropy-bound}.
\end{proof}

We stress that the orbit count is intentionally coarse: we need only a
uniform polynomial upper bound, and the sign-dependent cancellations can
reduce the centralizer dimension.  For the LPW normal forms analyzed
directly in Section~\ref{sec:normal-form-dynamics}, the actual growth powers
of the operator Schmidt rank are determined there.

\section{Operator entanglement of LPW normal-form circuits}
\label{sec:normal-form-dynamics}

In this section the physical gate itself is the normal form $R=N$ of
\eqref{eq:inv-normal-form-action}; dual unitarity is
not assumed.  We ask how much information a one-site insertion can carry
through the rectangular braid.  The answer is much smaller than the full LPW
centralizer suggests.  We derive the quadratic upper bound in
Subsection~\ref{subsec:normal-form-quadratic} and establish the sharp growth
classification in Subsection~\ref{subsec:normal-form-sharpness}.  Depending
on the block structure, the {maximal} one-site operator Schmidt rank remains bounded or
grows linearly or quadratically with time.  No higher polynomial power is
possible for an LPW normal form, and all three behaviours are attained.

This section generalizes the operator-spreading analysis of Medenjak
\cite{Medenjak2022} within the same model family.  The model studied there is
the all-positive $1+(D-1)$ normal form: it has one empty state and several
particle colours, particles move through empty sites, and their colours form
a single file.  Here we treat the full signed-block family, with arbitrary
block sizes and signs.  Medenjak found that the information needed to describe
an evolved local operator grows at most linearly with time.  This agrees with
our linear branch.  The full family also contains a quadratic branch, but no
higher power.  Rule~54 has a different local dynamics, yet its
exact construction likewise uses $O(t^2)$ terms for a local operator
\cite{AlbaDubailMedenjak2019}.  Thus the same maximal power occurs in both
settings, although the mechanisms are different.

Before deriving the physical rank bound, let us return to the centralizer
method developed in Section~\ref{sec:centralizer-output-space} and used in
Section~\ref{sec:inv-polynomial-proof}.  For these normal forms, the
centralizer does not give the right count.  It bounds the space in which the
rectangular output may live, but a local source need not explore that entire
space.  For the one-block normal form, $N=\pm\id$, every local observable is stationary
although $c_t(N)=D^{2t}$.  This mismatch is not restricted to the identity
gate.  For the all-positive $2+2$ normal form, a direct cycle-index count gives
\begin{equation}
 c_t(N)=4^t\binom{t+3}{3}.
 \label{eq:normal-form-22-centralizer}
\end{equation}
By contrast, we prove later in
Subsection~\ref{subsec:normal-form-quadratic} that every one-site operator has
rank $O(t^2)$, and in Subsection~\ref{subsec:normal-form-sharpness} that the
rank of a matrix unit connecting the two blocks grows quadratically with
time; see
\eqref{eq:normal-form-XXC-22-quadratic}.  Thus the
symmetry-allowed space can be exponentially larger than the part reached by a
local source.  We therefore analyze the source-reachable space directly.  We
first derive the blockwise single-file rule and then count the Schmidt
directions reached by a local matrix unit.  Finally, we construct matrix units
that attain the linear and quadratic bounds.

\subsection{Blockwise single-file dynamics}
\label{subsec:normal-form-single-file}

Let us first recall the signed-block notation of
Subsection~\ref{subsec:inv-LPW-normal-form}.  We write
\begin{equation*}
 V=\bigoplus_\alpha V_\alpha,
 \qquad
 d_\alpha:=\dim V_\alpha,
 \qquad
 \varepsilon_\alpha\in\{+1,-1\}.
\end{equation*}
The collection of integers $\{d_\alpha\}_\alpha$ is the union, with
multiplicities, of $\{d_i^+\}_{i=1}^p$ and $\{d_j^-\}_{j=1}^q$ introduced in
Eqs.~\eqref{eq:inv-block-sizes}--\eqref{eq:inv-block-decomposition}.  In this
notation the normal-form gate acts as
\begin{equation*}
 N(u\otimes v)=
 \begin{cases}
  \varepsilon_\alpha u\otimes v,
    &u,v\in V_\alpha,\\
  v\otimes u,
    &u\in V_\alpha,\ v\in V_\beta,\ \alpha\ne\beta.
 \end{cases}
\end{equation*}

For each block $V_\alpha$, choose a basis labelled by a set $X_\alpha$, so
that $|X_\alpha|=d_\alpha$, with the sets $X_\alpha$ mutually disjoint.  We write
$X=\bigsqcup_\alpha X_\alpha$ for the resulting basis labels of $V$.  For a
word $A$, let $A_\alpha$ be the ordered subword formed by its letters in
$X_\alpha$, and let
$n_\alpha(A)=|A_\alpha|$.  We use the equal-length rectangular circuit $W_t$
introduced in Section~\ref{sec:rectangular-reduction}.  Its input is
$A\mid B$, with $A,B\in X^t$, and it moves the right word $B$ to the left
through the left word $A$, producing the output order
$B_{\rm out}\mid A_{\rm out}$.

The normal-form rule has a simple consequence.  Within each block the colours
form a single file: they never overtake one another.  This is the direct
blockwise extension of the single-file constraint for charged hard-core gases
\cite{KrajnikEtAl2024}.  For the partition $1+(D-1)$, the singleton is the
vacancy and this is the usual single-file dynamics of coloured particles.  In
a general partition, colours from distinct blocks pass through one another,
but every block retains its own internal order.  Thus the sequence of block
labels of $B$ emerges on the left and that of $A$ on the right, while the
rectangle only splits each ordered colour string into a left and a right part.

\begin{lemma}[Blockwise single-file action]
\label{lem:normal-form-single-file}
For the normal-form gate $N$, the equal-length rectangular circuit is
Hermitian, $W_t^\dagger=W_t$, and acts as
\begin{equation}
 W_t\ket{A,B}
 =\left(\prod_\alpha
   \varepsilon_\alpha^{n_\alpha(A)n_\alpha(B)}\right)
   \ket{L,R},
 \qquad
 L_\alpha R_\alpha=A_\alpha B_\alpha,
 \quad |L_\alpha|=n_\alpha(B),
 \quad |R_\alpha|=n_\alpha(A).
 \label{eq:normal-form-single-file-action}
\end{equation}
Here $L$ has the block-label word of $B$, while $R$ has the block-label word
of $A$.  Consequently, the conjugation action of $W_t$ on a
computational-basis matrix unit is obtained by applying
\eqref{eq:normal-form-single-file-action} independently to its ket and bra
words.
\end{lemma}

\begin{proof}
Every letter of $A$ crosses every letter of $B$ once.  Letters from distinct
blocks exchange positions, while letters from the same block do not exchange
and contribute the sign $\varepsilon_\alpha$.  Hence the block-label words
exchange sides, but the ordered string $A_\alpha B_\alpha$ is preserved and
split after $n_\alpha(B)$ letters.  There are
$n_\alpha(A)n_\alpha(B)$ equal-block collisions, which gives the phase.
Finally, since $N^2=\id$, its braid representation factors through the
symmetric group.  The circuit $W_t$ represents the permutation exchanging
the two length-$t$ blocks.  This permutation is an involution, and unitarity
therefore gives $W_t^\dagger=W_t$.  The conjugation statement follows by
applying \eqref{eq:normal-form-single-file-action} separately to the ket and
bra words.
\end{proof}

The physical content is that the rectangle redistributes a pre-existing
ordered colour string but never creates a new ordering.  For example, if
$A_\alpha=a_1a_2$ and $B_\alpha=b_1$, then the two $\alpha$-subwords at the
outputs are $a_1$ and $a_2b_1$.  We call the position between the two output
subwords the \emph{split position} of the $\alpha$-string.  Within this block,
we only need to track this split position.  In this example, the split has
moved one position to the left.

\subsection{Active-block memory and the upper bound}
\label{subsec:normal-form-quadratic}

Subsection~\ref{subsec:normal-form-single-file} determined the output of one
computational-basis word.  Equation~\eqref{eq:normal-form-single-file-action}
defines a separate split position for every block, but this is not yet a
count of Schmidt terms.  The calculation has three steps.  We first expand
the local operator and the identity inputs in computational-basis matrix
units.  We then evolve every term separately using the single-file rule.
Finally, we resum the output terms by grouping them according to the residual
split data that still correlate the two outputs.  Once these data are fixed,
each group factorizes across the cut, so the number of groups bounds the
operator Schmidt rank.

We evolve each matrix unit separately and sum the resulting
operators at the end.  An arbitrary one-site operator is a sum of at most
$D^2$ matrix units, so this termwise calculation can change the final
prefactor but not the growth power.

Fix $x\in X_\alpha$ and $y\in X_\beta$.  When $\alpha=\beta$, we allow
$x=y$; thus this case includes both diagonal and off-diagonal matrix units.
The first step gives
\begin{align}
 \id_A^{\otimes t}
 &=\sum_{A\in X^t}\ket{A}\bra{A},&
 E_{xy}\otimes\id^{\otimes(t-1)}
 &=\sum_{C\in X^{t-1}}\ket{xC}\bra{yC}.
 \label{eq:normal-form-source-expansion}
\end{align}
Thus we have to evolve input matrix units
$\ket{A,xC}\bra{A,yC}$.  The ket and bra words differ only at the
distinguished site, where they carry $x$ and $y$, respectively; when $x=y$,
the two words coincide.  If $\alpha=\beta$, the distinguished matrix unit
involves a single active block.  If $\alpha\ne\beta$, it involves two active
blocks, namely those containing $x$ and $y$.

To resum the terms in Eq.~\eqref{eq:normal-form-source-expansion}, we now
identify which input data can still correlate the two output factors.  Data
that can be summed independently on the two sides require no label.  We
collect only the remaining data into a candidate Schmidt label.  The
spectator and active blocks play different roles in this separation.

Consider first a spectator block
$\delta\notin\{\alpha,\beta\}$.  Its ket and bra strings coincide, so the
single-file rule splits the same ordered colour string on the two layers.
After fixing the block-label words, the sum over its internal colours
factorizes between the two outputs.  A spectator block therefore carries no
information across the Schmidt cut.

An active block, by contrast, can retain information about the distinguished
matrix-unit insertion.  If $\alpha=\beta$, the inserted ket and bra entries
lie in the same output factor, and we only need their position relative to
the split.  This gives at most $O(t)$ possibilities.  If
$\alpha\ne\beta$, the ket and bra splits in each active
block are displaced by one.  Either the difference again stays in one output
factor, or one internal colour lies in different output factors on the two
layers.  In the latter case we retain this colour in addition to the split
position.  A non-singleton active block therefore contributes at most
$O(td_\gamma)$ labels.  For a singleton block the transferred colour is
fixed and the block contributes no growing label.  Since a matrix unit has at
most two active blocks, this explains the quadratic upper bound in
Theorem~\ref{thm:normal-form-label-grouping}.
We call the retained split position and, when needed, the transferred colour
the \emph{active-block memory}.

This memory count is not by itself a Schmidt decomposition.  We must also
show that, after the retained data are fixed, the remaining left and right
output choices are independent.  In other words, after a label is fixed, the
remaining left and right output choices must vary independently, and the
collision signs must be absorbed into one output factor.  The following
theorem makes this factorization precise.

\begin{theorem}[Active-block Schmidt decomposition]
\label{thm:normal-form-label-grouping}
Fix $x\in X_\alpha$ and $y\in X_\beta$.  The spectator blocks can be
summed independently across the output cut.  The remaining active-block data
can be grouped into a finite label set $\Lambda_{\alpha\beta}(t)$ such that
\begin{equation}
 \Phi_t(E_{xy})=
 \sum_{\lambda\in\Lambda_{\alpha\beta}(t)}
 L_\lambda\otimes R_\lambda.
 \label{eq:normal-form-grouped-product}
\end{equation}
Consequently,
\begin{align}
 \OSR(\Phi_t(E_{xy}))
 &\leq
 \begin{cases}
  1,&d_\alpha=1,\\
  2t,&d_\alpha>1,
 \end{cases}
 &&x,y\in X_\alpha,
 \label{eq:normal-form-same-block-bound}\\
 \OSR(\Phi_t(E_{xy}))
 &\leq q_\alpha(t)q_\beta(t),
 &&x\in X_\alpha,
 \quad y\in X_\beta,
 \quad\alpha\ne\beta,
 \label{eq:normal-form-different-block-bound}
\end{align}
where
\begin{equation}
 q_\gamma(t)=
 \begin{cases}
  1,&d_\gamma=1,\\
  t(d_\gamma+1),&d_\gamma>1.
 \end{cases}
 \label{eq:normal-form-active-label-count}
\end{equation}
\end{theorem}

The proof is presented in
Appendix~\ref{app:normal-form-label-grouping}.

With the matrix-unit decomposition in hand, we now sum over the at most
$D^2$ components of the initial operator.  This gives the physical rank and
entropy bounds.

\begin{corollary}[Quadratic one-site upper bound for LPW normal forms]
\label{cor:normal-form-quadratic}
Let the physical brickwork gate be the signed-block normal form
$R=N$.  For every one-site operator initially at site zero and every
$t\geq1$,
\begin{align}
 \OSR_c(O(t))&\leq C_Dt^2,
 & C_D&:=D^2\bigl((D+1)^2+2\bigr),
 \label{eq:normal-form-quadratic-bound}\\
 S_\alpha^{\rm op}(O(t))&\leq2\log t+\log C_D,
 &&\alpha\geq0,
 \label{eq:normal-form-log-entropy}
\end{align}
\end{corollary}

\begin{proof}
Lemma~\ref{lem:yb-rectangle} allows us to compute the physical Schmidt rank
from the rectangular circuit.  Expand the initial operator as
$O=\sum_{x,y}o_{xy}E_{xy}$.  By linearity of $\Phi_t$ and subadditivity of
the operator Schmidt rank,
\begin{equation}
 \OSR(\Phi_t(O))
 \leq\sum_{x,y:o_{xy}\ne0}\OSR(\Phi_t(E_{xy})).
 \label{eq:normal-form-matrix-unit-reduction}
\end{equation}
This sum contains at most $D^2$ terms, independently of time.  Applying
Theorem~\ref{thm:normal-form-label-grouping} to every matrix unit, and using
$q_\alpha(t)\leq t(d_\alpha+1)$, gives
\begin{align*}
 \OSR(\Phi_t(O))
 &\leq 2t\sum_\alpha d_\alpha^2
 +t^2\sum_{\alpha\ne\beta}
 d_\alpha d_\beta(d_\alpha+1)(d_\beta+1)\\
 &\leq 2D^2t
 +t^2\left(\sum_\alpha d_\alpha(d_\alpha+1)\right)^2\\
 &\leq D^2\bigl((D+1)^2+2\bigr)t^2.
\end{align*}
Here we used $t\geq1$, $\sum_\alpha d_\alpha^2\leq D^2$, and
$\sum_\alpha d_\alpha(d_\alpha+1)\leq D(D+1)$.  This proves
\eqref{eq:normal-form-quadratic-bound}.  The entropy estimate follows from
$S_\alpha^{\rm op}\leq\log\OSR$.
\end{proof}

\subsection{Sharpness and growth classification}
\label{subsec:normal-form-sharpness}

Subsection~\ref{subsec:normal-form-quadratic} gives the relevant upper bounds.
We now show that the linear and quadratic powers are attained.  For these
matching lower bounds, we reverse the preceding strategy: instead of grouping
source terms, we project onto selected output matrix units that retain one or
two independent split positions.

\begin{theorem}[Sharpness of the active-block bounds]
\label{thm:normal-form-sharp-powers}
If exactly one block $V_\alpha$ is non-singleton and at least one other block
is present, there exist distinct $x,y\in X_\alpha$ such that, for every
$t\geq1$,
\begin{equation*}
 \OSR_c(E_{xy}(t))\geq t.
\end{equation*}
If at least two blocks $V_\alpha$ and $V_\beta$ are non-singletons, there
exist $x\in X_\alpha$ and $y\in X_\beta$ such that, for every $t\geq1$,
\begin{equation*}
 \OSR_c(E_{xy}(t))\geq\frac{t(t+1)}{2}.
\end{equation*}
\end{theorem}

\begin{proof}
By Lemma~\ref{lem:yb-rectangle}, it is enough to work with the output
$\Phi_t(E_{xy})$ of the rectangular circuit.  We isolate a part of this
operator by projecting the two output operator spaces onto spans of selected
computational-basis matrix units.  Equivalently, we expand
$\Phi_t(E_{xy})$ in products of output matrix units and discard every
component outside the two selected spans.  If $\Pi_B$ and $\Pi_A$ denote
these Hilbert--Schmidt projections, then
\begin{equation}
 \OSR\bigl((\Pi_B\otimes\Pi_A)\Phi_t(E_{xy})\bigr)
 \leq \OSR\bigl(\Phi_t(E_{xy})\bigr),
 \label{eq:normal-form-local-projection}
\end{equation}
because applying a linear map separately to the two factors cannot increase
the number of terms in a product decomposition.  If the selected part is a
sum of $r$ matched products $F_i\otimes G_i$, with nonzero coefficients and
with the two families of matrix units separately orthonormal, then its phases
can be absorbed into one family.  The resulting sum is an operator Schmidt
decomposition of rank $r$.
Equation~\eqref{eq:normal-form-local-projection} therefore gives the same
lower bound for the full output.  For a basis letter $w$, we write $w^r$ for
the word formed by $r$ copies of $w$, with $w^0$ the empty word.
Moreover, $W_t$ is a signed permutation in the computational basis.  Thus
different input matrix units give different output matrix units, and a
selected contribution cannot be cancelled by another input term.

First assume that exactly one block $V_\alpha$ is non-singleton and that there
is another block.  Choose distinct $x,y\in X_\alpha$ in the non-singleton
block and one letter $z\in X_\beta$ from a different block.  In the expansion
of the evolved $E_{xy}$, select for $0\leq k<t$ the input words
\begin{align}
 A_k&=x^kz^{t-k},& C&=x^{t-1}.
 \label{eq:normal-form-linear-words}
\end{align}
The two $B$ words are $x^t$ on the ket layer and $yx^{t-1}$ on the bra
layer.  Their block occupations agree, so the ket and bra signs cancel.  The
single-file rule therefore maps this input term to $F_k\otimes G_k$, where
\begin{align}
 F_k&=\ket{x^t}\bra{x^kyx^{t-k-1}},&
 G_k&=\ket{A_k}\bra{A_k}.
 \label{eq:normal-form-linear-output-units}
\end{align}
Let $\Pi_B$ and $\Pi_A$ project onto the spans of the $F_k$ and $G_k$,
respectively.  A term that survives both projections would have output
$F_k\otimes G_l$ for some $k$ and $l$.  The operator $G_l$ fixes
$A=A_l=x^lz^{t-l}$.  On the bra layer, the single-file rule preserves the
ordered $\alpha$-string.  The input string is therefore
$x^lyx^{t-1}$, whereas the selected output factors give
$x^kyx^{t-1-k+l}$.  Since $x\ne y$, the position of $y$ forces $k=l$.
The remaining letters then fix $C=x^{t-1}$, so the surviving term is
precisely the one selected above.  Hence
\begin{equation}
 (\Pi_B\otimes\Pi_A)\Phi_t(E_{xy})
 =\sum_{k=0}^{t-1}F_k\otimes G_k.
 \label{eq:normal-form-linear-projection}
\end{equation}
The $F_k$ are distinct matrix units, and so are the $G_k$.  They form two
orthonormal families, and the selected part therefore has operator Schmidt
rank $t$.  Equation~\eqref{eq:normal-form-local-projection} and
Lemma~\ref{lem:yb-rectangle} give
$\OSR_c(E_{xy}(t))\geq t$.  Together with the upper bound $2t$, this proves
linear growth.

For the quadratic case, choose two non-singleton blocks and distinct letters
$x,u\in X_\alpha$ and $y,v\in X_\beta$.  For each
$0\leq a\leq k\leq t-1$, select the input words
\begin{equation}
 A^{(a)}=u^av^{t-a},
 \qquad
 C^{(k)}=u^kv^{t-1-k}.
 \label{eq:normal-form-quadratic-words}
\end{equation}
The word $A^{(a)}$ contains $a$ letters from block $\alpha$ and $t-a$
letters from block $\beta$.  The ket--bra sign is therefore
$\varepsilon_\alpha^a\varepsilon_\beta^{t-a}$, and the single-file rule maps
the input term to
$\varepsilon_\alpha^a\varepsilon_\beta^{t-a}
 L_{a,k}\otimes R_{a,k}$, where
\begin{align}
 L_{a,k}
 &=\ket{u^axu^{k-a}v^{t-1-k}}
   \bra{vu^kv^{t-1-k}},\nonumber\\
 R_{a,k}
 &=\ket{u^av^{t-a}}
   \bra{u^av^{k-a}yv^{t-1-k}}.
 \label{eq:normal-form-quadratic-output-units}
\end{align}
Now let $\Pi_B$ and $\Pi_A$ project onto the spans of the $L_{a,k}$ and
$R_{a,k}$, respectively.  Consider a term that survives with output
$L_{a,k}\otimes R_{a',k'}$.  The block words and the position of $x$ in
$L_{a,k}$ recover $(a,k)$.  Independently, the block words and the position
of $y$ in $R_{a',k'}$ recover $(a',k')$.  For a term in the source expansion,
the ket and bra layers must have the common words $A$ and $C$ displayed in
\eqref{eq:normal-form-source-expansion}.  More explicitly, the preserved
ket $\alpha$-string is
$u^axu^{k-a+a'}$ in the selected outputs and $u^{a'}xu^k$ at the input.
The position of $x$ forces $a=a'$.  The preserved bra $\beta$-string then
has $t-k+k'-a$ letters $v$ before $y$ in the selected outputs, but $t-a$
such letters at the input.  The position of $y$ therefore forces $k=k'$.
Conversely, every pair $0\leq a\leq k<t$ occurs through the input term in
\eqref{eq:normal-form-quadratic-words}.  The selected part is therefore
\begin{equation}
 (\Pi_B\otimes\Pi_A)\Phi_t(E_{xy})
 =\sum_{0\leq a\leq k<t}
   \varepsilon_\alpha^a\varepsilon_\beta^{t-a}
   L_{a,k}\otimes R_{a,k}.
 \label{eq:normal-form-quadratic-projection}
\end{equation}
Every displayed coefficient is $\pm1$ and may be absorbed into
$R_{a,k}$.  The two families in
\eqref{eq:normal-form-quadratic-projection} are orthonormal by the recovery
observation above.  Their common number is
\begin{equation}
 \#\{(a,k):0\leq a\leq k<t\}
 =\frac{t(t+1)}{2}.
 \label{eq:normal-form-quadratic-lower}
\end{equation}
It follows from \eqref{eq:normal-form-local-projection} and
Lemma~\ref{lem:yb-rectangle} that
$\OSR_c(E_{xy}(t))\geq t(t+1)/2$.  Combining this lower bound with
\eqref{eq:normal-form-different-block-bound} proves quadratic growth.
\end{proof}

The simplest quadratic example is the all-positive $2+2$ normal form.  For
$x$ and $y$ chosen in its two different blocks, the preceding proof gives
\begin{equation}
 \frac{t(t+1)}{2}
 \leq\OSR_c(E_{xy}(t))\leq9t^2.
 \label{eq:normal-form-XXC-22-quadratic}
\end{equation}
Thus quadratic growth already occurs in an elementary XXC normal-form
circuit.

The linear branch recovers familiar integrable models.  For the
all-positive partition $1+(D-1)$ {with $D>2$}, the gate is the two-block XXC generator of
Ref.~\cite{Maassarani1998XXC}, equivalently the $1+(D-1)$ multiplicity model
of Ref.~\cite{Maassarani1999Multiplicity}.  Its $1+2$ member is the charged
hard-core gas of Ref.~\cite{MedenjakKlobasProsen2017}, while the general
$1+(D-1)$ case is the quantum hard-core gas studied in
Ref.~\cite{Medenjak2022}.  Medenjak's linearly growing construction is
therefore the case with exactly one non-singleton block in our
classification.  Combining
Theorems~\ref{thm:normal-form-label-grouping}
and~\ref{thm:normal-form-sharp-powers} extends this result to the full family:
the maximal one-site rank is bounded when there is only one block or when all
blocks are singletons, grows linearly when there is exactly one
non-singleton block and at least one other block, and grows quadratically when
there are at least two non-singleton blocks.

\section{Non-degenerate Yang--Baxter maps and phase dressings}
\label{sec:yb-maps}

In this section we analyze non-degenerate permutation Yang--Baxter gates and
arbitrary phase-dressed Yang--Baxter-map gates.  The bare gate sends each two-site
computational-basis state to a unique computational-basis state, while a
phase dressing also multiplies the output by a phase.  In the latter class
the support satisfies the braid relation, while the full dressed gate need
not do so.  We exploit this common
basis-state structure to derive upper bounds on the growth of the OSR.
For involutive supports whose phase dressing preserves the braid relation,
Subsection~\ref{subsec:phase-involutive-support} will combine this structure
with the centralizer method of Section~\ref{sec:centralizer-output-space} to
obtain {an additional polynomial upper bound}.

Our key method is to use the so-called ``guitar coordinates'' or ``guitar map''; this map is introduced in 
Subsection~\ref{subsec:guitar}.
The construction goes back to Etingof, Schedler, and Soloviev, who
showed that the nontrivial symmetric-group action associated with an
involutive non-degenerate Yang--Baxter map is conjugate to the ordinary permutation
action \cite{EtingofSchedlerSoloviev1999}.
The extension of the construction beyond involutive solutions was developed
by Soloviev and by Lu, Yan, and Zhu
\cite{Soloviev2000,LuYanZhu2000}.
The terminology ``guitar map'' and
a useful graphical formulation are discussed in Sections~5--6 of
Ref.~\cite{LebedVendramin2017}. 

\subsection{Set-theoretical Yang--Baxter maps}
\label{subsec:yb-definition}

Let $r$ be a set-theoretical Yang--Baxter map on the finite set $X$ of $D$
local colours, in the sense of Subsection~\ref{subsec:yb-gates}.
We write its components as
\begin{equation}
 r:X^2\longrightarrow X^2,
 \qquad r(x,y)=\bigl(\lambda_x(y),\mu_y(x)\bigr)
 \label{eq:yb-components}
\end{equation}
and assume it is non-degenerate, so $\lambda_x:X\to X$ and
$\mu_y:X\to X$ are permutations for every $x,y\in X$.  We define the
corresponding physical permutation gate by its action on the computational
basis elements:
\begin{equation}
 R\ket{x,y}=\ket{r(x,y)}.
 \label{eq:yb-permutation-gate}
\end{equation}
This gate is unitary, and its reshuffling has one nonzero entry in every row
and column precisely because of non-degeneracy, so it is dual-unitary.
Unit-modulus phase dressings do not destroy dual unitarity.

\tikzset{
  opent/guitar wire/.style={
    draw=black!78,line width=0.9pt,line cap=round,line join=round
  },
  opent/guitar endpoint/.style={circle,fill=black!78,inner sep=1.45pt},
  opent/guitar interaction/.style={
    circle,draw=opentblue!92!black,fill=opentblue!92!black,
    minimum size=2.6mm,inner sep=0pt,line width=0.6pt
  },
  opent/guitar permutation/.style={
    circle,draw=opentblue!92!black,fill=white,
    minimum size=3.8mm,inner sep=0pt,line width=0.85pt
  }
}

\newcommand{\GuitarDX}{1.28}
\newcommand{\GuitarDY}{0.86}
\newcommand{\GuitarLeft}{0.36}
\newcommand{\GuitarBottom}{-1.08}
\newcommand{\GuitarGap}{0.105}

\newcommand{\DrawGuitarNetwork}[3][]{%
  \begin{scope}[#1]
    \pgfmathtruncatemacro{\GuitarN}{#3}
    \foreach \i in {1,...,\GuitarN} {
      \pgfmathsetmacro{\gyi}{(\GuitarN-\i)*\GuitarDY}
      \pgfmathsetmacro{\gxi}{\i*\GuitarDX}
      \ifnum\i=1
        \draw[opent/guitar wire] (\GuitarLeft,\gyi) -- (\gxi,\gyi);
      \else
        \draw[opent/guitar wire]
          (\GuitarLeft,\gyi) -- ({\GuitarDX-\GuitarGap},\gyi);
        \pgfmathtruncatemacro{\GuitarMiddle}{\i-2}
        \ifnum\GuitarMiddle>0
          \foreach \j in {1,...,\GuitarMiddle} {
            \pgfmathsetmacro{\gxa}{\j*\GuitarDX+\GuitarGap}
            \pgfmathsetmacro{\gxb}{(\j+1)*\GuitarDX-\GuitarGap}
            \draw[opent/guitar wire] (\gxa,\gyi) -- (\gxb,\gyi);
          }
        \fi
        \pgfmathsetmacro{\gxlast}{(\i-1)*\GuitarDX+\GuitarGap}
        \draw[opent/guitar wire] (\gxlast,\gyi) -- (\gxi,\gyi);
      \fi
    }
    \foreach \i in {1,...,\GuitarN} {
      \pgfmathsetmacro{\gyi}{(\GuitarN-\i)*\GuitarDY}
      \pgfmathsetmacro{\gxi}{\i*\GuitarDX}
      \draw[opent/guitar wire] (\gxi,\GuitarBottom) -- (\gxi,\gyi);
      \coordinate (#2-in-\i) at (\gxi,\GuitarBottom);
      \coordinate (#2-out-\i) at (\GuitarLeft,\gyi);
      \coordinate (#2-turn-\i) at (\gxi,\gyi);
    }
  \end{scope}%
}

\newcommand{\LabelGuitarNetwork}[2]{%
  \pgfmathtruncatemacro{\GuitarN}{#2}
  \foreach \i in {1,...,\GuitarN} {
    \node[below=2pt] at (#1-in-\i) {$x_{\i}$};
    \node[left=3pt] at (#1-out-\i) {$z_{\i}$};
  }%
}

\newcommand{\DrawRVertex}[5][]{%
  \begin{scope}[#1,shift={(#2,#3)}]
    \draw[opent/guitar wire] ({-#4},{-#5}) -- (#4,#5);
    \draw[opent/guitar wire] (#4,{-#5}) -- ({-#4},#5);
    \node[opent/guitar interaction] at (0,0) {};
  \end{scope}%
}

\newcommand{\SwapDetourSpan}{0.30}
\newcommand{\SwapDetourLift}{0.12}

%
\newcommand{\DrawPVertex}[5][]{%
  \begin{scope}[#1,shift={(#2,#3)}]
    \pgfmathsetmacro{\pw}{#4}
    \pgfmathsetmacro{\ph}{#5}
    \pgfmathsetmacro{\pa}{\SwapDetourSpan}
    \pgfmathsetmacro{\pl}{\SwapDetourLift}

    \draw[opent/guitar wire]
      ({-\pw},{\ph})
      --
      ({\pw},{-\ph});

    %
    \draw[
      white,
      line width=3.1pt,
      line cap=round,
      line join=round
    ]
      ({-\pa*\pw},{-\pa*\ph})
      .. controls
        ({-0.55*\pa*\pw},{-0.55*\pa*\ph})
        and
        ({(-0.35*\pa-\pl)*\pw},
         {(-0.35*\pa+\pl)*\ph})
      ..
      ({-\pl*\pw},{\pl*\ph})
      .. controls
        ({(0.35*\pa-\pl)*\pw},
         {(0.35*\pa+\pl)*\ph})
        and
        ({0.55*\pa*\pw},{0.55*\pa*\ph})
      ..
      ({\pa*\pw},{\pa*\ph});

    \draw[opent/guitar wire]
      ({-\pw},{-\ph})
      --
      ({-\pa*\pw},{-\pa*\ph})
      .. controls
        ({-0.55*\pa*\pw},{-0.55*\pa*\ph})
        and
        ({(-0.35*\pa-\pl)*\pw},
         {(-0.35*\pa+\pl)*\ph})
      ..
      ({-\pl*\pw},{\pl*\ph})
      .. controls
        ({(0.35*\pa-\pl)*\pw},
         {(0.35*\pa+\pl)*\ph})
        and
        ({0.55*\pa*\pw},{0.55*\pa*\ph})
      ..
      ({\pa*\pw},{\pa*\ph})
      --
      ({\pw},{\ph});
  \end{scope}%
}

\subsection{The guitar map for involutive non-degenerate Yang--Baxter maps}
\label{subsec:guitar}

We first introduce the guitar map for the involutive case and connect it to the LPW
classification in Theorem~\ref{thm:inv-LPW-classification}.

If we have an involutive Yang--Baxter map, then \eqref{eq:yb-permutation-gate} 
yields an involutive permutation operator. It follows that $R$ gives a representation of the symmetric group.
We assumed that $r$ is non-degenerate, therefore $R$ is dual-unitary.  In this case the conditions of the LPW theorem
apply, and the representation is unitarily equivalent to another one given by a ``normal form'' gate $N$. Below we show
that in our case the normal form gate is simply the SWAP gate. Therefore, the LPW theorem implies that the
representation in question is unitarily equivalent to the natural permutation representation on the qudits.

The guitar map is an explicit realization of the LPW intertwiner $U_n$ in
\eqref{eq:inv-LPW-intertwiner} for this class \cite{Deadman2024}. We will see
that the guitar map takes the form of an MPO with fixed bond dimension, and
we will use this special property to obtain strict bounds for the OSR.

Later, in Subsection~\ref{subsec:guitar-general}, we also treat the guitar map for noninvolutive but still non-degenerate Yang--Baxter maps. These
still yield dual-unitary permutation gates, but they are not involutive, so the LPW theorem no longer applies.
Nevertheless, the guitar map can be constructed in essentially the same way.

In Subsection~\ref{subsec:phase-general} we will also consider
arbitrary phase-dressed versions of such Yang--Baxter-map gates.

We first determine the LPW class directly.

\begin{proposition}[LPW class of an involutive non-degenerate map]
\label{prop:yb-involutive-LPW-class}
Let $r$ be a finite non-degenerate involutive Yang--Baxter map on $X$, with
$|X|=D$, and let $R$ be its permutation matrix.  Then
\begin{equation}
 T_R=(\operatorname{id}\otimes\operatorname{Tr})R=\id_{\mathbb C^X}.
 \label{eq:yb-involutive-partial-trace}
\end{equation}
Consequently, the LPW normal form of $R$ is the ordinary swap $\Swap$, with
LPW data
\begin{equation}
 \bigl((1^D),\varnothing\bigr).
 \label{eq:yb-involutive-LPW-data}
\end{equation}
\end{proposition}

\begin{proof}
In the computational basis, a matrix element of the partial trace counts the
colours which close the second strand:
\begin{equation}
 (T_R)_a^{\ c}
 =\#\bigl\{b\in X:\ r(a,b)=(c,b)\bigr\}.
 \label{eq:yb-partial-trace-count}
\end{equation}
Suppose that $r(a,b)=(c,b)$.  Involutivity gives $r(c,b)=(a,b)$.
Comparing the second components of these two collisions yields
\begin{equation}
 \mu_b(a)=b=\mu_b(c).
\end{equation}
Right non-degeneracy makes $\mu_b$ injective, and therefore $a=c$.

It remains to show that for each $a$ there is exactly one closing colour.  Let
$b_a=\lambda_a^{-1}(a)$.  We can write $r(a,b_a)=(a,d)$ for some $d$.
Applying $r$ once more gives $r(a,d)=(a,b_a)$, and comparison of the first
components gives
\begin{equation}
 \lambda_a(d)=a=\lambda_a(b_a).
\end{equation}
Left non-degeneracy implies $d=b_a$.  Thus $(a,b_a)$ is a fixed point of $r$,
and it is the unique pair contributing to
\eqref{eq:yb-partial-trace-count}.  Hence $T_R=\id$.  The LPW normal-form
classification, Theorem~\ref{thm:inv-LPW-classification}, then gives $D$
positive one-dimensional blocks, which is precisely $\Swap$.
\end{proof}

LPW therefore guarantees, at every tensor level $n$, a unitary $U_n$
satisfying \eqref{eq:inv-LPW-generator-intertwiner} with the normal form
$N=\Swap$.  For the present set-theoretical class, the guitar map supplies a
particularly useful explicit choice of $U_n$ for every $n$.  It permutes the
computational basis and can be evaluated by a left-to-right scan carrying only
one permutation of the finite colour set.

We now introduce the guitar map and study its properties. In the remainder of this subsection we work purely
at the level of the classical maps. The extension to linear operators follows directly by lifting these classical
permutations to operators.

The guitar map for length $N$ is a classical map $J_N: X^N\to X^N$ with the
following intertwining property, the classical counterpart of
\eqref{eq:inv-LPW-generator-intertwiner}:
\begin{equation}
 J_Nr_{j,j+1}=p_{j,j+1}J_N, \qquad  j=1,\dots,N-1,
 \label{eq:yb-involutive-guitar-intertwiner}
\end{equation}
where $p_{j,j+1}$ is the classical swap operation.

We now construct the map.  We begin with two sites.  Let
\begin{equation}
 r(x,y)=(u,v),
 \qquad u=\lambda_x(y),\qquad v=\mu_y(x),
 \label{eq:yb-involutive-local-collision}
\end{equation}
and define
\begin{equation}
 J_2(x,y)=\bigl(x,\lambda_x(y)\bigr)=(x,u).
 \label{eq:yb-two-site-J}
\end{equation}
The second new coordinate is the colour that emerges on the left.  Figure~\ref{fig:yb-J2-diagonal-check} records
this elementary 
change of coordinates.  The diagonal arrow is only a bookkeeping device: it
shows that the second entry $y$ is replaced by $u=\lambda_x(y)$, while the
first entry $x$ is carried along unchanged.

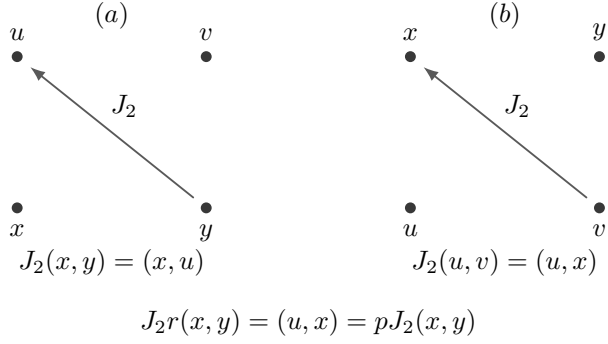
\begin{figure}[t]
\centering
\begin{tikzpicture}[opent figure,x=1cm,y=1cm]
  \node at (1.25,2.55) {$(a)$};
  \node[opent/guitar endpoint,label=below:$x$] (Jax) at (0,0) {};
  \node[opent/guitar endpoint,label=below:$y$] (Jay) at (2.5,0) {};
  \node[opent/guitar endpoint,label=above:$u$] (Jau) at (0,2.0) {};
  \node[opent/guitar endpoint,label=above:$v$] (Jav) at (2.5,2.0) {};
  \draw[opent/flow arrow,shorten <=4pt,shorten >=4pt]
    (Jay) -- (Jau) node[pos=.55,above right=-1pt] {$J_2$};
  \node at (1.25,-0.72) {$J_2(x,y)=(x,u)$};

  \begin{scope}[shift={(5.2,0)}]
    \node at (1.25,2.55) {$(b)$};
    \node[opent/guitar endpoint,label=below:$u$] (Jbu) at (0,0) {};
    \node[opent/guitar endpoint,label=below:$v$] (Jbv) at (2.5,0) {};
    \node[opent/guitar endpoint,label=above:$x$] (Jbx) at (0,2.0) {};
    \node[opent/guitar endpoint,label=above:$y$] (Jby) at (2.5,2.0) {};
    \draw[opent/flow arrow,shorten <=4pt,shorten >=4pt]
      (Jbv) -- (Jbx) node[pos=.55,above right=-1pt] {$J_2$};
    \node at (1.25,-0.72) {$J_2(u,v)=(u,x)$};
  \end{scope}

  \node[align=center] at (3.85,-1.52)
    {$J_2r(x,y)=(u,x)=pJ_2(x,y)$};
\end{tikzpicture}
\caption{The two-site guitar coordinates and the direct involutive check.
Panel (a) uses the first output $u$ of $r(x,y)=(u,v)$ as the second guitar
coordinate.  In panel (b), involutivity gives $r(u,v)=(x,y)$, so the same
construction returns $J_2(u,v)=(u,x)$.  The two dressed colours are therefore
exchanged by the collision.}
\label{fig:yb-J2-diagonal-check}
\end{figure}

Since $r$ is involutive, $r(u,v)=(x,y)$. The second diagram in Figure~\ref{fig:yb-J2-diagonal-check} shows the
application of $J_2$ in this configuration, giving
\begin{equation}
  J_2(u,v)=(u,x).
\end{equation}
It follows that
\begin{equation}
 J_2r(x,y)=J_2(u,v)=\bigl(u,\lambda_u(v)\bigr)
 =(u,x)=pJ_2(x,y),
 \label{eq:yb-two-site-intertwiner}
\end{equation}
where $p(x,u)=(u,x)$ is the ordinary flip on $X^2$. Thus we obtain the $N=2$ case of relation
\eqref{eq:yb-involutive-guitar-intertwiner}. The map $J_2$ is a bijection by
nondegeneracy. 

For $N>2$, it is useful to represent the change of
coordinates as a triangular network. Our graphical notation is inspired by standard conventions in integrability and
the tensor-network literature, especially by Ref.~\cite{LebedVendramin2017}, which introduced closely related notation.

In the diagrams the wires carry the colours $1,\dots,D$, and the crossings can change the colours.
When we draw the triangular diagrams,
the incoming variables are always placed on the bottom or on the right end of the diagrams, and the
output variables are given on the left. This notation is specifically convenient for the guitar map, and will not be
used elsewhere in this document.

We use three different crossings: the guitar-map recolouring defined by $J_2$, the ordinary flip,
and the actual Yang--Baxter map.
The graphical notation is summarized in Figure~\ref{fig:yb-guitar-local-notation}.  At a
guitar crossing, the uninterrupted vertical strand is the controlling colour
and is not modified.  The horizontal strand passes underneath it, and its colour is transformed.
The
horizontal strand is read from right to left.  We draw the physical
Yang--Baxter map $r$ as a genuine crossing with a filled vertex.  By contrast,
the ordinary flip $p$ is drawn as a virtual crossing: the strands are
only permuted and there is no additional colour interaction.

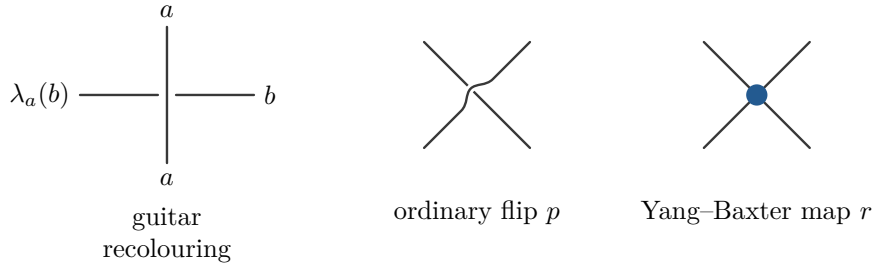
\begin{figure}[t]
\centering
\begin{tikzpicture}[opent figure,x=1cm,y=1cm]
  \begin{scope}[shift={(0,0)}]
    \draw[opent/guitar wire] (0,-0.9) -- (0,0.9);
    \draw[opent/guitar wire] (-1.15,0) -- (-.12,0);
    \draw[opent/guitar wire] (.12,0) -- (1.15,0);
    \node[below] at (0,-0.9) {$a$};
    \node[above] at (0,0.9) {$a$};
    \node[left] at (-1.15,0) {$\lambda_a(b)$};
    \node[right] at (1.15,0) {$b$};
    \node[align=center] at (0,-1.85) {guitar\\recolouring};
  \end{scope}

  \begin{scope}[shift={(4.1,0)}]
    \DrawPVertex{0}{0}{.70}{.70}
    \node at (0,-1.55) {ordinary flip $p$};
  \end{scope}

  \begin{scope}[shift={(7.8,0)}]
    \DrawRVertex{0}{0}{.70}{.70}
    \node at (0,-1.55) {Yang--Baxter map $r$};
  \end{scope}
\end{tikzpicture}
\caption{Local graphical conventions.  In the guitar vertex the vertical
colour $a$ is unchanged, while the underpassing colour $b$ is replaced by
$\lambda_a(b)$.  The crossing with a small bridge denotes a pure permutation of the two
strands, whereas the filled circle denotes the actual local Yang--Baxter
interaction.}
\label{fig:yb-guitar-local-notation}
\end{figure}

The triangular picture of the guitar map $J_2$ is shown in Fig.~\ref{fig:yb-guitar-J2-triangle}.  The $x_2$ strand turns to the left and passes under the uninterrupted
$x_1$ strand, so it acquires precisely the permutation $\lambda_{x_1}$. The intertwining relation
in Eq.~\eqref{eq:yb-two-site-intertwiner} is shown in Fig.~\ref{fig:J2-intertwining-graphical}.

\begin{figure}[t]
\centering
\begin{tikzpicture}[opent figure]
  \DrawGuitarNetwork{gJtwo}{2}
  \LabelGuitarNetwork{gJtwo}{2}
\end{tikzpicture}
\caption{Triangular graphical notation for $J_2$, read from the bottom to the
left.  The first colour is unchanged, $z_1=x_1$.  The second strand passes
under the first one and therefore leaves as
$z_2=\lambda_{x_1}(x_2)$.}
\label{fig:yb-guitar-J2-triangle}
\end{figure}
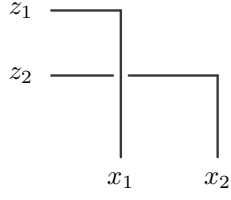

\newcommand{\JTwoSep}{1.55}         
\newcommand{\JTwoMid}{0.775}        
\newcommand{\JTwoStub}{0.78}        
\newcommand{\JTwoGap}{0.11}         
\newcommand{\JTwoTop}{0.95}         
\newcommand{\JTwoMidY}{0.00}        
\newcommand{\JTwoBottom}{0.95}      
\newcommand{\JTwoRHalfY}{0.55}      
\newcommand{\JTwoPHalfY}{0.55}      
\newcommand{\JTwoHalfX}{0.775}      

\newcommand{\DrawJTwo}[1]{%
  \begin{scope}[shift={#1}]
    \draw[opent/guitar wire] (0,\JTwoTop) -- (\JTwoSep,\JTwoTop);
    \draw[opent/guitar wire] (\JTwoSep,\JTwoTop) -- (\JTwoSep,\JTwoMidY-\JTwoBottom);

    \draw[opent/guitar wire] (0,\JTwoMidY) -- (\JTwoSep-\JTwoGap,\JTwoMidY);
    \draw[opent/guitar wire] (\JTwoSep+\JTwoGap,\JTwoMidY) -- (2*\JTwoSep,\JTwoMidY);
    \draw[opent/guitar wire] (2*\JTwoSep,\JTwoMidY) -- (2*\JTwoSep,\JTwoMidY-\JTwoBottom);

    \node[left]  at (0,\JTwoTop) {$z_1$};
    \node[left]  at (0,\JTwoMidY) {$z_2$};
    \node[below] at (\JTwoSep,\JTwoMidY-\JTwoBottom) {$x_1$};
    \node[below] at (2*\JTwoSep,\JTwoMidY-\JTwoBottom) {$x_2$};
  \end{scope}
}

\begin{figure}[t]
\centering
\begin{tikzpicture}[x=1cm,y=1cm,baseline=(current bounding box.center)]

  \begin{scope}[shift={(-4.2,0)}]
    \node at (1.55,1.55) {$J_2r$};

    \DrawJTwo{(0,0)}

    \DrawRVertex
      {1.5*\JTwoSep}
      {\JTwoMidY-\JTwoBottom-\JTwoRHalfY}
      {\JTwoHalfX}
      {\JTwoRHalfY}
  \end{scope}

  \node at (0,0) {$=$};

  \begin{scope}[shift={(2.1,0)}]
    \node at (1.55,1.55) {$pJ_2$};

    \DrawJTwo{(0,0)}

    \DrawPVertex
      {-\JTwoHalfX}
      {0.5*(\JTwoTop+\JTwoMidY)}
      {\JTwoHalfX}
      {0.5*(\JTwoTop-\JTwoMidY)}
  \end{scope}

\end{tikzpicture}
\caption{The elementary two-site intertwining relation \(J_2r=pJ_2\), read from bottom to top.}
\label{fig:J2-intertwining-graphical}
\end{figure}
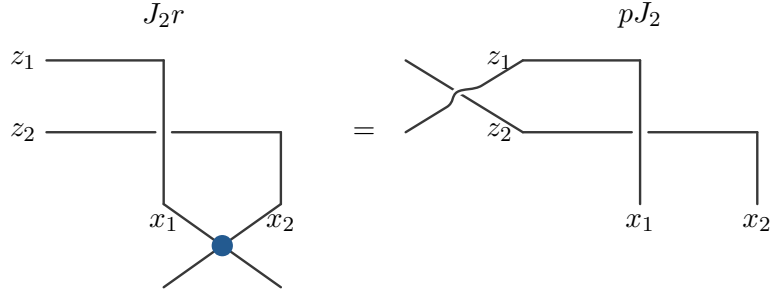

For $N=3$ we introduce the guitar map with the explicit formula
\begin{equation}
 J_3(x_1,x_2,x_3)=
 \bigl(x_1,\lambda_{x_1}(x_2),
 \lambda_{x_1}\lambda_{x_2}(x_3)\bigr).
 \label{eq:yb-involutive-J3}
\end{equation}
Writing the output as $(z_1,z_2,z_3)$ we can draw this map as the triangular network in
Figure~\ref{fig:yb-guitar-J3-triangle}.  The third strand first passes under
$x_2$ and then under $x_1$.  Since the picture is read from the bottom to the
left, the order of these two crossings gives
$z_3=\lambda_{x_1}\lambda_{x_2}(x_3)$, in agreement with
\eqref{eq:yb-involutive-J3}.

\begin{figure}[t]
\centering
\begin{tikzpicture}[opent figure]
  \DrawGuitarNetwork{gJthree}{3}
  \LabelGuitarNetwork{gJthree}{3}
\end{tikzpicture}
\caption{Triangular graphical notation for $J_3$.  The uninterrupted
vertical strands carry the controlling colours.  Each horizontal underpass
applies one further $\lambda$-permutation, so
$z_1=x_1$, $z_2=\lambda_{x_1}(x_2)$, and
$z_3=\lambda_{x_1}\lambda_{x_2}(x_3)$.}
\label{fig:yb-guitar-J3-triangle}
\end{figure}
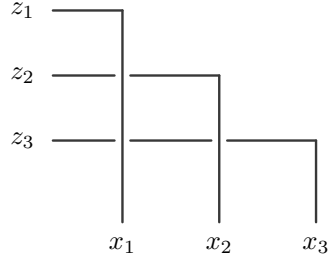

For $N=3$, Eq.~\eqref{eq:yb-involutive-guitar-intertwiner} reduces to the two identities
\begin{equation}
  p_{12}J_3=J_3r_{12},\qquad  p_{23}J_3=J_3r_{23}.
\end{equation}
They are depicted in Figs.~\ref{fig:yb-guitar-J3-first} and~\ref{fig:yb-guitar-J3-second}.
We prove them diagrammatically using two ``sliding moves''. 

\begin{figure}[t]
\centering
\begin{tikzpicture}[opent figure,x=1cm,y=1cm]
  \DrawGuitarNetwork{gFirstProofStart}{3}
  \pgfmathsetmacro{\gFirstProofCX}{1.5*\GuitarDX}
  \pgfmathsetmacro{\gFirstProofCY}{\GuitarBottom-0.60}
  \DrawRVertex
    {\gFirstProofCX}{\gFirstProofCY}{0.5*\GuitarDX}{0.60}
  \draw[opent/guitar wire]
    (gFirstProofStart-in-3) -- ++(0,-1.20);
  \node[below=2pt] at ($(gFirstProofStart-in-1)+(0,-1.20)$) {$x_1$};
  \node[below=2pt] at ($(gFirstProofStart-in-2)+(0,-1.20)$) {$x_2$};
  \node[below=2pt] at ($(gFirstProofStart-in-3)+(0,-1.20)$) {$x_3$};
  \node at (2.15,2.45) {$J_3r_{12}$};

  \node at (4.55,0.35) {$=$};

  \begin{scope}[shift={(5.20,0)}]
    \pgfmathsetmacro{\gProofXOne}{\GuitarDX}
    \pgfmathsetmacro{\gProofXTwo}{2*\GuitarDX}
    \pgfmathsetmacro{\gProofXThree}{3*\GuitarDX}
    \pgfmathsetmacro{\gProofYOne}{2*\GuitarDY}
    \pgfmathsetmacro{\gProofYTwo}{\GuitarDY}
    \pgfmathsetmacro{\gProofRBuffer}{0.14}
    \pgfmathsetmacro{\gProofRBottom}{\gProofRBuffer}
    \pgfmathsetmacro{\gProofRTop}{\GuitarDY-\gProofRBuffer}
    \pgfmathsetmacro{\gProofRHalfY}{0.5*\GuitarDY-\gProofRBuffer}

    \draw[opent/guitar wire]
      (\gProofXOne,\GuitarBottom) -- (\gProofXOne,\gProofRBottom);
    \draw[opent/guitar wire]
      (\gProofXTwo,\GuitarBottom) -- (\gProofXTwo,\gProofRBottom);
    \draw[opent/guitar wire]
      (\gProofXThree,\GuitarBottom) -- (\gProofXThree,0);

    \draw[opent/guitar wire]
      (\GuitarLeft,0) -- ({\GuitarDX-\GuitarGap},0);
    \draw[opent/guitar wire]
      ({\GuitarDX+\GuitarGap},0)
      -- ({2*\GuitarDX-\GuitarGap},0);
    \draw[opent/guitar wire]
      ({2*\GuitarDX+\GuitarGap},0) -- (3*\GuitarDX,0);

    \draw[opent/guitar wire]
      (\GuitarLeft,\gProofYOne) -- (\gProofXOne,\gProofYOne);
    \draw[opent/guitar wire]
      (\gProofXOne,\gProofRTop) -- (\gProofXOne,\gProofYOne);
    \draw[opent/guitar wire]
      (\gProofXTwo,\gProofRTop) -- (\gProofXTwo,\gProofYTwo);
    \draw[opent/guitar wire]
      (\GuitarLeft,\gProofYTwo)
      -- ({\GuitarDX-\GuitarGap},\gProofYTwo);
    \draw[opent/guitar wire]
      ({\GuitarDX+\GuitarGap},\gProofYTwo)
      -- (\gProofXTwo,\gProofYTwo);

    \DrawRVertex
      {1.5*\GuitarDX}{0.5*\GuitarDY}
      {0.5*\GuitarDX}{\gProofRHalfY}

    \node[below=2pt] at (\gProofXOne,\GuitarBottom) {$x_1$};
    \node[below=2pt] at (\gProofXTwo,\GuitarBottom) {$x_2$};
    \node[below=2pt] at (\gProofXThree,\GuitarBottom) {$x_3$};
    \node at (2.15,2.45) {after one sliding move};
  \end{scope}

  \node at (9.84,0.35) {$=$};

  \DrawGuitarNetwork[shift={(11.7,0)}]{gFirstProofFinal}{3}
  \pgfmathsetmacro{\gFirstProofPY}{1.5*\GuitarDY}
  \pgfmathsetmacro{\gFirstProofPX}{11.7+\GuitarLeft-0.71}
  \DrawPVertex
    {\gFirstProofPX}{\gFirstProofPY}{0.71}{0.5*\GuitarDY}
  \draw[opent/guitar wire]
    (gFirstProofFinal-out-3) -- ++(-1.42,0);
  \node[below=2pt] at (gFirstProofFinal-in-1) {$x_1$};
  \node[below=2pt] at (gFirstProofFinal-in-2) {$x_2$};
  \node[below=2pt] at (gFirstProofFinal-in-3) {$x_3$};
  \node at (13.85,2.45) {$p_{12}J_3$};
\end{tikzpicture}
\caption{Graphical proof of $J_3r_{12}=p_{12}J_3$.  The physical crossing is
inserted on the first two bare inputs.  We first apply one sliding move, shown
in Fig.~\ref{fig:guitar-slide-r}, and then use the defining two-site identity
$J_2r=pJ_2$ shown in Fig.~\ref{fig:J2-intertwining-graphical}.}
\label{fig:yb-guitar-J3-first}
\end{figure}
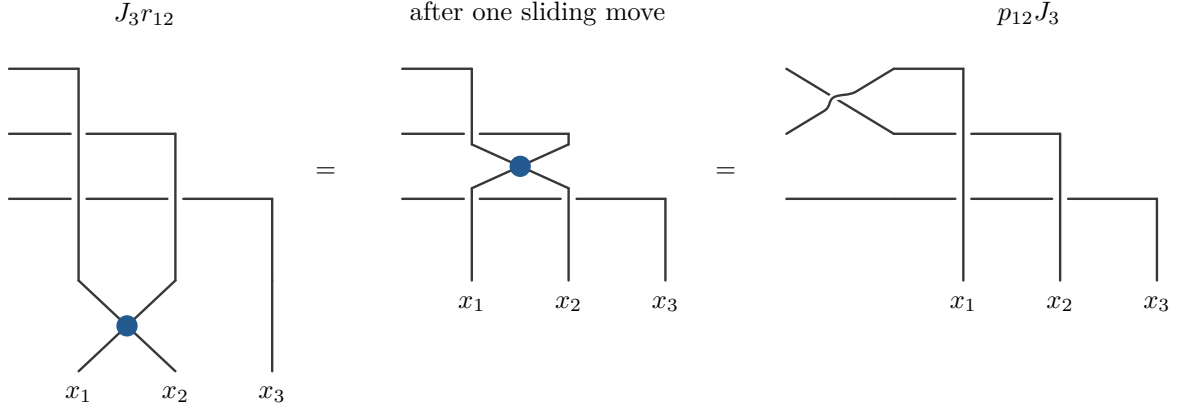

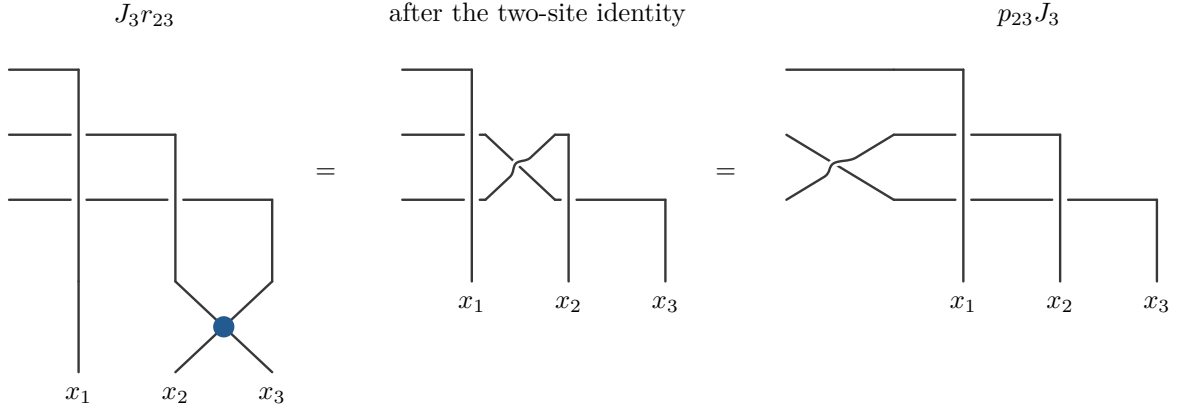
\begin{figure}[t]
\centering
\begin{tikzpicture}[opent figure,x=1cm,y=1cm]
  \DrawGuitarNetwork{gSecondProofStart}{3}
  \draw[opent/guitar wire]
    (gSecondProofStart-in-1) -- ++(0,-1.20);
  \pgfmathsetmacro{\gSecondCX}{2.5*\GuitarDX}
  \pgfmathsetmacro{\gSecondCY}{\GuitarBottom-0.60}
  \DrawRVertex{\gSecondCX}{\gSecondCY}{0.5*\GuitarDX}{0.60}
  \node[below=2pt] at ($(gSecondProofStart-in-1)+(0,-1.20)$) {$x_1$};
  \node[below=2pt] at ($(gSecondProofStart-in-2)+(0,-1.20)$) {$x_2$};
  \node[below=2pt] at ($(gSecondProofStart-in-3)+(0,-1.20)$) {$x_3$};
  \node at (2.15,2.45) {$J_3r_{23}$};

  \node at (4.55,0.35) {$=$};

  \begin{scope}[shift={(5.20,0)}]
    \pgfmathsetmacro{\gSecondProofXOne}{\GuitarDX}
    \pgfmathsetmacro{\gSecondProofXTwo}{2*\GuitarDX}
    \pgfmathsetmacro{\gSecondProofXThree}{3*\GuitarDX}
    \pgfmathsetmacro{\gSecondProofYOne}{2*\GuitarDY}
    \pgfmathsetmacro{\gSecondProofYTwo}{\GuitarDY}
    \pgfmathsetmacro{\gSecondProofPCX}{1.5*\GuitarDX}
    \pgfmathsetmacro{\gSecondProofPCY}{0.5*\GuitarDY}
    \pgfmathsetmacro{\gSecondProofPHalfX}{0.46}

    \draw[opent/guitar wire]
      (\gSecondProofXOne,\GuitarBottom)
      -- (\gSecondProofXOne,\gSecondProofYOne);
    \draw[opent/guitar wire]
      (\gSecondProofXTwo,\GuitarBottom)
      -- (\gSecondProofXTwo,\gSecondProofYTwo);
    \draw[opent/guitar wire]
      (\gSecondProofXThree,\GuitarBottom)
      -- (\gSecondProofXThree,0);

    \draw[opent/guitar wire]
      (\GuitarLeft,\gSecondProofYOne)
      -- (\gSecondProofXOne,\gSecondProofYOne);
    \draw[opent/guitar wire]
      (\GuitarLeft,\gSecondProofYTwo)
      -- ({\GuitarDX-\GuitarGap},\gSecondProofYTwo);
    \draw[opent/guitar wire]
      ({\GuitarDX+\GuitarGap},\gSecondProofYTwo)
      -- ({\gSecondProofPCX-\gSecondProofPHalfX},\gSecondProofYTwo);
    \draw[opent/guitar wire]
      (\GuitarLeft,0) -- ({\GuitarDX-\GuitarGap},0);
    \draw[opent/guitar wire]
      ({\GuitarDX+\GuitarGap},0)
      -- ({\gSecondProofPCX-\gSecondProofPHalfX},0);

    \DrawPVertex
      {\gSecondProofPCX}{\gSecondProofPCY}
      {\gSecondProofPHalfX}{0.5*\GuitarDY}

    \draw[opent/guitar wire]
      ({\gSecondProofPCX+\gSecondProofPHalfX},\gSecondProofYTwo)
      -- (\gSecondProofXTwo,\gSecondProofYTwo);
    \draw[opent/guitar wire]
      ({\gSecondProofPCX+\gSecondProofPHalfX},0)
      -- ({2*\GuitarDX-\GuitarGap},0);
    \draw[opent/guitar wire]
      ({2*\GuitarDX+\GuitarGap},0)
      -- (\gSecondProofXThree,0);

    \node[below=2pt] at (\gSecondProofXOne,\GuitarBottom) {$x_1$};
    \node[below=2pt] at (\gSecondProofXTwo,\GuitarBottom) {$x_2$};
    \node[below=2pt] at (\gSecondProofXThree,\GuitarBottom) {$x_3$};
    \node at (2.15,2.45) {after the two-site identity};
  \end{scope}

  \node at (9.84,0.35) {$=$};

  \DrawGuitarNetwork[shift={(11.7,0)}]{gSecondProofFinal}{3}
  \draw[opent/guitar wire]
    (gSecondProofFinal-out-1) -- ++(-1.42,0);
  \pgfmathsetmacro{\gSecondPY}{0.5*\GuitarDY}
  \pgfmathsetmacro{\gSecondPX}{11.7+\GuitarLeft-0.71}
  \DrawPVertex{\gSecondPX}{\gSecondPY}{0.71}{0.5*\GuitarDY}
  \node[below=2pt] at (gSecondProofFinal-in-1) {$x_1$};
  \node[below=2pt] at (gSecondProofFinal-in-2) {$x_2$};
  \node[below=2pt] at (gSecondProofFinal-in-3) {$x_3$};
  \node at (13.85,2.45) {$p_{23}J_3$};
\end{tikzpicture}
\caption{Graphical proof of $J_3r_{23}=p_{23}J_3$.  The physical crossing is
inserted on the last two bare inputs.  We first use the defining two-site
identity $J_2r=pJ_2$ shown in Fig.~\ref{fig:J2-intertwining-graphical}, and
then apply one sliding move, shown in Fig.~\ref{fig:guitar-slide-p}.}
\label{fig:yb-guitar-J3-second}
\end{figure}

\newcommand{\SlideSep}{1.55}         
\newcommand{\SlideMid}{0.775}        
\newcommand{\SlideStub}{0.78}        
\newcommand{\SlideGap}{0.11}         

\newcommand{\SlideRHalfY}{0.55}      
\newcommand{\SlideTop}{\SlideRHalfY} 
\newcommand{\SlideBot}{\SlideRHalfY} 

\newcommand{\SlideDoubleY}{0.30}     
\newcommand{\SlideSwapHalfY}{\SlideDoubleY} 
\newcommand{\SlideSwapHalfX}{0.70}   

\newcommand{\DrawSingleGuitarCross}[2]{%
  \draw[opent/guitar wire] (#1,\SlideTop) -- (#1,-\SlideBot);
  \draw[opent/guitar wire] (#1-\SlideStub,#2) -- (#1-\SlideGap,#2);
  \draw[opent/guitar wire] (#1+\SlideGap,#2) -- (#1+\SlideStub,#2);
}

\newcommand{\DrawDoubleGuitarCross}[1]{%
  \draw[opent/guitar wire] (#1,0.92) -- (#1,-0.92);
  \draw[opent/guitar wire] (#1-\SlideStub,\SlideDoubleY) -- (#1-\SlideGap,\SlideDoubleY);
  \draw[opent/guitar wire] (#1+\SlideGap,\SlideDoubleY) -- (#1+\SlideStub,\SlideDoubleY);
  \draw[opent/guitar wire] (#1-\SlideStub,-\SlideDoubleY) -- (#1-\SlideGap,-\SlideDoubleY);
  \draw[opent/guitar wire] (#1+\SlideGap,-\SlideDoubleY) -- (#1+\SlideStub,-\SlideDoubleY);
}

\begin{figure}[t]
\centering
\begin{tikzpicture}[x=1cm,y=1cm,baseline=(current bounding box.center)]

  \begin{scope}[shift={(0,0)}]
    \DrawSingleGuitarCross{0}{0}
    \DrawSingleGuitarCross{\SlideSep}{0}

    \DrawRVertex{\SlideMid}{-(\SlideBot+\SlideRHalfY)}{\SlideMid}{\SlideRHalfY}
  \end{scope}

  \node at (3.60,0) {$=$};

  \begin{scope}[shift={(5.20,0)}]
    \DrawRVertex{\SlideMid}{\SlideTop+\SlideRHalfY}{\SlideMid}{\SlideRHalfY}

    \DrawSingleGuitarCross{0}{0}
    \DrawSingleGuitarCross{\SlideSep}{0}
  \end{scope}

\end{tikzpicture}
\caption{Sliding the local Yang--Baxter interaction through two guitar crossings.}
\label{fig:guitar-slide-r}
\end{figure}
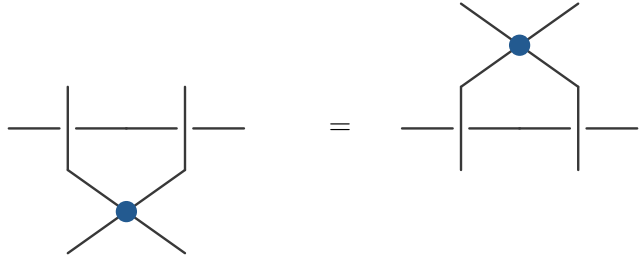

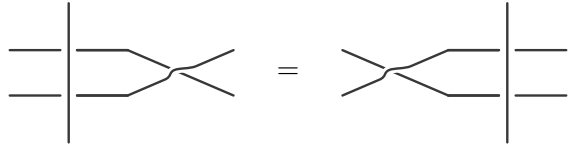
\begin{figure}[t]
\centering
\begin{tikzpicture}[x=1cm,y=1cm,baseline=(current bounding box.center)]

  \begin{scope}[shift={(-2.90,0)}]
    \DrawDoubleGuitarCross{0}

    \draw[opent/guitar wire]
      (\SlideGap,\SlideDoubleY) -- ({1.48-\SlideSwapHalfX},\SlideDoubleY);
    \draw[opent/guitar wire]
      (\SlideGap,-\SlideDoubleY) -- ({1.48-\SlideSwapHalfX},-\SlideDoubleY);

    \DrawPVertex{1.48}{0}{\SlideSwapHalfX}{\SlideSwapHalfY}
  \end{scope}

  \node at (0,0) {$=$};

  \begin{scope}[shift={(2.90,0)}]
    \DrawPVertex{-1.48}{0}{\SlideSwapHalfX}{\SlideSwapHalfY}

    \draw[opent/guitar wire]
      ({-1.48+\SlideSwapHalfX},\SlideDoubleY) -- (-\SlideGap,\SlideDoubleY);
    \draw[opent/guitar wire]
      ({-1.48+\SlideSwapHalfX},-\SlideDoubleY) -- (-\SlideGap,-\SlideDoubleY);

    \DrawDoubleGuitarCross{0}
  \end{scope}

\end{tikzpicture}
\caption{Sliding the ordinary swap through two guitar crossings on a single vertical strand.}
\label{fig:guitar-slide-p}
\end{figure}

The first move (depicted in Fig.~\ref{fig:guitar-slide-r}) tells us that a Yang--Baxter map acting on two vertical
lines can be pulled through an additional 
horizontal wire. Since the colours of the vertical lines are not modified by the crossing with the horizontal line, the
only non-trivial effect here is the transformation of the colour of the horizontal line. 
Using again the notation $r(x,y)=(u,v)$, we need to prove the relation
\begin{equation}
 \lambda_x\lambda_y=\lambda_u\lambda_v.
 \label{eq:yb-lambda-product}
\end{equation}
This identity is simply one component of the original Yang--Baxter relation for $r$. Consider the action of
$r_{12}r_{23}r_{12}$ on a triple $(x,y,z)$: the first component of the result becomes $\lambda_u\lambda_v(z)$. Consider
now the action $r_{23}r_{12}r_{23}$ on $(x,y,z)$: the first component becomes $\lambda_x\lambda_y(z)$. This proves
\eqref{eq:yb-lambda-product} and thereby proves the sliding move in Fig.~\ref{fig:guitar-slide-r}.

The second sliding move (depicted in Fig.~\ref{fig:guitar-slide-p}) tells us that permutations on horizontal lines
can be pulled through the crossings on vertical lines. This step follows in a straightforward way, because crossing with
the vertical line applies the same colour transformation for both horizontal lines, and the colour of the vertical line
is not changed in either case.

We can then use these sliding moves to prove the desired intertwining relations for $J_3$.

We now define the transformation on an arbitrary word
$x=(x_1,\ldots,x_N)$. The final formula is
\begin{equation}
 J_N(x)=\bigl(x_1,\lambda_{x_1}(x_2),
 \lambda_{x_1}\lambda_{x_2}(x_3),\ldots\bigr).
 \label{eq:yb-guitar-expanded}
\end{equation}
Alternatively, we obtain the same result
by a left-to-right
scan carrying one running permutation $g$:
\begin{align}
 g_0&=1,&
 J_N(x)_k&=g_{k-1}(x_k),&
 g_k&=g_{k-1}\lambda_{x_k}.
 \label{eq:yb-guitar}
\end{align}
The same scan also gives the inverse.  If $(z_1,z_2,\dots,z_N)=J_N(x_1,x_2,\dots,x_N)$ then we get
\begin{equation}
 x_k=g_{k-1}^{-1}(z_k),
 \qquad
 g_k=g_{k-1}\lambda_{x_k}.
 \label{eq:yb-guitar-inverse}
\end{equation}
Thus $J_N$ is a bijection.  This argument uses only left
non-degeneracy.

\begin{proposition}[Explicit LPW intertwiner]
\label{prop:yb-involutive-guitar}
For every $N\geq2$ and every $i=1,\ldots,N-1$, Eq.~\eqref{eq:yb-involutive-guitar-intertwiner} holds.
Equivalently, the permutation unitary
\begin{equation}
 \cJ_N\ket{x_1,\ldots,x_N}
 =\ket{J_N(x_1,\ldots,x_N)}
 \label{eq:yb-linearized-guitar}
\end{equation}
satisfies
\begin{equation}
 \cJ_NR_i=\Swap_i\cJ_N.
 \label{eq:yb-linearized-involutive-intertwiner}
\end{equation}
Hence $U_N=\cJ_N$ is an explicit choice in
\eqref{eq:inv-LPW-generator-intertwiner} for the swap normal form.
\end{proposition}

\begin{proof}
We provide a diagrammatic proof.  All required steps are already displayed for $N=3$ in
Figs.~\ref{fig:yb-guitar-J3-first} and~\ref{fig:yb-guitar-J3-second}.  For a
general index $i$, we slide the physical crossing upwards until it reaches the
corresponding two-site block, apply $J_2r=pJ_2$, and then slide the resulting
permutation crossing to the left.  This gives $J_Nr_i=p_iJ_N$, and
linearization gives Eq.~\eqref{eq:yb-linearized-involutive-intertwiner}.
\end{proof}

It is important to distinguish $J_N$ from an onsite change of basis.
The final $k$th coordinate depends on all colours to its left; therefore, $J_N$
typically does not factor into one-site transformations.  Its nonlocality is,
however, very restricted. The definition in Eq.~\eqref{eq:yb-guitar} implies that for the coordinate $x_k$ all the dependence on
the variables to the left of it is compressed into the permutation $g_{k-1}$. This observation allows us to present the
guitar map as a ``matrix product type'' mapping, such that the linearized version $\cJ_N$ actually becomes a Matrix
Product Operator (MPO).

This alternative notation for the guitar map is presented in Fig.~\ref{fig:yb-guitar-MPO}. This diagram simply
reflects the definition in Eq.~\eqref{eq:yb-guitar} and it is not a convenient representation for proving the intertwining
relation. However, we can read off the matrix product structure, and we observe a finite bond dimension.
\begin{figure}[t]
\centering
\begin{tikzpicture}[
  opent figure,
  x=1cm,y=1cm,
  guitar tensor/.style={
    opent/gate,minimum width=10mm,minimum height=8mm,
    font=\small
  },
  memory arrow/.style={
    -{Latex[length=1.8mm]},draw=opentblue!92!black,line width=0.85pt
  },
  physical arrow/.style={
    -{Latex[length=1.7mm]},draw=black!78,line width=0.75pt
  }
]
  \node[guitar tensor] (A1) at (0,0) {$A$};
  \node[guitar tensor] (A2) at (1.85,0) {$A$};
  \node at (3.15,0) {$\cdots$};
  \node[guitar tensor] (ANm) at (4.45,0) {$A$};
  \node[guitar tensor] (AN) at (6.30,0) {$A$};

  \draw[memory arrow] (-1.45,0)
    node[left,font=\small] {$g_0=1$} -- (A1.west);
  \draw[memory arrow] (A1.east) --
    node[above,font=\scriptsize] {$g_1$} (A2.west);
  \draw[memory arrow] (A2.east) -- (2.70,0);
  \draw[memory arrow] (3.60,0) --
    node[above,pos=0,xshift=-1mm,font=\scriptsize] {$g_{N-2}$} (ANm.west);
  \draw[memory arrow] (ANm.east) --
    node[above,font=\scriptsize] {$g_{N-1}$} (AN.west);
  \draw[memory arrow] (AN.east) -- (7.60,0)
    node[right,font=\small] {$g_N$};

  \draw[physical arrow] (0,-1.05)
    node[below,font=\small] {$x_1$} -- (A1.south);
  \draw[physical arrow] (A1.north) -- (0,1.05)
    node[above,font=\small] {$z_1$};
  \draw[physical arrow] (1.85,-1.05)
    node[below,font=\small] {$x_2$} -- (A2.south);
  \draw[physical arrow] (A2.north) -- (1.85,1.05)
    node[above,font=\small] {$z_2$};
  \draw[physical arrow] (4.45,-1.05)
    node[below,font=\small] {$x_{N-1}$} -- (ANm.south);
  \draw[physical arrow] (ANm.north) -- (4.45,1.05)
    node[above,font=\small] {$z_{N-1}$};
  \draw[physical arrow] (6.30,-1.05)
    node[below,font=\small] {$x_N$} -- (AN.south);
  \draw[physical arrow] (AN.north) -- (6.30,1.05)
    node[above,font=\small] {$z_N$};

  \node[opent/annotation,align=center] at (3.15,-1.68)
  {$z_k=g_{k-1}(x_k)$,\qquad
   $g_k=g_{k-1}\lambda_{x_k}$,\qquad
   $g_k\in K_\lambda$};
\end{tikzpicture}
\caption{Finite-memory matrix-product representation of the guitar
map $\cJ_N$.  The vertical legs carry the input and output colours,
while the horizontal leg carries the accumulated single site rotation.  All
information about the prefix is compressed into one element
$g\in K_\lambda$.}
\label{fig:yb-guitar-MPO}
\end{figure}
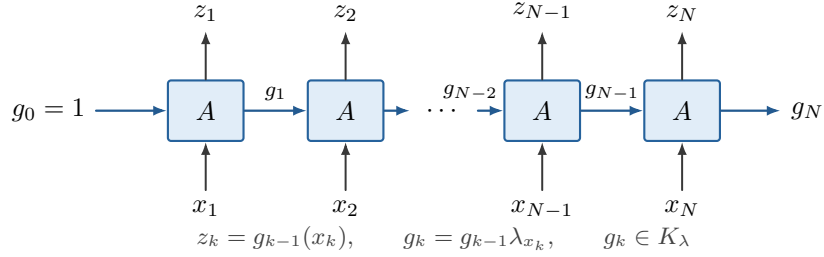

This bond dimension is given by counting the possibilities for the local permutations $g_k$. 
To formalize this we introduce the subgroup $K_\lambda$ of the symmetric group on $D$ elements  given by
\begin{equation}
 K_\lambda=\langle\lambda_x:x\in X\rangle
 \leq\operatorname{Sym}(X).
 \label{eq:yb-Klambda}
\end{equation}
Across any spatial cut, the left and right parts of
Figure~\ref{fig:yb-guitar-MPO} share only the bond label $g$.  The bond
dimension is therefore $|K_\lambda|$, independently of $N$.

This special property is not contained in the abstract LPW classification.  It is the extra information
which makes the guitar intertwiner useful for the operator-entanglement
problem.  Subsection~\ref{subsec:permutation-osr} converts this finite bond
label into the corresponding operator-Schmidt-rank bound.

\subsection{The general noninvolutive case}
\label{subsec:guitar-general}

We now remove the involutivity constraint while retaining nondegeneracy.
In this case the LPW theorem no longer applies; nevertheless, there is an
appropriate guitar map which considerably simplifies the action of the Yang--Baxter maps.
This extension of the guitar construction was developed
by Soloviev and by Lu, Yan, and Zhu
\cite{Soloviev2000,LuYanZhu2000}.  A modern graphical treatment, including
the associated shelf or rack and the guitar intertwiner, is given in
Ref.~\cite{LebedVendramin2017}.

In this subsection we discuss the differences from the involutive case, but we will not present the proofs in
the same detail as in the previous subsection.

The guitar map is defined via the same construction as before, with identical formulas. In particular, the two-site map
$J_2$ is defined via \eqref{eq:yb-guitar}, and also the definition \eqref{eq:yb-guitar-expanded}
and the corresponding triangular networks and MPOs (for example,
Figures~\ref{fig:yb-guitar-J2-triangle} and
\ref{fig:yb-guitar-J3-triangle}) are unchanged.

What changes is the local vertex which emerges after the pull-through.  It is no longer necessarily the
permutation $p$: a residual colour rotation remains. We now determine this residual interaction.

Write
\begin{equation}
 r(a,b)=\bigl(\lambda_a(b),\mu_b(a)\bigr)=(u,v).
 \label{eq:yb-rho-raw-collision}
\end{equation}
The coordinate $u$ is the colour which emerges on the left.
For fixed $a$ and $u$, non-degeneracy reconstructs the bare incoming colour as
$b=\lambda_a^{-1}(u)$.

The composition of $J_2$ and $r$ gives
\begin{equation}
  J_2r(a,b)=J_2(u,v)=\bigl(u,\lambda_u(v)\bigr).
\end{equation}
We write this result as the composition of a two-site map $s$ and $J_2$:
\begin{equation}
  sJ_2(a,b)=s(a,u).
\end{equation}
Comparing these two expressions gives
\begin{equation}
  s(a,u)=\bigl(u,\lambda_u(v)\bigr).
\end{equation}
We separate this map into a permutation and a controlled rotation:
\begin{equation}
  s(a,u)=(u,\rho_u(a)).
  \label{eq:yb-derived-controlled-swap}
\end{equation}
Substituting $v=\mu_b(a)$ and $b=\lambda^{-1}_a(u)$ gives
\begin{equation}
  \rho_u(a)=\lambda_u\left(\mu_{\lambda^{-1}_a(u)}(a)\right).
\end{equation}
Thus we obtain the explicit form of $s=J_2rJ_2^{-1}$.
If $r$ is involutive, then $r(u,v)=(a,b)$ and the formula gives $\rho_u(a)=a$.

After linearization, Eq.~\eqref{eq:yb-derived-controlled-swap} is
precisely the right-controlled-swap gate
introduced in Subsubsection~\ref{subsec:controlled-swap-right}.  More
explicitly, the one-site unitary $u_b$ in
Eq.~\eqref{eq:controlled-swap-two-orientations} is now the permutation
matrix $P_{\rho_b}$ defined by
$P_{\rho_b}\ket a=\ket{\rho_b(a)}$.  Thus the local scattering mechanism in
guitar coordinates is the controlled-swap mechanism already encountered in
Section~\ref{sec:controlled-swap}.  The new ingredient is that we reach this
form through the non-onsite guitar transformation.

The general guitar-map intertwining theorem is Proposition~6.2(4) of
Ref.~\cite{LebedVendramin2017}; translated to our conventions, it reads
\begin{equation}
 J_Nr_{j,j+1}=s_{j,j+1}J_N, \qquad  j=1,\dots,N-1,
 \label{eq:yb-non-involutive-guitar-intertwiner}
\end{equation}
After linearization, the same equation gives a unitary equivalence of the two
braid-group representations at every tensor level.  In contrast to the
involutive case, this is not an LPW equivalence to $\Swap$. Instead, we obtain unitary equivalence with a circuit that
performs controlled single-site rotations.

We note that the map $s$ also satisfies the Yang--Baxter equation, which follows from the intertwining property and the
invertibility of $J_N$.

We define the finite permutation group $K_\lambda$ in the same way as above
and another group $K_\rho$ by
\begin{equation}
 K_\rho=\langle\rho_x:x\in X\rangle
 \quad\leq\operatorname{Sym}(X).
\end{equation}
Their cardinalities are denoted by
\begin{equation}
 \kappa_\lambda=|K_\lambda|,
 \quad
 \kappa_\rho=|K_\rho|.
 \label{eq:yb-KH}
\end{equation}
The elements of $K_\rho$ are the possible net rotations
accumulated from the single-site controlled permutations.
For an involutive map, the identity $\rho_u(a)=a$ above holds for every
$a,u\in X$.  Hence $\rho_u=\id_X$ for every $u$, so
$K_\rho=\{\id_X\}$, $\kappa_\rho=1$, and $s$ reduces to the ordinary swap.

The next subsection,
Subsection~\ref{subsec:permutation-osr}, converts these two finite memories
into a time-independent bound on the operator Schmidt rank.

\subsection{OSR bounds for permutation gates}
\label{subsec:permutation-osr}

In this subsection we treat the general, not necessarily involutive case and
indicate explicitly how the argument specializes when $r$ is involutive.  We
turn the two finite colour memories identified above into a uniform bound on
operator entanglement. The guitar transformation carries one frame label in
$K_\lambda$, while the dynamics in guitar coordinates carries one net colour
permutation in $K_\rho$. Neither memory grows with the block length $t$. The
proof has three steps: we bound the coordinate change, the controlled
permutations, and finally combine the two estimates in the order
``dress, scatter, undress.''

The scattering step has a direct connection with
Section~\ref{sec:controlled-swap}, because the gate in
Eq.~\eqref{eq:yb-derived-controlled-swap} is a controlled permutation.  We
therefore reuse the controlled-swap calculation for this step.  The
additional work comes from the two guitar transformations, which dress the
local source before the scattering and return the result to the physical
colours afterwards.

Let us first recall the coordinate change. The scans
\eqref{eq:yb-guitar} and \eqref{eq:yb-guitar-inverse}, together with the
matrix-product representation in Figure~\ref{fig:yb-guitar-MPO}, pass only
one auxiliary state $g\in K_\lambda$ from one site to the next. Therefore
\begin{equation}
 \OSR(\cJ_N),\ \OSR(\cJ_N^{-1})
 \leq |K_\lambda|=\kappa_\lambda,
 \label{eq:yb-J-rank}
\end{equation}
uniformly in $N$.

We next consider the scattering of two length-$t$ words in guitar
coordinates,
\begin{equation*}
 A=(a_1,\ldots,a_t),
 \qquad
 B=(b_1,\ldots,b_t).
\end{equation*}
Let $\Gamma_t^s$ move the $B$ block to the left through the $A$ block. At
each crossing,
\begin{equation*}
 s(a,b)=\bigl(b,\rho_b(a)\bigr),
\end{equation*}
so the left-moving colour $b$ is transmitted unchanged and applies the
permutation $\rho_b$ to the colour that it crosses. The complete $B$ block
therefore acts on $A$ only through the ordered product
\begin{equation}
 \Gamma_t^s(A,B)=\bigl(B,h(B)A\bigr),
 \qquad
 h(B)=\rho_{b_t}\rho_{b_{t-1}}\cdots\rho_{b_1}\in K_\rho.
 \label{eq:yb-rack-block}
\end{equation}
Equation~\eqref{eq:yb-rack-block} is the permutation specialization of
the right-controlled block action
Eq.~\eqref{eq:controlled-swap-right-block-action}.
In the expression above, $h(B)$ acts componentwise on the word $A$. Thus the
full scattering history of $B$ is compressed into the single finite
permutation $h(B)$.
For an involutive map, $K_\rho=\{\id_X\}$, and hence
$h(B)=\id_X$ for every word $B$.  Equation~\eqref{eq:yb-rack-block} then
reduces to $\Gamma_t^s(A,B)=(B,A)$: the scattering is the ordinary block swap
and carries no dynamical memory.  Accordingly, the factor $\kappa_\rho^2$
below is absent.

This classical rule has a direct operator form. Let
$\mathcal H_t=(\mathbb C^X)^{\otimes t}$. For $u\in K_\rho$, define the
componentwise permutation $U_u$ and the projector onto words with total
permutation $u$ by
\begin{equation}
 U_u\ket{A}=\ket{uA},
 \qquad
 \Pi_u=\sum_{B:\,h(B)=u}\ket{B}\!\bra{B}.
 \label{eq:yb-rack-projectors}
\end{equation}
Here
\begin{equation*}
 uA=\bigl(u(a_1),\ldots,u(a_t)\bigr).
\end{equation*}
Let $\Swap_{\rm blk}\ket{A,B}=\ket{B,A}$ denote the exchange of the two
complete blocks. After restoring the output order to $A\mid B$, the block
collision becomes
\begin{equation}
 \widehat\Gamma_t^s
 :=\Swap_{\rm blk}\Gamma_t^s
 =\sum_{u\in K_\rho}U_u\otimes\Pi_u.
 \label{eq:yb-rack-controlled}
\end{equation}
This is a controlled permutation: the $B$ block selects, through the value
of $h(B)$, which permutation $U_u$ acts on the complete $A$ block.
{
The block exchange $\Swap_{\rm blk}$ identifies the output order
$B_{\rm out}\mid A_{\rm out}$ with $A\mid B$.  In the conjugation calculation
of Eq.~\eqref{eq:yb-rack-product-decomposition} below, it appears on both
sides of the source operator; this conjugation preserves the source's
operator Schmidt rank.
}

At this point the present argument departs from the direct calculation
of Subsubsection~\ref{subsec:controlled-swap-right}.  There the controlled
unitary acts on a one-site operator in the target block, and orthogonality of
the control projectors removes the independent ket and bra control labels.
Here the guitar transformation dresses the physical source.  The resulting
operator $M$, introduced explicitly later in
Eq.~\eqref{eq:yb-middle-rank}, need not be local or factorized across the two
blocks.  We must therefore apply the controlled-unitary calculation to an
operator of arbitrary Schmidt rank.

Suppose now that an operator on the output blocks has a Schmidt
decomposition across $B_{\rm out}\mid A_{\rm out}$,
\begin{equation}
 M=\sum_{q=1}^{\chi}L_q\otimes R_q,
 \qquad
 \chi=\OSR(M),
 \label{eq:yb-source-Schmidt}
\end{equation}
where $L_q$ acts on $B_{\rm out}$ and $R_q$ acts on $A_{\rm out}$. Using
\eqref{eq:yb-rack-controlled}, we obtain the literal product decomposition
\begin{align}
 (\Gamma_t^s)^{-1}M\Gamma_t^s
 &= (\widehat\Gamma_t^s)^{-1}
    (\Swap_{\rm blk}M\Swap_{\rm blk})\widehat\Gamma_t^s
    \nonumber\\
 &=\sum_{q=1}^{\chi}\sum_{u,v\in K_\rho}
 \bigl(U_u^{-1}R_qU_v\bigr)
 \otimes
 \bigl(\Pi_uL_q\Pi_v\bigr).
 \label{eq:yb-rack-product-decomposition}
\end{align}
If $M$ were a one-site operator on the target block tensored with the
identity on the control block, then
$\Pi_u\Pi_v=\delta_{u,v}\Pi_u$ would reduce this expression to the
one-index conjugation formula of
Eq.~\eqref{eq:controlled-swap-right-evolved-source}.  For the dressed source,
however, the off-diagonal sectors $\Pi_uL_q\Pi_v$ with $u\ne v$ can survive.

Thus we obtain
\begin{equation}
 \OSR\!\left((\Gamma_t^s)^{-1}M\Gamma_t^s\right)
 \leq \chi |K_\rho|^2
 =\chi\kappa_\rho^2.
 \label{eq:yb-rack-conjugation-bound}
\end{equation}
The square has a simple meaning. After vectorization, the ket and bra
histories may carry independent permutations $u$ and $v$. Equivalently, one
factor $|K_\rho|$ comes from the controlled permutation and the other from
its inverse.

We can now state the uniform bound for the original permutation circuit.

\begin{theorem}[Uniform OSR bound for bare non-degenerate maps]
\label{thm:yb-map}
For every finite invertible non-degenerate Yang--Baxter map, every
one-site operator initially placed at site zero, and every $t\geq1$, the
associated unphased permutation circuit satisfies
\begin{equation}
 \OSR_c(O(t))
 \leq \kappa_\lambda^4\kappa_\rho^2
 =|K_\lambda|^4|K_\rho|^2
 \leq(D!)^6.
 \label{eq:yb-map-bound}
\end{equation}
If $r$ is involutive, then $\kappa_\rho=1$, and the bound specializes to
\begin{equation}
 \OSR_c(O(t))\leq\kappa_\lambda^4\leq(D!)^4.
 \label{eq:yb-map-involutive-bound}
\end{equation}
\end{theorem}

\begin{proof}
  Applying the guitar intertwiner 
  at every
crossing of the rectangular core gives
\begin{equation}
 W_t=\cJ_{2t}^{-1}\Gamma_t^s\cJ_{2t}.
 \label{eq:yb-rectangle-guitar}
\end{equation}
Let $S_t(O)$ be the source operator defined in
\eqref{eq:yb-rectangle-operator}. Dressing it once gives
\begin{equation}
 M=\cJ_{2t}S_t(O)\cJ_{2t}^{-1},
 \qquad
 \OSR(M)\leq\kappa_\lambda^2,
 \label{eq:yb-middle-rank}
\end{equation}
because $\OSR(S_t(O))=1$ and operator Schmidt rank is submultiplicative under
products. Substituting \eqref{eq:yb-rectangle-guitar} into
$\Phi_t(O)=W_t^\dagger S_t(O)W_t$ gives the complete
``dress, scatter, undress'' form
\begin{equation}
 \Phi_t(O)
 =\cJ_{2t}^{-1}
  (\Gamma_t^s)^{-1}M\Gamma_t^s
  \cJ_{2t}.
 \label{eq:yb-dress-scatter-undress}
\end{equation}
Equations~\eqref{eq:yb-J-rank},
\eqref{eq:yb-rack-conjugation-bound}, and
\eqref{eq:yb-middle-rank} now give
\begin{equation}
 \OSR(\Phi_t(O))
 \leq
 \underbrace{\kappa_\lambda^2}_{\text{dress the source}}
 \underbrace{\kappa_\rho^2}_{\text{controlled permutations}}
 \underbrace{\kappa_\lambda^2}_{\text{return to bare colours}}
 =\kappa_\lambda^4\kappa_\rho^2.
 \label{eq:yb-rank-budget}
\end{equation}
Lemma~\ref{lem:yb-rectangle} identifies this rank with the physical
fixed-cut rank $\OSR_c(O(t))$. Finally,
$K_\lambda,K_\rho\leq\operatorname{Sym}(X)$, and hence
$\kappa_\lambda,\kappa_\rho\leq D!$. This proves
\eqref{eq:yb-map-bound}.  For involutive $r$, the identity
$\kappa_\rho=1$ gives \eqref{eq:yb-map-involutive-bound}.
\end{proof}

The absence of phases is essential to the present finite-memory argument.
In Subsection~\ref{subsec:phase-general} we add arbitrary unit-modulus
phase dressings.
The endpoint group element is then no longer sufficient, but the additional
history still compresses to finitely many occupation numbers, leading to a
polynomial rather than a uniform bound.

\subsection{Arbitrary phase dressings of non-degenerate Yang--Baxter maps}
\label{subsec:phase-general}

We now add arbitrary unit-modulus phases to the non-degenerate permutation gates of the
previous subsection.

Without phases, we proved in Subsection~\ref{subsec:permutation-osr} that
there are only finitely many Schmidt components, with a bound independent of
time.  After the phases are added, each of these components can split into
many independent contributions.  The special properties of the guitar map
nevertheless allow us to absorb the full phase dependence into the operators
on one side of the cut, leaving only a polynomial number of product terms
across the cut.  This organization does not restrict the possible accumulated
phases and is independent of their concrete values: the only Yang--Baxter
input is the braid relation for the underlying classical map $r$.

In guitar coordinates the permutation map still acts as
$s(a,b)=(b,\rho_b(a))$, so the complete $B$ block rotates the $A$ block by
one element of the finite group $K_\rho$.  However, the similarity transformation
by the guitar map does not, in general, turn the phased gate into a strictly
local two-site gate: although the permutation support remains local, the phase
at a crossing depends on the inverse-scan memory carried by the prefix to its
left.
This nonlocality is mild.  For fixed $A$, we will show that the total
phase depends only on the occupation numbers of finitely many state-labelled
``row types''.  Their number is fixed by the gate and independent of time,
which is the origin of the polynomial bound below.

To determine the phase associated with a single crossing, let $z,w$ be
the two incoming colours in guitar coordinates, and let $g\in K_\lambda$ be
the inverse-scan memory immediately to their left.  If $x,y$ are the
corresponding original colours, the forward guitar scan gives
\begin{equation*}
 z=g(x),
 \qquad
 w=g\lambda_x(y).
\end{equation*}
Inverting these relations, we obtain
\begin{equation}
 x=g^{-1}(z),
 \qquad
 y=\lambda_{g^{-1}(z)}^{-1}\!\bigl(g^{-1}(w)\bigr).
 \label{eq:yb-recover-pair}
\end{equation}
For later use, define the phase recovered from the guitar coordinates by
\begin{equation}
 \Omega(g;z,w)
 :=
 \omega\!\left(
 g^{-1}(z),
 \lambda_{g^{-1}(z)}^{-1}\!\bigl(g^{-1}(w)\bigr)
 \right).
 \label{eq:yb-local-recovered-phase}
\end{equation}
Thus the phase at a crossing is determined by the finite data $(g,z,w)$.
The complete prefix to the left of the crossing enters only through $g$.

\begin{lemma}[One-crossing guitar conjugation]
\label{lem:yb-phased-local-conjugation}
Let $N\geq2$, let $R_i$ act on sites $i,i+1$, and let
$\mathbf z=(z_1,\ldots,z_N)$ be a word in guitar coordinates.  Denote by
$g\in K_\lambda$ the inverse-scan memory immediately to the left of these
sites.  Then
\begin{equation}
 \cJ_NR_i\cJ_N^{-1}\ket{\mathbf z}
 =
 \Omega(g;z_i,z_{i+1})\ket{s_i(\mathbf z)},
 \label{eq:yb-phased-local-conjugation}
\end{equation}
where $s_i$ applies the derived map $s$ to sites $i,i+1$.  Moreover, the
inverse-scan memory immediately to the right of the two sites is the same
before and after the crossing.
\end{lemma}

\begin{proof}
Let $\mathbf x=J_N^{-1}(\mathbf z)$ and write $x=x_i$, $y=x_{i+1}$.
The inverse scan in Eq.~\eqref{eq:yb-guitar-inverse} gives precisely
Eq.~\eqref{eq:yb-recover-pair}, so the phase produced by $R_i$ is
$\Omega(g;z_i,z_{i+1})$.  For the support, write $r(x,y)=(u,v)$.  The
incoming pair leaves the scan memory
$g\lambda_x\lambda_y$, while the outgoing pair leaves
$g\lambda_u\lambda_v$.  These two elements agree by
Eq.~\eqref{eq:yb-lambda-product}.  Hence the guitar coordinates of the
suffix are unchanged, and the general guitar intertwiner
Eq.~\eqref{eq:yb-non-involutive-guitar-intertwiner} changes only the selected
pair, from $(z_i,z_{i+1})$ to $s(z_i,z_{i+1})$.  Multiplying this support
action by the recovered phase proves
Eq.~\eqref{eq:yb-phased-local-conjugation} and the memory statement.
\end{proof}

Lemma~\ref{lem:yb-phased-local-conjugation} allows us to propagate a single
memory in $K_\lambda$ through the complete rectangle.  Although the phase of
a conjugated crossing depends on the prefix, no information beyond this
finite memory is required.

We now sketch the full procedure for the complete rectangle.  Let
$A=(a_1,\ldots,a_t)$ and $B=(b_1,\ldots,b_t)$ be the incoming words
in guitar coordinates.  We evaluate the rectangle in $t$ rows.  With every
row we associate a \emph{row type}.  It depends only on the variables in
$B$, and it is nonlocal in the sense that the type of row $j$ depends on the
prefix $(b_1,\ldots,b_j)$.  This nonlocality is nevertheless carried entirely
by finite memories inherited from the guitar map.

We now define the row types.  In row $j$, the colour $b_j$ passes from
right to left through the current $A$ block.  At the beginning of this row,
let $k_j\in K_\lambda$ be the inverse-scan memory carried by the emitted
prefix $(b_1,\ldots,b_{j-1})$, and let
$h_j\in K_\rho$ be the rotation already applied to $A$.  Initially
$k_1=h_1=1$.  Completing row $j$ gives
\begin{equation}
 k_{j+1}=k_j\lambda_{k_j^{-1}(b_j)},
 \qquad
 h_{j+1}=\rho_{b_j}h_j.
 \label{eq:yb-type-edge}
\end{equation}
The values $(k_j,h_j,b_j)$ contain the finite memory needed for this
row.  We define its \emph{row type} by
$\alpha_j=(k_j,h_j,b_j)\in K_\lambda\times K_\rho\times X$.  There are at
most $m=D|K_\lambda||K_\rho|$ such types, which we label by
$\alpha=1,\ldots,m$.

For fixed $A$ and a row type $\alpha=(k,h,b)$, let $F_\alpha(A)$ be the
total scalar phase accumulated as $b$ crosses
$h(a_t),h(a_{t-1}),\ldots,h(a_1)$, starting with inverse-scan memory $k$.

For a word $B$, let $n_\alpha(B)$ be the number of rows of type $\alpha$.
Multiplying the row phases gives
\begin{equation}
 Q_t(A,B)
 =
 \prod_{\alpha=1}^{m}F_\alpha(A)^{\,n_\alpha(B)}.
 \label{eq:yb-phase-factorization}
\end{equation}
This is the decisive compression step.  The state labels $k$ and $h$ have
already absorbed the noncommutative effect of the preceding order.  Once the
rows are refined to these state-labelled types, their factors are scalars and
commute.  The phase factor $Q_t(A,B)$ may still assume exponentially many
different values as $A$ varies.  The compression concerns only its dependence
on the word $B$: for fixed $A$, the total phase depends on $B$ only through
the occupation numbers $n_\alpha(B)$, not the order of the rows.

The occupation vector $n(B)=(n_1(B),\ldots,n_m(B))$ has nonnegative integer
components with total $t$.  Hence the number of possible values is at most
\begin{equation}
 H_t
 =
 \binom{t+m-1}{m-1}
 =
 \binom{t+D|K_\lambda||K_\rho|-1}
 {D|K_\lambda||K_\rho|-1}.
 \label{eq:yb-summary-count}
\end{equation}
Not every vector counted here must be realized by a word $B$.  For fixed
local dimension and a fixed gate, $H_t$ is a polynomial in $t$.

The occupation vector contains the complete phase memory, but we do not
use its row-type counts to determine the final rotation of $A$.  Instead, we retain the
endpoint $h_{t+1}(B)\in K_\rho$ as a separate finite label and define
\begin{equation}
 \sigma(B):=\bigl(n(B),h_{t+1}(B)\bigr),
 \qquad
 \mathcal S_t:=\{\sigma(B):B\in X^t\}.
 \label{eq:yb-phase-summary}
\end{equation}
Since $h_{t+1}(B)$ has at most $|K_\rho|$ possible values,
\begin{equation}
 |\mathcal S_t|\leq |K_\rho|H_t.
 \label{eq:yb-phase-summary-count}
\end{equation}
Thus the phase history contributes polynomial memory, while the
possibly noncommutative rotation history contributes only a finite endpoint.

We now turn the scan into an operator identity for the complete rectangle.

\begin{lemma}[Phased rectangular action]
\label{lem:yb-phased-rectangle-action}
Let $W_t$ be the complete phase-dressed $t\times t$ rectangular circuit and
define its expression in guitar coordinates by
\begin{equation}
 \Gamma_t:=\cJ_{2t}W_t\cJ_{2t}^{-1},
 \qquad
 W_t=\cJ_{2t}^{-1}\Gamma_t\cJ_{2t}.
 \label{eq:yb-phased-rectangle-conjugation}
\end{equation}
Then, for every pair of words $A,B\in X^t$,
\begin{equation}
 \Gamma_t\ket{A,B}
 =
 Q_t(A,B)\ket{B,h_{t+1}(B)A}.
 \label{eq:yb-phased-rectangle}
\end{equation}
\end{lemma}

\begin{proof}
Crossings on disjoint strands commute, and evaluating the rectangle one
$B$ strand at a time preserves the order of the crossings along every strand.
We may therefore proceed row by row.  At the beginning of row $j$, suppose
that the current word is
$\ket{b_1,\ldots,b_{j-1},h_jA,b_j,\ldots,b_t}$ and that the accumulated phase
is $\prod_{\ell=1}^{j-1}F_{\alpha_\ell}(A)$.  This holds for $j=1$ because
$k_1=h_1=1$.  Applying
Lemma~\ref{lem:yb-phased-local-conjugation} throughout row $j$ moves $b_j$ to
the emitted prefix, changes the current $A$ block from $h_jA$ to
$\rho_{b_j}h_jA=h_{j+1}A$, contributes $F_{\alpha_j}(A)$, and updates the
scan memory to $k_{j+1}$ according to Eq.~\eqref{eq:yb-type-edge}.  Induction
over the rows gives
$\Gamma_t\ket{A,B}=[\prod_{j=1}^tF_{\alpha_j}(A)]
\ket{B,h_{t+1}(B)A}$.  Equation~\eqref{eq:yb-phase-factorization} identifies
the product with $Q_t(A,B)$ and proves
Eq.~\eqref{eq:yb-phased-rectangle}.
\end{proof}

We now use Eq.~\eqref{eq:yb-phase-summary-count} to bound the operator
Schmidt rank.

For $\sigma=(n,h)\in\mathcal S_t$, define
\begin{align}
 \Pi_\sigma
 &=
 \sum_{B:\,\sigma(B)=\sigma}\ket{B}\!\bra{B},
 &
 V_\sigma\ket{A}
 &=
 \left[\prod_{\alpha=1}^{m}F_\alpha(A)^{n_\alpha}\right]
 \ket{hA}.
 \label{eq:yb-phase-control-data}
\end{align}
Each $V_\sigma$ is a phase-permutation and hence unitary.  Since the basis
vectors $\ket{A,B}$ span the two-block space,
Lemma~\ref{lem:yb-phased-rectangle-action} is equivalent, after restoring
the output order with the block exchange of the previous subsection, to the
controlled identity
\begin{equation}
 \widehat\Gamma_t
 :=
 \Swap_{\rm blk}\Gamma_t
 =
 \sum_{\sigma\in\mathcal S_t}V_\sigma\otimes\Pi_\sigma.
 \label{eq:yb-phase-controlled}
\end{equation}
{
Thus $\widehat\Gamma_t$ is a controlled unitary with at most
$|K_\rho|H_t$ control sectors.  In the conjugation of a source operator $M$,
the block exchange replaces $M$ by $\Swap_{\rm blk}M\Swap_{\rm blk}$, which
preserves $\OSR(M)$.
}

We can now reuse the controlled-unitary calculation of
Eq.~\eqref{eq:yb-rack-product-decomposition}.  If $M$ has operator Schmidt
rank $\chi$ across the two blocks, expanding the two controlled sums in the
conjugation by $\Gamma_t$ gives at most one product operator for each original
Schmidt term and each pair $\sigma,\sigma'\in\mathcal S_t$.  Consequently,
\begin{equation}
 \OSR\!\left(\Gamma_t^{-1}M\Gamma_t\right)
 \leq
 \chi |\mathcal S_t|^2
 \leq
 \chi |K_\rho|^2H_t^2.
 \label{eq:yb-phase-rank-lemma}
\end{equation}
The square appears because the row and column indices of an operator, or
equivalently its ket and bra histories, can carry independent summaries.
Both the phases and the support permutations are retained exactly inside the
operators $V_\sigma$.

We can now state the resulting bound for the physical Heisenberg operator.

\begin{theorem}[Phase dressings of non-degenerate Yang--Baxter maps]
\label{thm:yb-phase}
Let $R$ be a phase-dressed Yang--Baxter-map gate of the form
\eqref{eq:setup-phase-dressed-map}, where the support $r$
is finite and non-degenerate and the unit-modulus phases $\omega(x,y)$
are arbitrary.  Then every one-site operator initially placed
at site zero satisfies, for every $t\geq1$,
\begin{equation}
 \OSR_c(O(t))
 \leq
 |K_\lambda|^4|K_\rho|^2
 \binom{t+D|K_\lambda||K_\rho|-1}
 {D|K_\lambda||K_\rho|-1}^{\!2}.
 \label{eq:yb-phase-bound}
\end{equation}
For fixed local dimension and a fixed gate this is polynomial in $t$.
Consequently, every nonnegative-R\'enyi operator entropy is $O(\log t)$.
\end{theorem}

\begin{proof}
The proof has the same outer structure as Theorem~\ref{thm:yb-map}.  Let
$S_t(O)$ be the one-site source in the rectangular network and put
$M=\cJ_{2t}S_t(O)\cJ_{2t}^{-1}$.  Equation~\eqref{eq:yb-J-rank} gives
$\OSR(M)\leq|K_\lambda|^2$.  The global conjugation identity
Eq.~\eqref{eq:yb-phased-rectangle-conjugation}, proved crossing by crossing
in Lemma~\ref{lem:yb-phased-rectangle-action}, identifies the physical
rectangle with $\cJ_{2t}^{-1}\Gamma_t\cJ_{2t}$.  The central conjugation
$\Gamma_t^{-1}M\Gamma_t$ therefore multiplies the rank by at most
$|K_\rho|^2H_t^2$, according to Eq.~\eqref{eq:yb-phase-rank-lemma}, while
the two outer guitar transformations contribute another factor
$|K_\lambda|^2$.  Submultiplicativity gives
$\OSR(\Phi_t(O))\leq|K_\lambda|^4|K_\rho|^2H_t^2$.  Substituting
Eq.~\eqref{eq:yb-summary-count} and using the rectangular reduction of
Lemma~\ref{lem:yb-rectangle} proves Eq.~\eqref{eq:yb-phase-bound}.
\end{proof}

We stress that this is an upper bound; it does not determine the actual
power of $t$.  Without phases, the complete $B$ block is summarized by the
single endpoint $h(B)\in K_\rho$, and the OSR remains uniformly bounded.

\subsection{\texorpdfstring{{Centralizer bound for involutive supports}}{Centralizer bound for involutive supports}}
\label{subsec:phase-involutive-support}

We now consider the smaller class in which the support $r$ is involutive
and the full phase-dressed gate $R$ satisfies the braid relation.  The full
gate need not be involutive.  We combine the centralizer restriction of
Section~\ref{sec:centralizer-output-space} with the guitar map introduced in
Subsection~\ref{subsec:guitar}.

The computation has two steps.  First, we show that the phases can only
reduce the contribution of each support orbit to the centralizer.  Second,
we use the guitar map $J_k$ to identify these support orbits with ordinary
simultaneous-permutation orbits.  Here the guitar map is used only to count
the centralizer dimension; it is not applied to the physical rectangular
circuit.

\begin{lemma}[Centralizer bound from an involutive support]
\label{lem:yb-phased-involutive-centralizer}
Let $R$ be a phase-dressed Yang--Baxter-map gate of the form
\eqref{eq:setup-phase-dressed-map}, where the support $r$ is finite,
non-degenerate, and involutive.  The phases are arbitrary and the full gate
is not assumed to satisfy the braid relation.  Then the centralizer defined
in Eq.~\eqref{eq:centralizer-definition} obeys
\begin{equation}
 c_k(R)\leq
 \binom{k+D^2-1}{D^2-1}.
 \label{eq:yb-phased-involutive-centralizer}
\end{equation}
\end{lemma}

\begin{proof}
For words $x,y\in X^k$, write $M_{x,y}=\ket{x}\!\bra{y}$.  Let $r_i$ act
as the classical map $r$ on entries $i,i+1$, and let $R_i$ be the
corresponding gate on these sites.  Then
\begin{equation}
 R_iM_{x,y}R_i^\dagger
 =\omega_i(x)\overline{\omega_i(y)}M_{r_i x,r_i y},
 \qquad
 \omega_i(x)=\omega(x_i,x_{i+1}).
 \label{eq:yb-phased-adjoint-monomial}
\end{equation}
Thus conjugation by the generators is monomial on the matrix-unit basis.
The orbit argument is the same as the one following
Eq.~\eqref{eq:inv-adjoint-monomial} in the proof of
Theorem~\ref{thm:inv-polynomial-rank}.  On each orbit obtained by applying
the same support generator $r_i$ to both words, one coefficient determines
all the others.  The phases can make the conditions around a closed path
inconsistent and thereby remove the invariant operator, but they cannot
produce more than one invariant operator on the orbit.

We now use the guitar map to count the support orbits.  Recall from
Subsection~\ref{subsec:guitar} that $J_k$ is a bijection and satisfies
Eq.~\eqref{eq:yb-involutive-guitar-intertwiner}.  Applying $J_k$ separately
to the two words $x$ and $y$ turns the simultaneous action of the support
generators into the simultaneous exchange of adjacent positions.  These are
the orbits counted earlier by the $D^2$ joint occupation numbers in
Eq.~\eqref{eq:inv-joint-occupations}, now for the pair
$(J_k(x),J_k(y))$.  Their number is
$\binom{k+D^2-1}{D^2-1}$, which proves
Eq.~\eqref{eq:yb-phased-involutive-centralizer}.
\end{proof}

The lemma bounds an abstract centralizer for every choice of phases.  To
turn this dimension count into a restriction on the evolved operator, we
must now require the full gate to satisfy the braid relation.

\begin{theorem}[Braided phase dressings with involutive support]
\label{thm:yb-phased-involutive-support}
Let $R$ satisfy the assumptions of
Lemma~\ref{lem:yb-phased-involutive-centralizer}, and suppose in addition
that the full phase-dressed gate $R$ satisfies the braid relation.  Then
every one-site operator initially placed at site zero obeys
\begin{equation}
 \OSR_c(O(t))\leq
 \binom{t+D^2-1}{D^2-1}\leq(t+1)^{D^2-1}.
 \label{eq:yb-phased-involutive-bound}
\end{equation}
Consequently, every nonnegative-R\'enyi operator entropy is at most
$(D^2-1)\log(t+1)$.
\end{theorem}

\begin{proof}
The braid relation for the full gate allows us to apply
Corollary~\ref{cor:centralizer-osr-bound}, which gives
$\OSR_c(O(t))\leq c_t(R)$.  Lemma~\ref{lem:yb-phased-involutive-centralizer}
then proves Eq.~\eqref{eq:yb-phased-involutive-bound}.  The entropy bound
follows from $S_\alpha^{\rm op}\leq\log\OSR$.
\end{proof}

The estimate in Eq.~\eqref{eq:yb-phased-involutive-bound} is numerically the
same as the bound of Theorem~\ref{thm:inv-polynomial-rank}, but the two
statements have different assumptions.  Here involutivity is required only
for the support, whereas the full gate must still satisfy the braid
relation.  On this intersection, both Theorem~\ref{thm:yb-phase} and the
present theorem apply, and the smaller of their two upper bounds may be
used.  We expect the present estimate to be sharper for many nontrivial
involutive supports.  Indeed, in this case $\kappa_\rho=1$, but the general
bound has polynomial degree $2D\kappa_\lambda-2$, while $\kappa_\lambda$ is
only bounded by $D!$.  The degree $D^2-1$ of the present estimate is
independent of this potentially large permutation group.

\section{Exponential growth of the OSR from a Yang--Baxter gate without dual
unitarity}
\label{sec:braid-only}

In this section we show that the braid relation and involutivity alone do not
prevent exponential operator-Schmidt-rank growth.  We construct
a fixed seven-state involutive Yang--Baxter gate without dual unitarity and a
fixed one-site operator for which
\(\OSR_c(O(t))\geq 2^{t-1}\).

The key idea is to choose the gate, the local
operator, and a special product vector so that the rectangular representation
of \(O(t)\) from Lemma~\ref{lem:yb-rectangle} maps this vector to a state
containing \(t-1\) Bell pairs across the output bipartition.  We construct the
gate in Subsection~\ref{subsec:bo-construction}; in
Subsection~\ref{subsec:bo-growth} we choose the local operator and product
vector and derive the exponential lower bound.  The Schmidt rank produced from
a single product vector lower-bounds the operator-Schmidt rank, so this action
already proves exponential growth.

We do not analyze the normalized
operator-Schmidt weights, and
therefore we do not claim a growth law for the positive-index R\'enyi or von
Neumann operator entropies.

\subsection{The construction of the gate}
\label{subsec:bo-construction}

Split the seven-dimensional local space into two orthogonal colour sectors of
dimensions three and four,
\begin{equation}
 V=H_3\oplus H_4,\qquad
 H_3=\operatorname{span}\{a,m,d\},\qquad
 H_4=\operatorname{span}\{b,h_0,h_1,p\},
 \label{eq:bo-colour-split}
\end{equation}
where the displayed vectors are orthonormal.  The subscripts on \(H_3\) and
\(H_4\) record the dimensions of the two sectors.

We first keep the cross-sector scattering general.  For any unitary map
\(C:H_3\otimes H_4\to H_4\otimes H_3\), we define a two-site gate by
\begin{equation}
 \begin{aligned}
 R_{\rm mark}|_{H_3\otimes H_3}&=\id,\qquad&
 R_{\rm mark}|_{H_4\otimes H_4}&=\id, \\
 R_{\rm mark}|_{H_3\otimes H_4}&=C,&
 R_{\rm mark}|_{H_4\otimes H_3}&=C^\dagger.
 \end{aligned}
 \label{eq:bo-R}
\end{equation}
Thus inputs from the same sector pass unchanged, whereas inputs from different
sectors exchange places and scatter through \(C\) or \(C^\dagger\).  The
following gate properties hold for every such \(C\).  After
Lemma~\ref{lem:bo-gate}, we will specify the matrix elements of \(C\) for the
exponential-rank construction.

\begin{lemma}[Properties of the general construction]
\label{lem:bo-gate}
For every unitary \(C:H_3\otimes H_4\to H_4\otimes H_3\), the gate
\eqref{eq:bo-R} is unitary, Hermitian, involutive, and satisfies the braid
relation.  The resulting gate \(R_{\rm mark}\) is not dual-unitary,
irrespective of the choice of \(C\).
\end{lemma}

\begin{proof}
The gate is the identity on \(H_3\otimes H_3\) and
\(H_4\otimes H_4\), while its two cross-sector restrictions are paired as
\(C\) and \(C^\dagger\).  Hence
\(R_{\rm mark}=R_{\rm mark}^\dagger\) and
\(R_{\rm mark}^2=\id\).

The braid relation can be checked on the eight tensor sectors of
\(V^{\otimes3}\).  On \(H_3^{\otimes3}\) and \(H_4^{\otimes3}\), both braid
words are trivial.  On a tensor sector with one exceptional factor, that
strand crosses the two equal-sector strands in the same order on both sides
of the braid relation.  On \(H_3\otimes H_4\otimes H_3\) and
\(H_4\otimes H_3\otimes H_4\), a cross-sector scattering is immediately
undone by its inverse.  These cases exhaust \(V^{\otimes3}\), so the braid
relation holds for any unitary \(C\).

Dual unitarity fails already at one row of the reshuffled matrix.  For every
\(y,v\in V\),
\begin{equation}
 (R_{\rm mark}^\Gamma)_{(a,m),(y,v)}
 =\bra{m,v}R_{\rm mark}\ket{a,y}=0.
 \label{eq:bo-zero-row}
\end{equation}
If \(y\in H_3\), the equal-sector identity leaves the first output equal to
\(a\); if \(y\in H_4\), the first output lies in \(H_4\).  Neither case can
produce \(m\in H_3\), so the reshuffling is singular.
\end{proof}

We now choose the cross-sector scattering.  Put
\(e_\pm=(h_0\otimes a\pm h_1\otimes d)/\sqrt2\), and define \(C\) on the
product basis by
\[
\renewcommand{\arraystretch}{1.25}
\begin{array}{c@{\ \mapsto\ }c@{\qquad}c@{\ \mapsto\ }c}
 a\otimes b   & b\otimes a   & d\otimes b   & b\otimes d \\
 m\otimes b   & e_+          & a\otimes h_0 & b\otimes m \\
 a\otimes h_1 & p\otimes m   & a\otimes p   & p\otimes a \\
 \hline
 m\otimes h_0 & h_0\otimes m & m\otimes h_1 & h_0\otimes d \\
 m\otimes p   & h_1\otimes a & d\otimes h_0 & h_1\otimes m \\
 d\otimes h_1 & p\otimes d   & d\otimes p   & e_-
\end{array}
\]
The first three lines contain all transitions used in the rank proof that we
will present below; the last three complete the map to a unitary.  Indeed, the
ten product outputs
are distinct basis vectors outside
\(\operatorname{span}\{h_0\otimes a,h_1\otimes d\}\), while \(e_+\) and
\(e_-\) are an orthonormal basis of that remaining two-dimensional space.

\subsection{The growth of the OSR}
\label{subsec:bo-growth}

We now prove that the gate constructed in
Subsection~\ref{subsec:bo-construction} has exponential OSR growth for a
fixed one-site source.  The precise statement is as follows.

\begin{theorem}[Exponential growth of the OSR]
\label{thm:bo-osr-growth}
For the seven-state gate \(R_{\rm mark}\) constructed in
Subsection~\ref{subsec:bo-construction}, let
\(O=\ket{h_0}\bra b\) be initially placed at site zero.  Then
\begin{equation}
 \OSR_c(O(t))\geq 2^{t-1},\qquad t\geq1.
 \label{eq:bo-osr-theorem}
\end{equation}
\end{theorem}

\begin{proof}
By Lemma~\ref{lem:yb-rectangle}, the OSR across the spatial cut equals the
OSR of the rectangular core \(\Phi_t(O)\).  We lower-bound this rank by
evaluating \(\Phi_t(O)\) on a suitable product vector.  Indeed, applying
\(\Phi_t(O)\) to a product vector cannot increase its operator Schmidt rank:
\begin{equation}
 \SR\bigl(\Phi_t(O)(\ket{v_L}\otimes\ket{v_R})\bigr)
 \leq\OSR(\Phi_t(O)).
 \label{eq:setup-product-vector}
\end{equation}
We choose the two factors to be
\begin{equation}
 \ket{v_L}=\ket b^{\otimes t},\qquad
 \ket{v_R}=\ket a^{\otimes t}.
 \label{eq:bo-product-vector}
\end{equation}
Equation~\eqref{eq:yb-rectangle-operator} gives
\(\Phi_t(O)=W_t^\dagger S_t(O)W_t\), so the rightmost \(W_t\) acts first.
Every crossing in its action on \(\ket{v_L}\otimes\ket{v_R}\) sees
\(b\otimes a\).  Since
\(C^\dagger(b\otimes a)=a\otimes b\), each crossing exchanges the two
local states without changing them.  Consequently,
\begin{equation}
 W_t(\ket{v_L}\otimes\ket{v_R})
 =\ket a^{\otimes t}\otimes\ket b^{\otimes t}.
 \label{eq:bo-first-block-exchange}
\end{equation}

The source \(S_t(O)\) contains \(O=\ket{h_0}\bra b\) on the first
\(B\)-strand.  It changes that \(b\) into \(h_0\) and leaves every other
strand unchanged.  Moreover, involutivity and unitarity give
\(W_t^\dagger=W_t\).  We therefore obtain
\begin{equation}
 \ket{\Psi_t}:=\Phi_t(O)(\ket{v_L}\otimes\ket{v_R})
 =W_t\bigl(\ket a^{\otimes t}\otimes\ket{h_0}
 \otimes\ket b^{\otimes(t-1)}\bigr).
 \label{eq:bo-Psi-definition}
\end{equation}

Figure~\ref{fig:bo-marker-process} shows for \(t=4\) the process by which
the Bell pairs needed for the lower bound are created.
We now explain the process displayed there in detail and derive its
row-by-row transitions in
Eqs.~\eqref{eq:bo-t4-after-B0}--\eqref{eq:bo-t4-after-B2}.

For \(t=4\), Eq.~\eqref{eq:bo-Psi-definition} instructs us to apply \(W_4\)
to
\[
 \ket a_{A_0}\ket a_{A_1}\ket a_{A_2}\ket a_{A_3}
 \ket{h_0}_{B_0}\ket b_{B_1}\ket b_{B_2}\ket b_{B_3}.
\]
In this notation, we write the strand labels \(A_j\) and \(B_k\)
explicitly as subscripts, while the tensor factors appear in their current
left-to-right order.

We begin with row \(B_0\).  The state \(h_0\) first crosses \(A_3\), and we
apply
\(C(a\otimes h_0)=b\otimes m\).  The resulting \(b\) on \(B_0\) then
crosses \(A_2,A_1,A_0\) by repeated use of
\(C(a\otimes b)=b\otimes a\).  After row \(B_0\), the state is
\begin{equation}
 \ket b_{B_0}\ket a_{A_0}\ket a_{A_1}\ket a_{A_2}
 \ket m_{A_3}\ket b_{B_1}\ket b_{B_2}\ket b_{B_3}.
 \label{eq:bo-t4-after-B0}
\end{equation}

We introduce the following terminology: we call \(m\) on \(A_3\) the marker,
the remaining \(a\)-states on \(A_0,A_1,A_2\) vacua, and the unprocessed
\(b\)-states on \(B_1,B_2,B_3\) blanks.

When executing the crossings associated with \(B_1\), the blank first meets
the marker on \(A_3\), followed by the vacuum on \(A_2\).  Reading the two
gates from right to left gives
\begin{equation}
 \begin{aligned}
  \ket a_{A_2}\ket m_{A_3}\ket b_{B_1}
  &\xrightarrow{(R_{\rm mark})_{23}}
  \frac{\ket a_{A_2}\ket{h_0}_{B_1}\ket a_{A_3}
       +\ket a_{A_2}\ket{h_1}_{B_1}\ket d_{A_3}}{\sqrt2}\\
  &\xrightarrow{(R_{\rm mark})_{12}}
  \frac{\ket b_{B_1}\ket m_{A_2}\ket a_{A_3}
       +\ket p_{B_1}\ket m_{A_2}\ket d_{A_3}}{\sqrt2}.
 \end{aligned}
 \label{eq:bo-first-bell-pair}
\end{equation}
The first crossing creates the two-component superposition, while the second
crossing recreates \(m\) and separates the two components onto \(B_1\) and
\(A_3\).  We now denote their two-strand state by
\begin{equation}
 \ket\beta_{B,A}
 :=\frac{\ket b_B\ket a_A+\ket p_B\ket d_A}{\sqrt2}.
 \label{eq:bo-beta}
\end{equation}
This is a Bell pair between
\(\operatorname{span}\{b,p\}\subset H_4\) and
\(\operatorname{span}\{a,d\}\subset H_3\), and hence carries one unit of
entanglement between an \(A\)-strand and a \(B\)-strand.  In
Eq.~\eqref{eq:bo-first-bell-pair}, \(m\) remains at tensor-product position
\(2\), although it is transferred from \(A_3\) to \(A_2\).

The \(B_1\)-leg of this Bell pair must still cross the vacua on \(A_1\) and
\(A_0\).  Its two components are \(b\) and \(p\), and
\begin{equation}
 C(a\otimes b)=b\otimes a,\qquad
 C(a\otimes p)=p\otimes a.
 \label{eq:bo-bell-B-transport}
\end{equation}
These crossings leave the Bell pair unchanged.  Grouping its two entangled
legs, the strand contents after row \(B_1\) are
\begin{equation}
 \ket b_{B_0}\ket a_{A_0}\ket a_{A_1}\ket m_{A_2}
 \ket\beta_{B_1,A_3}\ket b_{B_2}\ket b_{B_3}.
 \label{eq:bo-t4-after-B1}
\end{equation}

We now apply row \(B_2\), starting from a configuration with exactly one Bell
pair.  The blank on \(B_2\) first crosses its \(A_3\)-leg.  Expanding the
Bell pair makes this crossing explicit:
\begin{equation}
 \begin{aligned}
 &\frac{\ket b_{B_1}\ket a_{A_3}
             +\ket p_{B_1}\ket d_{A_3}}{\sqrt2}\ket b_{B_2}\\
 &\hspace{12mm}\longmapsto
 \frac{\ket b_{B_1}\ket b_{B_2}\ket a_{A_3}
             +\ket p_{B_1}\ket b_{B_2}\ket d_{A_3}}{\sqrt2}.
 \end{aligned}
 \label{eq:bo-preserve-first-bell-pair}
\end{equation}
Here we used \(C(a\otimes b)=b\otimes a\) and
\(C(d\otimes b)=b\otimes d\).  The correlations between \(B_1\) and
\(A_3\) are unchanged, so the first Bell pair is preserved while the blank
passes its \(A\)-leg.

The blank now reaches the marker on \(A_2\), with the vacuum on \(A_1\)
immediately before it.  Repeating the two crossings of
Eq.~\eqref{eq:bo-first-bell-pair} with relabelled strands gives
\begin{equation}
 \ket a_{A_1}\ket m_{A_2}\ket b_{B_2}
 \longmapsto
 \frac{\ket b_{B_2}\ket m_{A_1}\ket a_{A_2}
       +\ket p_{B_2}\ket m_{A_1}\ket d_{A_2}}{\sqrt2}.
 \label{eq:bo-second-bell-pair}
\end{equation}
This creates the second Bell pair, now between \(B_2\) and \(A_2\), and
transfers the marker to \(A_1\).  Its \(B_2\)-leg crosses the remaining
vacuum on \(A_0\) by Eq.~\eqref{eq:bo-bell-B-transport}.  Grouping the
entangled legs, the strand contents after row \(B_2\) are
\begin{equation}
 \ket b_{B_0}\ket a_{A_0}\ket m_{A_1}
 \ket\beta_{B_1,A_3}\ket\beta_{B_2,A_2}\ket b_{B_3}.
 \label{eq:bo-t4-after-B2}
\end{equation}

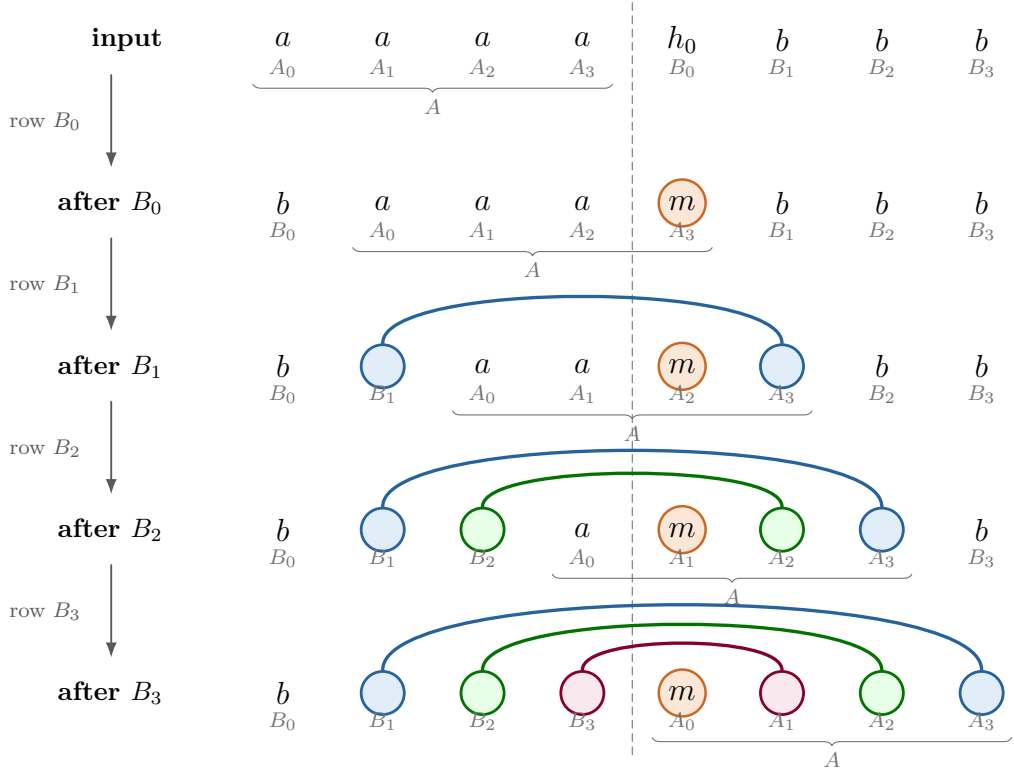
\begin{figure}[!t]
\centering
\begin{tikzpicture}[
  opent figure,
  x=1.32cm,
  y=0.72cm,
  state/.style={font=\large,minimum size=6mm,inner sep=1pt},
  marker state/.style={
    circle,draw=opentorange,fill=opentlightorange,line width=0.95pt,
    minimum size=6.2mm,inner sep=0pt,font=\large
  },
  bell one leg/.style={
    circle,draw=opentblue,fill=opentlightblue,line width=1.05pt,
    minimum size=5.7mm,inner sep=0pt
  },
  bell two leg/.style={
    circle,draw=green!45!black,fill=green!10,line width=1.05pt,
    minimum size=5.7mm,inner sep=0pt
  },
  bell three leg/.style={
    circle,draw=purple!70!black,fill=purple!9,line width=1.05pt,
    minimum size=5.7mm,inner sep=0pt
  },
  bell one/.style={draw=opentblue,line width=1.25pt},
  bell two/.style={draw=green!45!black,line width=1.25pt},
  bell three/.style={draw=purple!70!black,line width=1.25pt},
  row label/.style={font=\small\bfseries,anchor=east},
  strand tag/.style={font=\scriptsize,text=black!55},
  step label/.style={font=\scriptsize,anchor=east,text=black!65},
  A brace/.style={
    decorate,decoration={brace,amplitude=3pt,mirror},draw=opentgray
  }
]
  \draw[densely dashed,draw=opentgray]
    (3.5,12.68) -- (3.5,-1.22);

  \node[row label] at (-1.10,12) {input};
  \foreach \x/\lab in {0/A_0,1/A_1,2/A_2,3/A_3} {
    \node[state] at (\x,12) {$a$};
    \node[strand tag] at (\x,11.48) {$\lab$};
  }
  \node[state] at (4,12) {$h_0$};
  \node[strand tag] at (4,11.48) {$B_0$};
  \foreach \x/\lab in {5/B_1,6/B_2,7/B_3} {
    \node[state] at (\x,12) {$b$};
    \node[strand tag] at (\x,11.48) {$\lab$};
  }
  \draw[A brace] (-0.30,11.18) -- (3.30,11.18);
  \node[strand tag] at (1.5,10.77) {$A$};

  \node[row label] at (-1.10,9) {after \(B_0\)};
  \node[state] at (0,9) {$b$};
  \node[strand tag] at (0,8.48) {$B_0$};
  \foreach \x/\lab in {1/A_0,2/A_1,3/A_2} {
    \node[state] at (\x,9) {$a$};
    \node[strand tag] at (\x,8.48) {$\lab$};
  }
  \node[marker state] at (4,9) {$m$};
  \node[strand tag] at (4,8.48) {$A_3$};
  \foreach \x/\lab in {5/B_1,6/B_2,7/B_3} {
    \node[state] at (\x,9) {$b$};
    \node[strand tag] at (\x,8.48) {$\lab$};
  }
  \draw[A brace] (0.70,8.18) -- (4.30,8.18);
  \node[strand tag] at (2.5,7.77) {$A$};

  \node[row label] at (-1.10,6) {after \(B_1\)};
  \node[state] at (0,6) {$b$};
  \node[strand tag] at (0,5.48) {$B_0$};
  \node[bell one leg] (r1b1) at (1,6) {};
  \node[strand tag] at (1,5.48) {$B_1$};
  \foreach \x/\lab in {2/A_0,3/A_1} {
    \node[state] at (\x,6) {$a$};
    \node[strand tag] at (\x,5.48) {$\lab$};
  }
  \node[marker state] at (4,6) {$m$};
  \node[strand tag] at (4,5.48) {$A_2$};
  \node[bell one leg] (r1a3) at (5,6) {};
  \node[strand tag] at (5,5.48) {$A_3$};
  \foreach \x/\lab in {6/B_2,7/B_3} {
    \node[state] at (\x,6) {$b$};
    \node[strand tag] at (\x,5.48) {$\lab$};
  }
  \draw[bell one]
    (r1b1.north) .. controls +(0,1.15) and +(0,1.15) .. (r1a3.north);
  \draw[A brace] (1.70,5.18) -- (5.30,5.18);
  \node[strand tag] at (3.5,4.77) {$A$};

  \node[row label] at (-1.10,3) {after \(B_2\)};
  \node[state] at (0,3) {$b$};
  \node[strand tag] at (0,2.48) {$B_0$};
  \node[bell one leg] (r2b1) at (1,3) {};
  \node[strand tag] at (1,2.48) {$B_1$};
  \node[bell two leg] (r2b2) at (2,3) {};
  \node[strand tag] at (2,2.48) {$B_2$};
  \node[state] at (3,3) {$a$};
  \node[strand tag] at (3,2.48) {$A_0$};
  \node[marker state] at (4,3) {$m$};
  \node[strand tag] at (4,2.48) {$A_1$};
  \node[bell two leg] (r2a2) at (5,3) {};
  \node[strand tag] at (5,2.48) {$A_2$};
  \node[bell one leg] (r2a3) at (6,3) {};
  \node[strand tag] at (6,2.48) {$A_3$};
  \node[state] at (7,3) {$b$};
  \node[strand tag] at (7,2.48) {$B_3$};
  \draw[bell two]
    (r2b2.north) .. controls +(0,0.82) and +(0,0.82) .. (r2a2.north);
  \draw[bell one]
    (r2b1.north) .. controls +(0,1.38) and +(0,1.38) .. (r2a3.north);
  \draw[A brace] (2.70,2.18) -- (6.30,2.18);
  \node[strand tag] at (4.5,1.77) {$A$};

  \node[row label] at (-1.10,0) {after \(B_3\)};
  \node[state] at (0,0) {$b$};
  \node[strand tag] at (0,-0.52) {$B_0$};
  \node[bell one leg] (r3b1) at (1,0) {};
  \node[strand tag] at (1,-0.52) {$B_1$};
  \node[bell two leg] (r3b2) at (2,0) {};
  \node[strand tag] at (2,-0.52) {$B_2$};
  \node[bell three leg] (r3b3) at (3,0) {};
  \node[strand tag] at (3,-0.52) {$B_3$};
  \node[marker state] at (4,0) {$m$};
  \node[strand tag] at (4,-0.52) {$A_0$};
  \node[bell three leg] (r3a1) at (5,0) {};
  \node[strand tag] at (5,-0.52) {$A_1$};
  \node[bell two leg] (r3a2) at (6,0) {};
  \node[strand tag] at (6,-0.52) {$A_2$};
  \node[bell one leg] (r3a3) at (7,0) {};
  \node[strand tag] at (7,-0.52) {$A_3$};
  \draw[bell three]
    (r3b3.north) .. controls +(0,0.66) and +(0,0.66) .. (r3a1.north);
  \draw[bell two]
    (r3b2.north) .. controls +(0,1.12) and +(0,1.12) .. (r3a2.north);
  \draw[bell one]
    (r3b1.north) .. controls +(0,1.60) and +(0,1.60) .. (r3a3.north);
  \draw[A brace] (3.70,-0.82) -- (7.30,-0.82);
  \node[strand tag] at (5.5,-1.23) {$A$};

  \draw[opent/flow arrow] (-1.72,11.35) -- (-1.72,9.65);
  \node[step label] at (-1.92,10.50) {row \(B_0\)};
  \draw[opent/flow arrow] (-1.72,8.35) -- (-1.72,6.65);
  \node[step label] at (-1.92,7.50) {row \(B_1\)};
  \draw[opent/flow arrow] (-1.72,5.35) -- (-1.72,3.65);
  \node[step label] at (-1.92,4.50) {row \(B_2\)};
  \draw[opent/flow arrow] (-1.72,2.35) -- (-1.72,0.65);
  \node[step label] at (-1.92,1.50) {row \(B_3\)};

\end{tikzpicture}
\caption{The row process for \(t=4\).  The small labels identify the
strands, while the grey underbrace follows the
four \(A\)-strands as they move one site to the right in every row.  Ordinary
letters denote product states, the orange node denotes the marker, and two
equally coloured circles joined by an arc denote the two legs of one Bell
pair \(\ket\beta\).  The marker remains at the same position, and each row
after \(B_0\) creates one new Bell pair whose two legs lie on opposite sides
of the marker.  After the last row, the three nested Bell pairs give Schmidt
rank \(2^3\) across the fixed spatial cut.}
\label{fig:bo-marker-process}
\end{figure}

The first three rows have now exposed the full repeating process, which is
the same for every \(t\).  These computations exhaust the local events that
can occur.  A new blank first crosses the \(A\)-legs of the existing Bell
pairs without changing them, then uses the marker and the first remaining
vacuum to create one additional Bell pair.  Finally, the \(B\)-leg of the
new Bell pair crosses the remaining vacua without changing the pair.
Figure~\ref{fig:bo-marker-process} displays this calculation at fixed
spin-chain sites and also includes the last row \(B_3\).  The strand labels
move from row to row, but the marker remains at site \(0\), between the two
legs of every existing Bell pair.

After rows \(B_0,\ldots,B_{t-2}\), one vacuum remains on \(A_0\), the
marker is on \(A_1\), and the last blank is on \(B_{t-1}\).  The final row
first crosses the existing Bell pairs and then applies the same two-crossing
process to these three states.  No vacuum or blank remains afterward, and we
obtain
\begin{equation}
 \ket{\Psi_t}
 =\ket b_{B_0}\ket m_{A_0}
  \bigotimes_{j=1}^{t-1}\ket\beta_{B_j,A_{t-j}},
 \qquad t\geq2.
 \label{eq:bo-Psi-product}
\end{equation}

Every Bell pair \(\ket\beta\) has Schmidt rank two across
\(B_{\rm out}\mid A_{\rm out}\).  Therefore the selected product vector already
implies
\begin{equation}
 2^{t-1}=\SR_{B\mid A}(\Psi_t)
 \leq\OSR_{B\mid A}(\Phi_t(O))
 =\OSR_c(O(t)),\qquad t\geq2.
 \label{eq:bo-final}
\end{equation}
The last equality is the rectangular-core identity of
Lemma~\ref{lem:yb-rectangle}.  This proves Eq.~\eqref{eq:bo-osr-theorem} and
hence the exponential lower bound \eqref{eq:results-braid-only}.
\end{proof}

\section{Conclusions}
\label{sec:conclusions}

In this work we studied operator entanglement in circuits with Yang--Baxter
gates, as well as in a broader class of phase-dressed Yang--Baxter-map
circuits.  In several families we proved logarithmic upper bounds, or even
uniform bounds, on the operator entropies; the main results were summarized
earlier in \ref{subsec:main-results}.

The mechanisms for this special behaviour differ from those that appeared earlier in the literature. For
example, we do not directly use conservation of particles, in contrast to earlier proofs in the Rule 54 or hard-core gas
automata.
In our circuits the main mechanism is the existence of a highly non-local transformation that connects the interacting
circuits (with no particle or label conservation) to much more tractable cases with explicit label conservation.
In the case of the dual-unitary involutive family we used 
the main theorem of \cite{LechnerPennigWood2019}, guaranteeing the existence
of the intertwiners $U_n$ in \eqref{eq:inv-LPW-intertwiner},
whereas in the case of the involutive non-degenerate Yang--Baxter maps
the key object is the so-called guitar map  \cite{EtingofSchedlerSoloviev1999}.
For phase-dressed Yang--Baxter-map gates, the resulting polynomial bound
has a broader scope than the integrable setting: the braid relation is needed
only for the classical support, not for the full phase-dressed quantum gate.
When the support is involutive and the full dressed gate does satisfy the
braid relation, the centralizer and guitar-map arguments combine to give the
{additional} bound of Eq.~\eqref{eq:yb-phased-involutive-bound}, without requiring
the full gate to be involutive.

The motivation to consider these bounds is both theoretical and practical.  On the
theoretical side, they help clarify the complexity of time evolution in
quantum many-body systems and determine how strongly integrability constrains
this complexity.  On the practical side, a
polynomial bound on the exact OSR implies that an exact MPO with polynomial
bond dimension exists, while a sufficiently compressible operator-Schmidt
spectrum can allow efficient approximate MPO simulation.  This raises the
question whether the braid relation alone is sufficient to guarantee an
efficient MPO description of the Heisenberg evolution.

It is important to distinguish efficient simulability from an explicit
solution.  Knowing that an MPO with manageable bond dimension exists does not
mean that we have explicit formulas for its local tensors.  A separate open
question is therefore when exact time-dependent MPOs can be constructed
explicitly.  For the slowly growing families considered here, our aim was to
establish upper bounds on the OSR and operator entanglement rather than to
derive exact MPO formulas.

At the level of operator entanglement, the most important open question is
whether the braid relation alone is enough to guarantee sublinear growth of
the von Neumann operator entanglement.  We have shown that the exact operator Schmidt rank
can grow exponentially, but the rank does not determine the
operator entropies.  In particular, their growth for the seven-state gate of
Section~\ref{sec:braid-only} remains open.  More generally, it is not known
whether sublinear von Neumann operator entanglement follows from the braid
relation or whether a Yang--Baxter counterexample exists.
Concrete integrable models displaying logarithmic growth or saturation were
analyzed in Refs.~\cite{ProsenPizorn2007,PizornProsen2009,KlobasMedenjakProsenVanicat2019,AlbaDubailMedenjak2019,Medenjak2022,Dubail2017,MuthUnanyanFleischhauer2011,MurcianoDubailCalabrese2024};
the broader expectation that the operator entanglement of local observables
grows at most logarithmically under integrable dynamics was stated, explicitly
or in closely related simulability terms, in
Refs.~\cite{AlbaDubailMedenjak2019,BertiniKosProsenI,Medenjak2022,Dubail2017,ProsenZnidaric2007,Alba2021,CarignanoMarimonTagliacozzo2024,Alba2025,DowlingModiWhite2025,CerezoRoquebrunEtAl2025,JacobyGopalakrishnan2026,Dowling2026}.
{
The stronger logarithmic bound is ruled out by the separate result discussed
below in the final paragraph of this Section.  Whether the braid relation
guarantees sublinear von Neumann operator entanglement remains open.
}

{
During this work, we also compared two of our main families: involutive
dual-unitary Yang--Baxter gates and phase-dressed gates obtained from
non-degenerate Yang--Baxter maps.  Our searches at small local dimensions
found no gate in the first family that could not be brought to phase-dressed
permutation form by a local unitary change of basis.  This raised the
question whether the first family might be contained in the second.
We later found an example at local dimension $D=8$, constructed in
Appendix~\ref{app:nonmonomial-involutive} below, which admits no such form
even when different local bases are allowed on the two sites.  Since the
phase-dressed family also contains noninvolutive gates, neither family
contains the other, although they overlap.  The representation-theoretic
proof for the involutive dual-unitary family therefore applies beyond the
phase-dressed permutation setting.

The smallest local dimension allowing such an example and a full
classification of involutive dual-unitary gates remain open.
It would also be useful to understand how this construction fits into the
broader framework of
Galindo and Rowell~\cite{GalindoRowell2026}.
}

\bigskip

{
In the separate manuscript~\cite{sajat-sqrt}, we construct an involutive
Yang--Baxter gate without dual unitarity using ChatGPT 5.6 Sol.  The
corresponding brickwork circuit produces square-root growth of the von
Neumann operator entanglement entropy.  This growth law is proved
analytically and provides a counterexample to the expectation that
local-operator entanglement in integrable circuits grows at most
logarithmically in time.  The result does not contradict any statement in
this manuscript: the gate lies outside the classes for which we prove
logarithmic upper bounds here.
}

\section*{Acknowledgments}

The author was supported by the Hungarian National Research, Development
and Innovation Office, NKFIH Grant No.~K-145904.

\appendix

\section{Gliders for general Yang--Baxter gates}
\label{app:general-gliders}

In this appendix we verify the glider relation
\eqref{eq:setup-general-glider-translation} directly.  Write
$R_j=R_{j,j+1}$, so that
\begin{equation}
 R_jR_{j+1}R_j
 =R_{j+1}R_jR_{j+1},
 \qquad
 R_jR_l=R_lR_j
 \quad (|j-l|>1),
 \label{eq:app-glider-braid-relations}
\end{equation}
and $R_j^{-1}=R_j^\dagger$.  For even $i$, the first operator in
\eqref{eq:setup-general-gliders} is the braid-group band generator
\begin{equation}
 b_i=R_i^{-1}R_{i+1}R_i.
 \label{eq:app-glider-band-even}
\end{equation}
Only five gates in the two Floquet layers are relevant when $b_i$ is moved
through the circuit.  Denote their product by
\begin{equation}
 F_i=(R_{i-1}R_{i+1}R_{i+3})
     (R_iR_{i+2}).
 \label{eq:app-glider-local-floquet}
\end{equation}
The gates outside this causal window cancel against their inverses.  Repeatedly
using \eqref{eq:app-glider-braid-relations}, both sides of the desired
intertwining relation reduce to the same word:
\begin{align}
 F_i b_i
 &=R_{i-1}R_{i+3}
   R_{i+2}R_{i+1}R_{i+2}R_i
 \nonumber\\
 &=(R_{i+2}^{-1}R_{i+3}R_{i+2})F_i.
 \label{eq:app-glider-even-computation}
\end{align}
For example, the second equality uses
$R_{i+2}^{-1}R_{i+3}R_{i+2}
=R_{i+3}R_{i+2}R_{i+3}^{-1}$, which is the braid relation
with inverses.  Restoring the commuting gates outside this window gives
\begin{equation}
 U_F b_iU_F^{-1}=b_{i+2},
 \qquad i\ \text{even}.
 \label{eq:app-glider-even-translation}
\end{equation}

On the other sublattice the corresponding band generator has the opposite
orientation,
\begin{equation}
 \widetilde b_i=R_iR_{i+1}R_i^{-1},
 \qquad i\ \text{odd}.
 \label{eq:app-glider-band-odd}
\end{equation}
The same local calculation, with the two layers interchanged, gives
\begin{equation}
 U_F\widetilde b_iU_F^{-1}=\widetilde b_{i-2},
 \qquad i\ \text{odd}.
 \label{eq:app-glider-odd-translation}
\end{equation}
Equations~\eqref{eq:app-glider-even-translation} and
\eqref{eq:app-glider-odd-translation} are precisely
\eqref{eq:setup-general-glider-translation}.  Their adjoints are gliders as
well, and any finite product on one sublattice translates rigidly because
conjugation preserves products.  No involutivity assumption is used.

\section{Integrability of the circuits with involutive Yang--Baxter gates}

\label{app:tms}

In this appendix we establish integrability in the traditional transfer-matrix
sense for circuits generated by involutive Yang--Baxter gates.  We first recall
the staggered transfer-matrix construction of integrable Trotterization from
Ref.~\cite{VanicatZadnikProsen2018} and then specialize it to the Baxterized
$R$-matrix associated with our gate.  This specialization gives the commuting
family $T(u;\theta)$ defined in Eq.~\eqref{eq:int-staggered-transfer}.  We show
that the ratio of its values at the two regular points is the finite-parameter
brickwork circuit and that the braid limit yields commuting transfer matrices
for the circuit $U_F$ studied in the main text.

Consider a periodic chain of even length $L$, choose a real parameter $\theta$,
and introduce the alternating inhomogeneities
\begin{equation}
 \xi_j=(-1)^j\theta,
 \qquad j=1,\ldots,L.
 \label{eq:int-inhomogeneities}
\end{equation}
We use the Baxterized matrix $\check{\mathcal R}(u)$ introduced in
Eq.~\eqref{eq:int-unitary-baxterization} and its vertex form
$\mathcal R(u)=\Swap\check{\mathcal R}(u)$ from Eq.~\eqref{eq:int-vertex-R}.
The corresponding row-to-row transfer matrix is
\begin{equation}
 T(u;\theta)
 =\Tr_a\!
 \left[
 \mathcal R_{aL}(u-\xi_L)
 \mathcal R_{a,L-1}(u-\xi_{L-1})
 \cdots
 \mathcal R_{a1}(u-\xi_1)
 \right].
 \label{eq:int-staggered-transfer}
\end{equation}
The ordinary Yang--Baxter equation gives
\begin{equation}
 [T(u;\theta),T(v;\theta)]=0
 \qquad
 \text{for all }u,v.
 \label{eq:int-transfer-commutativity}
\end{equation}
All commutativity statements in this paragraph and below are at fixed
$\theta$.  Different values of the Trotter parameter define different
transfer matrices and, in general, different conserved charges.

The physical two-site gate at finite $\theta$ is
\begin{equation}
 G(\theta)=\check{\mathcal R}(2\theta).
 \label{eq:int-finite-gate}
\end{equation}
For every finite value of the inhomogeneity $\theta$, we use this gate to
build a brickwork time evolution.  We order the two layers with the odd layer
acting after the even layer and set
\begin{align}
 U_{\rm odd}(\theta)
 &=\prod_{j=1}^{L/2}G_{2j-1,2j}(\theta),
 &
 U_{\rm even}(\theta)
 &=\prod_{j=1}^{L/2}G_{2j,2j+1}(\theta),
 \label{eq:int-finite-layers}\\
 U_\theta&=U_{\rm odd}(\theta)U_{\rm even}(\theta),
 \qquad L+1\equiv1.
 \label{eq:int-finite-floquet}
\end{align}
The evolutions $U_\theta$ corresponding to different values of $\theta$ do
not commute in general.  The factor $2\theta$ is the difference between the
two alternating inhomogeneities.  Equivalently, one could choose
inhomogeneities $\pm\theta/2$ and write the local gate as
$\check{\mathcal R}(\theta)$.

Specializing the staggered transfer matrix to the two regular points
$u=\pm\theta$, and using $G(-\theta)=G(\theta)^{-1}$ from
Eq.~\eqref{eq:int-baxterization-properties}, gives the exact identity
\begin{equation}
 T(-\theta;\theta)^{-1}T(\theta;\theta)
 =U_{\rm odd}(\theta)U_{\rm even}(\theta)
 =U_\theta.
 \label{eq:int-exact-trotter-identity}
\end{equation}
No projective scalar is present in this normalization.  In particular,
$U_\theta$ is a unitary brickwork circuit for every real $\theta$.

Combining Eqs.~\eqref{eq:int-transfer-commutativity} and
\eqref{eq:int-exact-trotter-identity}, we obtain
\begin{equation}
 [U_\theta,T(u;\theta)]=0
 \qquad
 \text{for all }u.
 \label{eq:int-finite-U-commutes}
\end{equation}
The two regular points generate two staggered families of local charges.  A
branch-independent definition is obtained by expanding the transfer matrix
relative to its value at the regular point:
\begin{equation}
 Q_n^{(\sigma)}(\theta)
 =
 -i\left.
 \frac{d^n}{du^n}
 \log\!\left[
 T(u;\theta)T(\sigma\theta;\theta)^{-1}
 \right]
 \right|_{u=\sigma\theta},
 \qquad
 \sigma\in\{+,-\},\quad n\geq1.
 \label{eq:int-finite-charges}
\end{equation}
Here the logarithm is understood through its Taylor series around the
identity.  {
For each fixed order $n$, standard transfer-matrix arguments express the
charge as a sum of finite-range local densities on a sufficiently long
chain.  The required size may depend on $n$; we take $L$ large enough that
these densities do not wrap around the periodic chain.
} Combining their definition in
Eq.~\eqref{eq:int-finite-charges} with the commutativity relation
\eqref{eq:int-finite-U-commutes}, we obtain
\begin{equation}
 [Q_n^{(\sigma)}(\theta),U_\theta]=0.
 \label{eq:int-finite-charge-conservation}
\end{equation}
{
In this sufficiently large-volume regime, the normalization in
Eq.~\eqref{eq:int-finite-charges} gives Hermitian charges.
} This is most easily seen for the first braid-limit charges in
Eq.~\eqref{eq:int-endpoint-first-charges}.  Up to additive identity terms,
they are sums of the elementary gliders, which for an involutive gate take the
concrete form
\begin{align}
 \mathcal G^{\rm e}_{2k}
 &=R_{2k}R_{2k+1}R_{2k},&
 \mathcal G^{\rm o}_{2k-1}
 &=R_{2k-1}R_{2k}R_{2k-1}.
 \label{eq:int-involutive-glider-densities}
\end{align}
Since $R_j^\dagger=R_j$, both operators are manifestly Hermitian.

For small values of the inhomogeneity, these integrable circuits approximate
continuous-time evolution generated by a local Hamiltonian.  To identify this
Hamiltonian, we expand the local gate:
\begin{equation}
 G(\theta)
 =\id+2i\theta(R-\id)+O(\theta^2),
 \label{eq:int-small-theta-local}
\end{equation}
and therefore
\begin{equation}
 U_\theta
 =\id+2i\theta
 \sum_{j=1}^{L}(R_{j,j+1}-\id)
 +O(\theta^2).
 \label{eq:int-small-theta-global}
\end{equation}
Thus, up to the convention for the sign and scale of real time and up to an
additive multiple of the identity, the continuous-time Hamiltonian has local
density $R$,
\begin{equation}
 H_{\rm cont}\propto\sum_{j=1}^{L}R_{j,j+1}.
 \label{eq:int-continuous-Hamiltonian}
\end{equation}

We now turn to the circuit studied in the main text, whose two-site gate $R$
is the infinite-rapidity limit of the finite-parameter gate $G(\theta)$.
Indeed, Eq.~\eqref{eq:int-baxterization-properties} gives
\begin{equation}
 \lim_{\theta\to\infty}G(\theta)=R.
 \label{eq:int-local-braid-limit}
\end{equation}
Consequently,
\begin{equation}
 \lim_{\theta\to\infty}U_\theta
 =
 \left(\prod_{j=1}^{L/2}R_{2j-1,2j}\right)
 \left(\prod_{j=1}^{L/2}R_{2j,2j+1}\right)
 =U_F,
 \label{eq:int-pure-braid-limit}
\end{equation}
where the last equality uses the Floquet ordering of
Eq.~\eqref{eq:int-finite-floquet}.

To retain a nontrivial spectral parameter at this endpoint, we follow the two
moving regular points.  Keeping $u$ fixed while sending
$\theta\to\infty$ would send every local argument to $\pm\infty$ and would
remove the $u$ dependence.  We therefore define the recentered limits
\begin{equation}
 \tau_+(v)
 =\lim_{\theta\to\infty}T(\theta+v;\theta),
 \qquad
 \tau_-(v)
 =\lim_{\theta\to\infty}T(-\theta+v;\theta).
 \label{eq:int-recentered-transfer-limits}
\end{equation}
These limits exist at every finite $L$.  Indeed, let
\begin{equation}
 \mathcal R_\infty
 :=\lim_{u\to\pm\infty}\mathcal R(u)
 =\Swap R.
 \label{eq:int-vertex-infinite-limit2}
\end{equation}
Then
\begin{equation}
 \tau_\sigma(v)
 =\Tr_a\!
 \left[
 \mathcal L_{aL}^{(\sigma)}(v)
 \mathcal L_{a,L-1}^{(\sigma)}(v)
 \cdots
 \mathcal L_{a1}^{(\sigma)}(v)
 \right],
 \label{eq:int-endpoint-transfer-explicit2}
\end{equation}
where
\begin{align}
 \mathcal L_{aj}^{(+)}(v)
 &=
 \begin{cases}
  \mathcal R_{aj}(v),&j\ \text{even},\\
  (\mathcal R_\infty)_{aj},&j\ \text{odd},
 \end{cases}
 &
 \mathcal L_{aj}^{(-)}(v)
 &=
 \begin{cases}
  (\mathcal R_\infty)_{aj},&j\ \text{even},\\
  \mathcal R_{aj}(v),&j\ \text{odd}.
 \end{cases}
 \label{eq:int-endpoint-Lax-operators2}
\end{align}
The finite-$\theta$ transfer matrices commute for arbitrary spectral
parameters.  Taking the finite-dimensional limits in
Eq.~\eqref{eq:int-transfer-commutativity} gives the full endpoint
commutativity relations
\begin{equation}
 [\tau_\sigma(v),\tau_{\sigma'}(w)]=0,
 \qquad
 \sigma,\sigma'\in\{+,-\}.
 \label{eq:int-endpoint-transfer-commutativity2}
\end{equation}
At $v=0$, the two endpoint transfer matrices are the limits of the two
special transfer matrices.  Equation~\eqref{eq:int-exact-trotter-identity}
therefore gives
\begin{equation}
 U_F=\tau_-(0)^{-1}\tau_+(0).
 \label{eq:int-endpoint-Floquet-ratio2}
\end{equation}
Combining this result with
Eq.~\eqref{eq:int-endpoint-transfer-commutativity}, we obtain
\begin{equation}
 [U_F,\tau_\pm(v)]=0
 \qquad
 \text{for all }v.
 \label{eq:int-endpoint-U-commutes2}
\end{equation}
This establishes the Yang--Baxter integrability of the pure braid circuit
directly at the endpoint.

\section{Conjugation property for controlled-swap gates}
\label{app:controlled-swap-conjugation}

In this appendix we prove the conjugation statement used in
Lemma~\ref{lem:controlled-swap-finite-conjugation-memory}.  We first convert
the matrix-element constraint in
Eq.~\eqref{eq:controlled-swap-conjugation-rule} into an action on coordinate
subspaces graded by the unitaries in \(\mathcal S\).  We then show that every
generator, and hence every element of the generated group, permutes this
finite set.  This is the direct form, in our notation and braid convention,
of the group-type braided-vector-space argument in Proposition~4.2 and
Subsection~4.2 of Ref.~\cite{GalindoRowell2014}.

Let \(\mathcal S\) be the set in
Eq.~\eqref{eq:controlled-swap-support-set}.  For any \(v\in U(V)\), define
the coordinate subspace
\begin{equation*}
 V_v=\operatorname{span}\{\ket b:u_b=v\}.
\end{equation*}
These subspaces are mutually orthogonal,
\(\bigoplus_{v\in\mathcal S}V_v=V\), and \(V_v\ne\{0\}\) precisely when
\(v\in\mathcal S\).

We now derive the action on these spaces directly from
Eq.~\eqref{eq:controlled-swap-conjugation-rule}.  Take \(v\in\mathcal S\)
and a basis vector \(\ket b\in V_v\), so that \(u_b=v\), and expand
\[
 u_a\ket b=\sum_x(u_a)_{xb}\ket x.
\]
For every nonzero coefficient, Eq.~\eqref{eq:controlled-swap-conjugation-rule}
gives \(u_x=u_avu_a^{-1}\).  Hence every basis vector appearing in the
expansion belongs to \(V_{u_avu_a^{-1}}\).  Taking the span over all
\(\ket b\in V_v\) gives
\begin{equation}
 u_aV_v\subseteq V_{u_av u_a^{-1}}.
 \label{eq:app-controlled-swap-homogeneous-spaces}
\end{equation}
Conversely, suppose that
Eq.~\eqref{eq:app-controlled-swap-homogeneous-spaces} holds.  If
\((u_a)_{xb}\ne0\) and \(v=u_b\), then the basis vector \(\ket x\) occurs
in \(u_a\ket b\), which belongs to \(V_{u_avu_a^{-1}}\).  The definition of
the coordinate subspaces therefore gives
\(u_x=u_avu_a^{-1}=u_au_bu_a^{-1}\), recovering
Eq.~\eqref{eq:controlled-swap-conjugation-rule}.  The two formulations are
thus equivalent, as stated in Proposition~4.2 of
Ref.~\cite{GalindoRowell2014}.

It remains to verify that the target label belongs to \(\mathcal S\).  Fix
\(v\in\mathcal S\) and choose \(b\) such that \(u_b=v\).  Since \(u_a\) is
unitary,
\[
 \sum_x\lvert(u_a)_{xb}\rvert^2=1,
\]
so there is at least one \(x\) for which \((u_a)_{xb}\ne0\).
Equation~\eqref{eq:controlled-swap-conjugation-rule} then gives
\(u_x=u_au_bu_a^{-1}=u_avu_a^{-1}\).  Since \(u_x\in\mathcal S\) by
definition, this proves \(u_avu_a^{-1}\in\mathcal S\).  Thus conjugation by
\(u_a\) defines a self-map
\[
 C_a:\mathcal S\longrightarrow\mathcal S,
 \qquad C_a(v)=u_avu_a^{-1}.
\]
This map is injective: if \(C_a(v)=C_a(w)\), multiplication by \(u_a^{-1}\)
from the left and by \(u_a\) from the right gives \(v=w\).  Since
\(\mathcal S\) is finite, \(C_a\) is also surjective and therefore a
permutation.  Its inverse is conjugation by \(u_a^{-1}\), which consequently
also permutes \(\mathcal S\).

Finally, let \(G=\langle u_a:a\in X\rangle\), as in
Lemma~\ref{lem:controlled-swap-finite-conjugation-memory}, and write an
arbitrary element as
\[
 g=u_{a_k}^{\epsilon_k}\cdots u_{a_1}^{\epsilon_1},
 \qquad \epsilon_j\in\{1,-1\}.
\]
Conjugation by \(g\) is the composition
\(C_{a_k}^{\epsilon_k}\circ\cdots\circ C_{a_1}^{\epsilon_1}\).  Every
factor is a permutation of \(\mathcal S\), so their composition is also a
permutation.  Therefore
\[
 g\mathcal Sg^{-1}=\mathcal S,
 \qquad g\in G.
\]
This proves Eq.~\eqref{eq:controlled-swap-conjugation-memory} and completes
the proof of Lemma~\ref{lem:controlled-swap-finite-conjugation-memory}.

\section{Proof of the active-block Schmidt decomposition}
\label{app:normal-form-label-grouping}

In this appendix we prove
Theorem~\ref{thm:normal-form-label-grouping}.  We evolve the
computational-basis terms of a single matrix unit, separate the spectator
blocks from the blocks containing its distinguished entries, and collect the
data that still correlate the two outputs into labels.  We then show that
fixing these labels makes the remaining output choices independent and count
the possible labels.

Let us recall the notation needed for the proof.  The signed-block normal
form has a decomposition
\begin{equation*}
 V=\bigoplus_\gamma V_\gamma,
 \qquad
 d_\gamma=\dim V_\gamma,
 \qquad
 \varepsilon_\gamma\in\{+1,-1\}.
\end{equation*}
For each block, $X_\gamma$ is a set of basis labels with
$|X_\gamma|=d_\gamma$, and the sets $X_\gamma$ are mutually disjoint.  We
write $X=\bigsqcup_\gamma X_\gamma$.  For any word $Z$ over $X$, let
$Z_\gamma$ be its ordered subword with letters in $X_\gamma$, and let
$n_\gamma(Z)=|Z_\gamma|$.

We use the rectangular map $\Phi_t$, whose output factors are denoted by
$B_{\rm out}$ and $A_{\rm out}$.  The single-file rule of
Lemma~\ref{lem:normal-form-single-file} states that an input $A\mid B$ gives
output words $L\mid R$ satisfying
\begin{equation*}
 L_\gamma R_\gamma=A_\gamma B_\gamma,
 \qquad
 |L_\gamma|=n_\gamma(B),
 \qquad
 |R_\gamma|=n_\gamma(A)
\end{equation*}
in every block.  The corresponding sign is
$\prod_\gamma\varepsilon_\gamma^{n_\gamma(A)n_\gamma(B)}$.

\begin{proof}[Proof of Theorem~\ref{thm:normal-form-label-grouping}]
Fix $x\in X_\alpha$ and $y\in X_\beta$; when $\alpha=\beta$, the labels
$x$ and $y$ need not be distinct.  Writing $E_{xy}=\ket{x}\bra{y}$, the
source expansion consists of the input matrix units
\begin{equation}
 \ket{A,xC}\bra{A,yC},
 \qquad
 A\in X^t,
 \quad C\in X^{t-1}.
 \label{eq:app-normal-form-source-terms}
\end{equation}
Denote the outputs of the ket word $A\mid xC$ by $L_x\mid R_x$ and those of
the bra word $A\mid yC$ by $L_y\mid R_y$.  Applying the single-file rule to
the two layers of each term in
\eqref{eq:app-normal-form-source-terms} gives the output product
\begin{equation}
 \sigma_x\sigma_y
 \bigl(\ket{L_x}\bra{L_y}\bigr)_{B_{\rm out}}
 \otimes
 \bigl(\ket{R_x}\bra{R_y}\bigr)_{A_{\rm out}},
 \qquad \sigma_x,\sigma_y\in\{\pm1\}.
 \label{eq:normal-form-output-product-unit}
\end{equation}
For every block $\gamma$, write
\begin{equation*}
 a_\gamma=n_\gamma(A)=|A_\gamma|,
 \qquad
 k_\gamma=n_\gamma(C)=|C_\gamma|.
\end{equation*}

We first identify the label data and show that, once these data are fixed,
the left and right output choices vary independently.  Consider a spectator
block $\delta\notin\{\alpha,\beta\}$.  Its ket and bra inputs coincide, and
the single-file rule gives the bijection
\begin{equation}
 (A_\delta,C_\delta)
 \longleftrightarrow(L_\delta,R_\delta),
 \qquad
 L_\delta R_\delta=A_\delta C_\delta,
 \quad |L_\delta|=k_\delta,
 \quad |R_\delta|=a_\delta.
 \label{eq:normal-form-spectator-bijection}
\end{equation}
Hence its internal colours can be summed independently in the two output
factors and require no Schmidt label.

We now specify the data retained by an active block.  Suppose first that
$\alpha=\beta$ and $d_\alpha>1$.  The ket and bra
$\alpha$-strings have the same split.  When $a_\alpha\leq k_\alpha$, the
distinguished entry lies in $B_{\rm out}$ and its position is fixed by
$a_\alpha$.  When $a_\alpha>k_\alpha$, it lies in $A_{\rm out}$ and the
correlation between the two outputs is fixed by $k_\alpha$.  We therefore
use the label
\begin{equation*}
 \ell_{\alpha\alpha}(A,C)=
 \begin{cases}
  (\mathsf B,a_\alpha),&a_\alpha\leq k_\alpha,\\
  (\mathsf A,k_\alpha),&a_\alpha>k_\alpha.
 \end{cases}
\end{equation*}
If $d_\alpha=1$, the two matrix-unit labels coincide and we set
$\ell_{\alpha\alpha}(A,C)=*$.

Now suppose that $\alpha\ne\beta$.  For each
$\gamma\in\{\alpha,\beta\}$ with $d_\gamma>1$, define
\begin{equation}
 \ell_\gamma(A,C)=
 \begin{cases}
  (\mathsf{local},a_\gamma),&a_\gamma\leq k_\gamma,\\
  (\mathsf{transfer},k_\gamma,u_\gamma),&a_\gamma>k_\gamma,
 \end{cases}
 \qquad
 u_\gamma=(A_\gamma)_{k_\gamma+1}\in X_\gamma.
 \label{eq:normal-form-active-label}
\end{equation}
In the first branch the ket--bra difference remains in one output factor, and
fixing $a_\gamma$ fixes its position.  In the second branch, write
$A_\gamma=A_\gamma^{\rm L}u_\gamma A_\gamma^{\rm R}$ with
$|A_\gamma^{\rm L}|=k_\gamma$.  The colour $u_\gamma$ lies at the split: it
belongs to different output factors on the two layers.  Fixing
$(k_\gamma,u_\gamma)$ removes this correlation.  If $d_\gamma=1$, the only
possible transferred colour is fixed, and the ket--bra matrix unit is already
determined by the block-label words on the two sides.  We then set
$\ell_\gamma(A,C)=*$.

Let $\Lambda_{\alpha\alpha}(t)$ be the image of
$\ell_{\alpha\alpha}$ in the same-block case.  In the different-block case,
let $\Lambda_{\alpha\beta}(t)$ be the image of the joint label
$(\ell_\alpha,\ell_\beta)$.

This also accounts for arbitrary interleavings of different blocks.  The
block-label words on $B_{\rm out}$ are those of $xC$ and $yC$, whereas the
block-label word on $A_{\rm out}$ is that of $A$.  They are therefore
side-local data.  Conversely, fix a label $\lambda$, and let
$\mathcal F_\lambda$ and $\mathcal G_\lambda$ be the corresponding sets of
left and right output matrix units.  The single-file reconstruction in each
block combines every $F\in\mathcal F_\lambda$ with every
$G\in\mathcal G_\lambda$ and reconstructs a unique pair $(A,C)$.  Thus the
two sets of choices are independent.

The signs do not obstruct this factorization.  They cancel when
$\alpha=\beta$.  For $\alpha\ne\beta$, their ket--bra ratio is
$\varepsilon_\alpha^{a_\alpha}\varepsilon_\beta^{a_\beta}$, which depends
only on the right output block word and can be absorbed into the right factor.
Consequently, with a suitable sign $\eta_\lambda(G)$,
\begin{equation*}
 L_\lambda=\sum_{F\in\mathcal F_\lambda}F,
 \qquad
 R_\lambda=\sum_{G\in\mathcal G_\lambda}
 \eta_\lambda(G)G
\end{equation*}
gives \eqref{eq:normal-form-grouped-product}.

It remains to count the labels.  In the same-block case there are at most
$t$ values in each of the $\mathsf B$ and $\mathsf A$ branches, while a
singleton has the single label $*$.  For $\alpha\ne\beta$, a non-singleton
block contributes at most $t$ labels of the form
$(\mathsf{local},a_\gamma)$ and at most $td_\gamma$ labels of the form
$(\mathsf{transfer},k_\gamma,u_\gamma)$.  A singleton again contributes one
label.  The two active-block labels form a joint tuple, which proves
\eqref{eq:normal-form-same-block-bound}--\eqref{eq:normal-form-active-label-count}.
\end{proof}

{
\section{An involutive dual-unitary Yang--Baxter gate without a monomial local form}
\label{app:nonmonomial-involutive}

In this appendix we ask whether the involutive dual-unitary structure can
always be reduced, by a local change of basis, to a classical permutation
with phases. We give an explicit counterexample. The construction defined
in Subsection~\ref{subsec:nonmonomial-construction} couples a six-state
permutation sector to a two-dimensional sector through unitary rotations.

\begin{theorem}[No monomial local form]
\label{thm:nonmonomial-local-form}
There exists a unitary gate $R_\star$ on $V\otimes V$, with
$V\cong\mathbb C^8$, such that
\begin{equation}
 R_\star^2=\id,\qquad
 (R_\star^\Gamma)^\dagger R_\star^\Gamma=\id,\qquad
 (R_\star)_{12}(R_\star)_{23}(R_\star)_{12}
 =(R_\star)_{23}(R_\star)_{12}(R_\star)_{23},
 \label{eq:nonmonomial-gate-conditions}
\end{equation}
but
\begin{equation}
 (V_1^\dagger\otimes V_2^\dagger)R_\star(V_1\otimes V_2)
 \label{eq:nonmonomial-local-conjugate}
\end{equation}
is not monomial for any $V_1,V_2\in U(8)$.
\end{theorem}

Here monomial means that every row and every column has exactly one
nonzero entry, of unit modulus; arbitrary phases are allowed. We use the
reshuffling convention of Eq.~\eqref{eq:setup-reshuffle} and unnormalized
partial traces. The theorem excludes local unitary conjugations, including
the homogeneous case $V_1=V_2$.

The proof has three steps. We define the gate in
Subsection~\ref{subsec:nonmonomial-construction}, verify all four gate
conditions analytically in Subsection~\ref{subsec:nonmonomial-verification},
and exclude every monomial local form in
Subsection~\ref{subsec:nonmonomial-obstruction}.

\subsection{The construction of the gate \texorpdfstring{$R_\star$}{R*}}
\label{subsec:nonmonomial-construction}

We construct $R_\star$ on $V\otimes V$, where the eight-dimensional local
space is decomposed as $V=V_6\oplus V_2$, with $\dim V_6=6$ and
$\dim V_2=2$. We first define a permutation rule for pairs of basis states
in $V_6$, followed by unitary rotations acting on $V_2$. We then combine
these ingredients into the gate specified below in
Eq.~\eqref{eq:nonmonomial-gate}.

We label the six basis states of $V_6$ by pairs $\xi=(a,z)$, where
$a=(a_1,a_2)$ is a nonzero vector in $\mathbb F_2^2$ and
$z\in\mathbb F_2$ is an additional bit.  Here $\mathbb F_2=\{0,1\}$, with
arithmetic modulo two. The label set is therefore
\[
 X_6=(\mathbb F_2^2\setminus\{(0,0)\})\times\mathbb F_2.
\]
There are three choices for $a$ and two for $z$, giving six labels.

The scattering rule uses a linear map $\mathsf C$ on the binary vectors
and a bilinear form $\beta$.  For $a=(a_1,a_2)$ and $b=(b_1,b_2)$, we set
\begin{equation}
 \begin{gathered}
 \mathsf C:\mathbb F_2^2\longrightarrow\mathbb F_2^2,
 \qquad \mathsf C(a_1,a_2)=(a_2,a_1+a_2),\\
 \beta:\mathbb F_2^2\times\mathbb F_2^2\longrightarrow\mathbb F_2,
 \qquad \beta(a,b)=a_1b_2+a_2b_1.
 \end{gathered}
 \label{eq:nonmonomial-label-data}
\end{equation}
Thus $\beta$ takes two vectors as arguments and returns one bit.
Directly from these definitions, $\mathsf C^3=\id$,
$\mathsf C+\mathsf C^2=\id$, and
$\beta(\mathsf C a,\mathsf C b)=\beta(a,b)$.
For $(a,z),(b,w)\in X_6$, define
$r:X_6\times X_6\to X_6\times X_6$ by
\begin{equation}
 r\bigl((a,z),(b,w)\bigr)
 =\bigl((\mathsf C b,w+\beta(a,b)),\,
 (\mathsf C^2a,z+\beta(b,\mathsf C a))\bigr).
 \label{eq:nonmonomial-six-state-map}
\end{equation}
The $\beta$ terms shift the bits $z,w$; they do not act on the vector
coordinates.  Non-degeneracy, involutivity, and the braid relation for $r$
are verified later in
Subsection~\ref{subsec:nonmonomial-verification}.

We now introduce the rotations that will act on the two-dimensional
sector $V_2$ in the mixed collisions specified below in
Eq.~\eqref{eq:nonmonomial-gate}. Using the Pauli matrices, we first set
\begin{equation}
 \mathsf I=i\sigma_x,\qquad
 \mathsf J=i\sigma_y,\qquad
 \mathsf K=\mathsf I\mathsf J=-i\sigma_z,\qquad
 \Omega=-\frac12(\id_2+\mathsf I+\mathsf J+\mathsf K).
 \label{eq:nonmonomial-quaternion-units}
\end{equation}
For $a=(a_1,a_2)\ne(0,0)$ and $z\in\mathbb F_2$, define
\begin{equation}
 q_a=\mathsf I^{a_1}\mathsf J^{a_2},\qquad
 T_{(a,z)}=(-1)^zq_a\Omega.
 \label{eq:nonmonomial-rotations}
\end{equation}
These matrices satisfy the compatibility relation
\begin{equation}
 r(\xi,\eta)=(\xi',\eta')
 \quad\Longrightarrow\quad
 T_\xi T_\eta=T_{\xi'}T_{\eta'},
 \label{eq:nonmonomial-compatibility}
\end{equation}
which we also prove in Subsection~\ref{subsec:nonmonomial-verification}.
It ensures that the rotation acquired by the two-dimensional state agrees
on the two sides of the braid relation. Relations of this form are the
defining relations of the structure group of an involutive set-theoretical
solution; see Ref.~\cite{CedoJespersDelRio2008}. The verification uses only
the explicit formulas
\eqref{eq:nonmonomial-label-data}--\eqref{eq:nonmonomial-rotations}.

Identifying $V_6=\mathbb C[X_6]$ and $V_2=\mathbb C^2$, we define
$R_\star$ by its action on the four tensor sectors:
\begin{equation}
 \begin{aligned}
 R_\star(\ket{\xi}\otimes\ket{\eta})
   &=\ket{\xi'}\otimes\ket{\eta'},
   &&r(\xi,\eta)=(\xi',\eta'),\\
 R_\star(\ket{\xi}\otimes\ket{\psi})
   &=T_\xi\ket{\psi}\otimes\ket{\xi},\\
 R_\star(\ket{\psi}\otimes\ket{\xi})
   &=\ket{\xi}\otimes T_\xi^{-1}\ket{\psi},\\
 R_\star(\ket{\psi}\otimes\ket{\phi})
   &=-\ket{\phi}\otimes\ket{\psi},
   &&\xi,\eta\in X_6,\quad\ket{\psi},\ket{\phi}\in V_2.
 \end{aligned}
 \label{eq:nonmonomial-gate}
\end{equation}
Thus two six-state labels scatter through $r$, a mixed collision exchanges
the sectors and rotates the two-dimensional state, and two states in $V_2$
undergo minus the ordinary swap. This specifies a $64\times64$ matrix.

\subsection{Elementary properties of the gate}
\label{subsec:nonmonomial-verification}

We first verify the six-state map and the rotations, then check the full
gate sector by sector. We also record the traces and a commutator needed
later in Subsection~\ref{subsec:nonmonomial-obstruction}.

\begin{lemma}[The six-state Yang--Baxter map]
\label{lem:nonmonomial-six-state-map}
The map $r$ defined in Eq.~\eqref{eq:nonmonomial-six-state-map} is
non-degenerate and involutive, and satisfies the braid relation on
$X_6^{\times3}$.
\end{lemma}
\begin{proof}
With $(a,z)$ fixed, the first output of $r$ is a
bijection of $(b,w)$ because $\mathsf C$ is invertible and the bit $w$ is only
translated. With $(b,w)$ fixed, the same argument applies to the second
output as a function of $(a,z)$. Hence $r$ is non-degenerate. Applying it
twice returns the vector labels $a,b$; the two returned bits are
\[
 \begin{aligned}
 z+\beta(b,\mathsf C a)+\beta(\mathsf C b,\mathsf C ^2a)&=z,\\
 w+\beta(a,b)+\beta(\mathsf C ^2a,\mathsf C ^2b)&=w.
 \end{aligned}
\]
We used $\mathsf C $-invariance of $\beta$ in each line, so $r^2=\id$. For an input
$((a,z),(b,w),(c,u))$, both braid words have vector components
$(\mathsf C ^2c,b,\mathsf C a)$ and bit components
\begin{equation}
 \begin{aligned}
 &u+\beta(a,\mathsf C c)+\beta(b,c),\\
 &w+\beta(a,b)+\beta(c,\mathsf C b),\\
 &z+\beta(b,\mathsf C a)+\beta(c,a).
 \end{aligned}
 \label{eq:nonmonomial-map-braid}
\end{equation}
These expressions follow by substituting
Eq.~\eqref{eq:nonmonomial-six-state-map} three times, using
$\mathsf C^3=\id$, $\mathsf C+\mathsf C^2=\id$, and the bilinearity and
$\mathsf C$-invariance of $\beta$.
They prove the braid relation for $r$.
\end{proof}

\begin{lemma}[Unitary rotations]
\label{lem:nonmonomial-unitary-rotations}
The matrix $\Omega$ is unitary and satisfies $\Omega^3=\id_2$.
Every matrix $T_\xi$, with $\xi\in X_6$, is unitary.
\end{lemma}
\begin{proof}
The matrices $\mathsf I,\mathsf J,\mathsf K$ are
anti-Hermitian, square to $-\id_2$, and anticommute in pairs. Writing
$S=\mathsf I+\mathsf J+\mathsf K$ gives $S^2=-3\id_2$, and therefore
\begin{equation}
 \Omega^\dagger\Omega=\frac14(\id_2-S)(\id_2+S)=\id_2,\qquad
 \Omega^2=\frac12(-\id_2+S),\qquad
 \Omega^3=\id_2.
 \label{eq:nonmonomial-omega}
\end{equation}
Multiplying the quaternion units also gives
$\Omega\mathsf I\Omega^{-1}=\mathsf J$,
$\Omega\mathsf J\Omega^{-1}=\mathsf K$, and
$\Omega\mathsf K\Omega^{-1}=\mathsf I$. Since $\mathsf C $ cycles the three nonzero
vectors of $\mathbb F_2^2$, these identities imply
\begin{equation}
 \Omega q_a\Omega^{-1}=q_{\mathsf C a},\qquad
 q_aq_b=(-1)^{\beta(a,b)}q_bq_a.
 \label{eq:nonmonomial-quaternion-identities}
\end{equation}
Each $q_a$ is a product of unitary quaternion units. Since $\Omega$ is
unitary, $T_{(a,z)}=(-1)^zq_a\Omega$ is unitary as well.
\end{proof}

\begin{lemma}[Compatibility of the rotations]
\label{lem:nonmonomial-compatible-rotations}
For $\xi,\eta\in X_6$ and $r(\xi,\eta)=(\xi',\eta')$, the rotations satisfy
\[
 T_\xi T_\eta=T_{\xi'}T_{\eta'}.
\]
\end{lemma}
\begin{proof}
Write $\xi=(a,z)$ and $\eta=(b,w)$. Using
Eqs.~\eqref{eq:nonmonomial-six-state-map}
and~\eqref{eq:nonmonomial-quaternion-identities}, we obtain
\begin{equation}
 \begin{aligned}
 T_\xi T_\eta&=(-1)^{z+w}q_aq_{\mathsf C b}\Omega^2,\\
 T_{\xi'}T_{\eta'}
   &=(-1)^{z+w+\beta(a,b)+\beta(b,\mathsf C a)}q_{\mathsf C b}q_a\Omega^2.
 \end{aligned}
 \label{eq:nonmonomial-compatibility-proof}
\end{equation}
The form $\beta$ is symmetric in characteristic two, and $\mathsf C +\mathsf C ^2=\id$ gives
$\beta(a,b)+\beta(b,\mathsf C a)=\beta(a,\mathsf C b)$. Commuting $q_a$ past $q_{\mathsf C b}$ in
the first line therefore proves Eq.~\eqref{eq:nonmonomial-compatibility}.
\end{proof}

The next lemma supplies the rotation data used later to exclude a
monomial local basis in Subsection~\ref{subsec:nonmonomial-obstruction}.

\begin{lemma}[Traces and noncommutativity]
\label{lem:nonmonomial-rotation-data}
For every $(a,z)\in X_6$, we have $\Tr T_{(a,z)}=(-1)^z$.
The matrices $T_{((1,0),0)}$ and $T_{((0,1),0)}$ do not commute.
\end{lemma}
\begin{proof}
The three matrices with $z=0$ are
\begin{equation}
 \begin{aligned}
 T_{((1,0),0)}&=\tfrac12(\id_2-\mathsf I+\mathsf J-\mathsf K),\\
 T_{((0,1),0)}&=\tfrac12(\id_2-\mathsf I-\mathsf J+\mathsf K),\\
 T_{((1,1),0)}&=\tfrac12(\id_2+\mathsf I-\mathsf J-\mathsf K).
 \end{aligned}
 \label{eq:nonmonomial-explicit-rotations}
\end{equation}
Since the quaternion units are traceless, and multiplying them gives
$[T_{((1,0),0)},T_{((0,1),0)}]=\mathsf J+\mathsf K$, we have
\begin{equation}
 \Tr T_{(a,z)}=(-1)^z,\qquad
 [T_{((1,0),0)},T_{((0,1),0)}]
 =\begin{pmatrix}-i&1\\-1&i\end{pmatrix}\ne0.
 \label{eq:nonmonomial-traces-commutator}
\end{equation}
These are the two claimed properties.
\end{proof}

We now assemble the gate conditions from these properties of the map and
the rotations.

\begin{lemma}[Ordinary unitarity and involutivity]
\label{lem:nonmonomial-unitary-involutive}
The gate $R_\star$ defined in Eq.~\eqref{eq:nonmonomial-gate} is unitary
and involutive:
\[
 R_\star^\dagger R_\star=R_\star^2=\id.
\]
\end{lemma}
\begin{proof}
On $V_6\otimes V_6$, the gate is the involutive permutation $r$ of
Lemma~\ref{lem:nonmonomial-six-state-map}. The mixed
sectors are exchanged by mutually inverse unitary maps, by
Lemma~\ref{lem:nonmonomial-unitary-rotations}, and on
$V_2\otimes V_2$ it is $-\Swap$. Thus
$R_\star^\dagger R_\star=R_\star^2=\id$.
\end{proof}

\begin{lemma}[Dual unitarity]
\label{lem:nonmonomial-dual-unitarity}
The gate $R_\star$ is dual-unitary:
\[
 (R_\star^\Gamma)^\dagger R_\star^\Gamma=\id.
\]
\end{lemma}
\begin{proof}
We group the rows and columns of $R_\star^\Gamma$
according to their sectors. Up to independent row and column
permutations, the reshuffling is the direct sum
\begin{equation}
 P_r^\Gamma
 \;\oplus\;\bigoplus_{\xi\in X_6}T_\xi
 \;\oplus\;\bigoplus_{\xi\in X_6}(T_\xi^{-1})^{\mathsf T}
 \;\oplus\;(-\Swap|_{V_2\otimes V_2}).
 \label{eq:nonmonomial-reshuffled-blocks}
\end{equation}
Here $P_r$ is the permutation matrix of $r$.
For example, the mixed matrix element
$\bra{c,\xi}R_\star\ket{\xi,b}=(T_\xi)_{cb}$ appears at row $(\xi,c)$
and column $(b,\xi)$ of the reshuffling. The reverse mixed sector gives
the transpose of $T_\xi^{-1}$. The first block is a permutation matrix by
non-degeneracy of $r$ (Lemma~\ref{lem:nonmonomial-six-state-map}).
The remaining blocks are unitary by
Lemma~\ref{lem:nonmonomial-unitary-rotations} and unitarity of the swap.
Since row and column permutations preserve unitarity, this proves
$(R_\star^\Gamma)^\dagger R_\star^\Gamma=\id$.
\end{proof}

\begin{lemma}[The braid relation]
\label{lem:nonmonomial-braid}
The gate $R_\star$ satisfies
\[
 (R_\star)_{12}(R_\star)_{23}(R_\star)_{12}
 =(R_\star)_{23}(R_\star)_{12}(R_\star)_{23}.
\]
\end{lemma}
\begin{proof}
We check the eight sectors of $V^{\otimes3}$. On $V_6^{\otimes3}$ the
claim follows from Lemma~\ref{lem:nonmonomial-six-state-map}; on
$V_2^{\otimes3}$ it is the braid
relation for $-\Swap$. For two labels $\xi,\eta\in X_6$ and one vector
$\ket{\psi}\in V_2$, write $r(\xi,\eta)=(\xi',\eta')$. In the three possible
sector orders, the two braid words have the same label outputs and
compare the following rotations of $\ket{\psi}$:
\begin{equation}
 \begin{array}{c|c|c}
 \text{input sector}&R_{12}R_{23}R_{12}&R_{23}R_{12}R_{23}\\ \hline
 V_6\otimes V_6\otimes V_2&T_{\xi'}T_{\eta'}&T_\xi T_\eta\\
 V_6\otimes V_2\otimes V_6&T_{\xi'}^{-1}T_\xi&T_{\eta'}T_\eta^{-1}\\
 V_2\otimes V_6\otimes V_6&T_\eta^{-1}T_\xi^{-1}&T_{\eta'}^{-1}T_{\xi'}^{-1}
 \end{array}
 \label{eq:nonmonomial-two-label-braid}
\end{equation}
Here $R$ abbreviates $R_\star$ only in the table headings. Each row follows
from Eq.~\eqref{eq:nonmonomial-gate}, and equality of its two entries is
the compatibility identity of
Lemma~\ref{lem:nonmonomial-compatible-rotations} or a rearrangement of it.
With one label and two vectors $\ket{\psi},\ket{\phi}\in V_2$, both braid words give
\begin{equation}
 \begin{aligned}
 \ket{\xi}\otimes\ket{\psi}\otimes\ket{\phi}
   &\longmapsto -T_\xi\ket{\phi}\otimes T_\xi\ket{\psi}\otimes\ket{\xi},\\
 \ket{\psi}\otimes\ket{\xi}\otimes\ket{\phi}
   &\longmapsto -T_\xi\ket{\phi}\otimes\ket{\xi}\otimes T_\xi^{-1}\ket{\psi},\\
 \ket{\psi}\otimes\ket{\phi}\otimes\ket{\xi}
   &\longmapsto -\ket{\xi}\otimes T_\xi^{-1}\ket{\phi}\otimes T_\xi^{-1}\ket{\psi}.
 \end{aligned}
 \label{eq:nonmonomial-one-label-braid}
\end{equation}
The minus sign is the single crossing of the two $V_2$ states. These cases
exhaust the sector decomposition and prove the braid relation.
\end{proof}

Lemmas~\ref{lem:nonmonomial-unitary-involutive}--\ref{lem:nonmonomial-braid}
establish ordinary unitarity and all identities in
Eq.~\eqref{eq:nonmonomial-gate-conditions}. The absence of a monomial local
form will be proved in Subsection~\ref{subsec:nonmonomial-obstruction} below.

\subsection{Excluding every monomial local basis}
\label{subsec:nonmonomial-obstruction}

To exclude a monomial local form, we first need to constrain the allowed
local bases. We do this using the partial traces. In
Lemma~\ref{lem:nonmonomial-partial-traces} below, we extend the argument of
Proposition~\ref{prop:yb-involutive-LPW-class} to gates with phases, using
only involutivity and dual unitarity. The partial traces computed below in
Eq.~\eqref{eq:nonmonomial-partial-traces} will then force any such bases to
respect the decomposition $V=V_6\oplus V_2$. We do not require the braid
relation, since different basis changes on the two sites need not preserve it.

\begin{lemma}[Partial traces in a monomial basis]
\label{lem:nonmonomial-partial-traces}
Both partial traces of an involutive dual-unitary monomial gate are
diagonal in its local product bases, and every diagonal entry is $+1$ or
$-1$.
\end{lemma}
\begin{proof}
Write $M\ket{x,y}=\nu(x,y)\ket{s_1(x,y),s_2(x,y)}$, with
$|\nu(x,y)|=1$. Dual unitarity makes $y\mapsto s_1(x,y)$ and
$x\mapsto s_2(x,y)$ bijective for each fixed value of the other input.
Involutivity gives $s^2=\id$. The counting argument in the proof of
Proposition~\ref{prop:yb-involutive-LPW-class} uses only these properties
of $s$, not the braid relation. It shows that no pair contributes to an
off-diagonal entry of $\Tr_2M$, while exactly one fixed pair contributes
to each diagonal entry.

At such a fixed pair $(x,y)$, involutivity of $M$ gives
$\nu(x,y)^2=1$. Its contribution is therefore $+1$ or $-1$, so
$\Tr_2M$ has the stated form. Exchanging the two sites gives the same
conclusion for $\Tr_1M$ and completes the proof.
\end{proof}

For the gate in Eq.~\eqref{eq:nonmonomial-gate},
Lemma~\ref{lem:nonmonomial-partial-traces} applied to the bare permutation
on $V_6\otimes V_6$ gives the identity on $V_6$. The mixed sectors have
zero partial trace because they exchange orthogonal sectors, while
$-\Swap$ has partial trace $-\id_{V_2}$. Therefore
\begin{equation}
 Q:=\Tr_1R_\star=\Tr_2R_\star=\id_{V_6}\oplus(-\id_{V_2}).
 \label{eq:nonmonomial-partial-traces}
\end{equation}
The two eigenvalues distinguish the six-dimensional sector from the
two-dimensional one. Any monomializing basis must respect this distinction.

\begin{proof}[Completion of the proof of Theorem~\ref{thm:nonmonomial-local-form}]
Suppose that the matrix $M$ in Eq.~\eqref{eq:nonmonomial-local-conjugate}
is monomial. Local unitary conjugation preserves ordinary unitarity and
involutivity. It also preserves dual unitarity: in our index convention,
\begin{equation}
 M^\Gamma=(V_1^{\mathsf T}\otimes V_1^\dagger)
 R_\star^\Gamma(V_2\otimes\overline{V_2}),
 \label{eq:nonmonomial-local-reshuffle}
\end{equation}
whose two outer factors are unitary.
Lemma~\ref{lem:nonmonomial-partial-traces} therefore applies to $M$, even
though we do not assume that $M$ obeys the braid relation. Since
\begin{equation}
 \Tr_2M=V_1^\dagger QV_1,\qquad
 \Tr_1M=V_2^\dagger QV_2,
 \label{eq:nonmonomial-conjugated-traces}
\end{equation}
both local bases consist of eigenvectors of $Q$. Each is thus the union
of a basis of $V_6$ and a basis of $V_2$, up to ordering.

Let $\{\ket{w_i}\}$ and $\{\ket{v_j}\}$ be the two bases on $V_2$. The restriction
$R_\star=-\Swap$ sends $\ket{w_i}\otimes\ket{v_j}$ to
$-\ket{v_j}\otimes\ket{w_i}$. For this
to be a single product-basis vector, every $\ket{v_j}$ must be proportional to
some $\ket{w_k}$, and conversely. Hence the two $V_2$ bases coincide up to
phases and permutation. Rephasing and reordering preserve monomiality,
so we may use the same basis $\ket{w_1},\ket{w_2}$ on both sites.

The two bases on $V_6$ may still differ. Write them as
\[
 \ket{f_\alpha}=\sum_{\xi\in X_6}c_{\xi\alpha}\ket{\xi},\qquad
 \ket{g_\gamma}=\sum_{\xi\in X_6}d_{\xi\gamma}\ket{\xi}.
\]
For every input $\ket{f_\alpha}\otimes\ket{w_\beta}$,
Eq.~\eqref{eq:nonmonomial-gate} and monomiality give
\begin{equation}
 \sum_{\xi\in X_6}c_{\xi\alpha}T_\xi\ket{w_\beta}\otimes\ket{\xi}
 =e^{i\theta}\ket{w_{\beta'}}\otimes\ket{g_\gamma}
 \label{eq:nonmonomial-basis-obstruction}
\end{equation}
for some $\beta',\gamma,\theta$ depending on the input. For each $\xi$,
choose $\alpha$ with $c_{\xi\alpha}\ne0$; such an $\alpha$ exists because
the $\ket{f_\alpha}$ form a basis. Comparing the $\ket{\xi}$ coefficients shows,
for each $\beta$, that $T_\xi\ket{w_\beta}$ is proportional to a single
$\ket{w_{\beta'}}$. Unitarity of $T_\xi$ makes it a monomial matrix in this
common two-dimensional basis.

A two-dimensional monomial matrix is either diagonal or antidiagonal.
Equation~\eqref{eq:nonmonomial-traces-commutator} gives nonzero trace for
all six $T_\xi$, excluding the antidiagonal possibility. All six matrices
would therefore be diagonal in the same basis and would commute. The
nonzero commutator in Eq.~\eqref{eq:nonmonomial-traces-commutator} is a
contradiction. This excludes every pair $V_1,V_2$ and completes the proof.
\end{proof}

}

\section*{Conflict of Interest}

The author declares that there is no conflict of interest.

\section*{Data Availability}

No datasets were generated or analysed during this study.


\providecommand{\href}[2]{#2}\begingroup\raggedright\endgroup

\clearpage

\begin{figure}[p]
\centering
\begin{tikzpicture}[
  line cap=round,
  line join=round,
  section/.style={
    ellipse,
    draw=black!65,
    fill=black!3,
    line width=0.85pt,
    align=center,
    font=\footnotesize,
    text width=33mm,
    inner xsep=2pt,
    inner ysep=3pt,
    minimum height=13mm,
    execute at begin node={
      \hyphenpenalty=10000
      \exhyphenpenalty=10000
      \relax
    }
  },
  input section/.style={
    section,
    draw=opentorange,
    fill=opentlightorange
  },
  result section/.style={
    section,
    draw=black!65,
    fill=black!3
  },
  dependency/.style={
    -{Latex[length=2.2mm,width=1.55mm]},
    draw=black!68,
    line width=0.75pt
  },
  merge/.style={
    diamond,
    aspect=1.45,
    draw=black!55,
    fill=white,
    line width=0.7pt,
    inner sep=1.4pt,
    font=\scriptsize\bfseries,
    align=center
  }
]


\node[section] (s2) at (0,0)
  {\textbf{Sec.~\ref{sec:problem-results}}\\
   The setting and main results};


\node[section] (s3) at (-4.10,-2.55)
  {\textbf{Sec.~\ref{sec:rectangular-reduction}}\\
   Rearranging the circuit};

\node[input section] (s7) at (4.10,-2.55)
  {\textbf{Sec.~\ref{sec:involutive-unitary}}\\
   Involutive Yang--Baxter gates};


\node[input section] (s10) at (0,-4.05)
  {\textbf{Sec.~\ref{sec:yb-maps}}\\
   Non-degenerate Yang--Baxter maps\\
   and phase dressings};


\node[section] (s5) at (-4.35,-6.50)
  {\textbf{Sec.~\ref{sec:elementary-bounded-rank}}\\
   Elementary bounded-rank\\
   constructions};

\node[section] (s4) at (1.25,-6.50)
  {\textbf{Sec.~\ref{sec:centralizer-output-space}}\\
   Braid invariance of\\
   the output space};

\node[merge] (both) at (4.85,-6.50) {both};


\node[input section] (s6) at (-4.35,-9.60)
  {\textbf{Sec.~\ref{sec:qubit-classification}}\\
   Yang--Baxter gates on qubits};

\node[result section] (s8) at (1.25,-9.60)
  {\textbf{Sec.~\ref{sec:inv-polynomial-proof}}\\
   Polynomial operator Schmidt rank\\
   for involutive gates with\\
   dual unitarity};

\node[result section] (s9) at (4.85,-12.70)
  {\textbf{Sec.~\ref{sec:normal-form-dynamics}}\\
   Operator entanglement of\\
   LPW normal-form circuits};


\node[result section] (s11) at (-4.35,-12.70)
  {\textbf{Sec.~\ref{sec:braid-only}}\\
   Exponential growth\\
   of the OSR from a\\
   Yang--Baxter gate without\\
   dual unitarity};


\draw[dependency]
  (s2.south west) to[bend right=8] (s3.north);

\draw[dependency]
  (s2.south east) to[bend left=8] (s7.north);

\draw[dependency]
  (s3.south) -- (s5.north);

\draw[dependency]
  (s3.south east) -- (s10.north west);

\draw[dependency]
  (s3.south)
  .. controls ++(0,-1.55) and ++(-1.35,0.75) ..
  (s4.north west);

\draw[dependency,rounded corners=8pt]
  (s3.west) -- ++(-0.75,0) -- ++(0,-10.15) -- (s11.west);

\draw[dependency]
  (s5.south) -- (s6.north);

\draw[dependency]
  (s4.east) -- (both.west);

\draw[dependency]
  (s7.south) to[out=-90,in=65] (both.north east);

\draw[dependency]
  (both.south west) to[out=-115,in=55] (s8.45);

\draw[dependency]
  (both.south) -- (s9.north);

\end{tikzpicture}

\caption{Logical reading paths through the paper. A solid arrow indicates
that the source section supplies definitions, results, or useful context for
the target section. The yellowish nodes indicate sections that import crucial
external input into the paper: the LPW classification in
Section~\ref{sec:involutive-unitary}, the guitar-map technology in
Section~\ref{sec:yb-maps}, and the classification of qubit Yang--Baxter gates
in Section~\ref{sec:qubit-classification}.
{
Appendix~\ref{app:nonmonomial-involutive} separates the gate families of
Sections~\ref{sec:inv-polynomial-proof} and~\ref{sec:yb-maps}; it is not an
input to their rank proofs.
}
}
\label{fig:section-dependencies}
\end{figure}
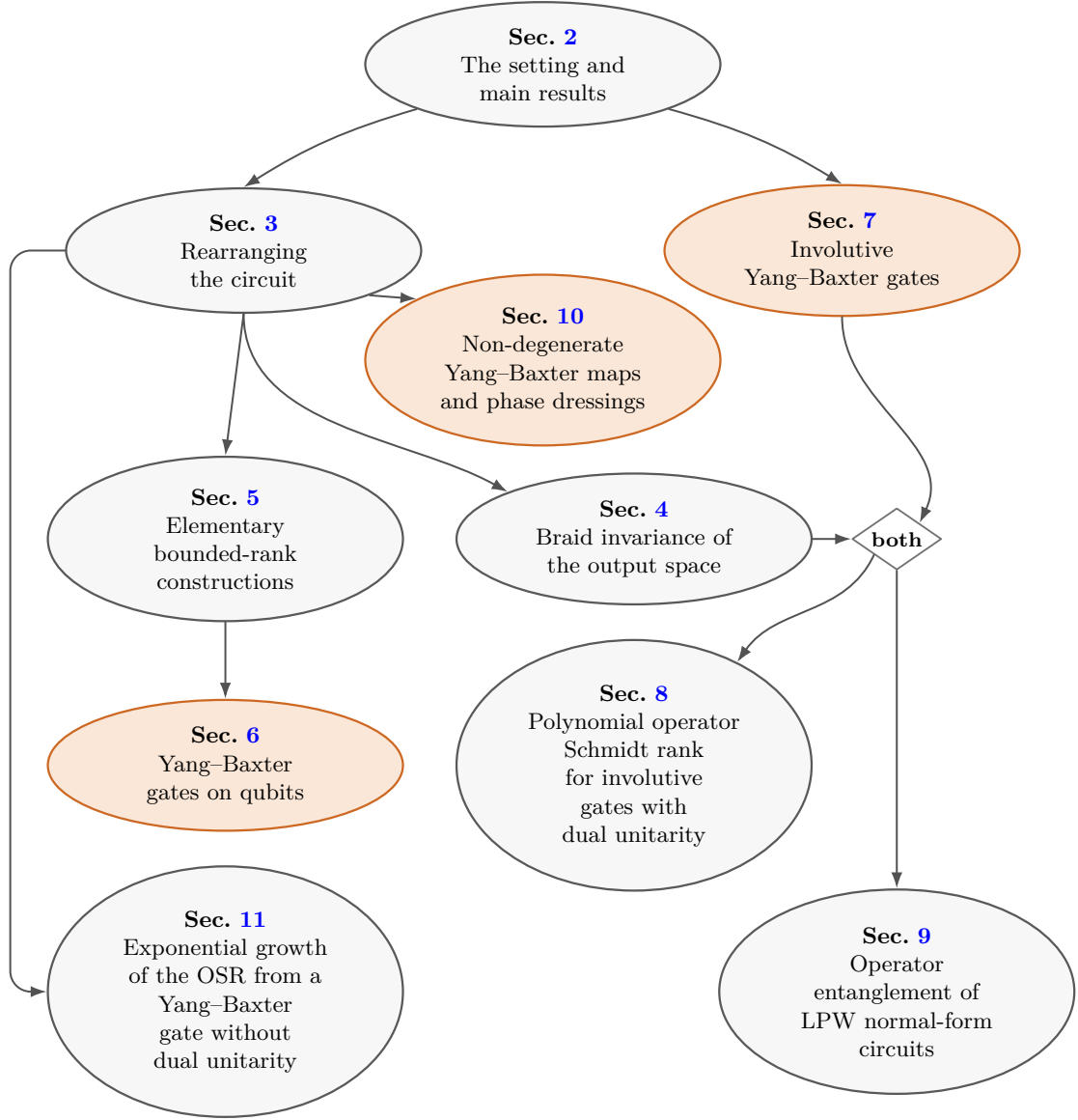

\end{document}